\documentclass[11pt]{article}

\usepackage[T1]{fontenc}        % modern font encoding
\usepackage{amsmath, amssymb, amsfonts, amsthm, mathtools, cancel}

\usepackage{graphicx}
\usepackage{tikz}
\usetikzlibrary{arrows.meta, positioning}
\usepackage{caption}
\usepackage{subcaption}         % use this, drop subfig (they conflict)

\usepackage{booktabs}
\usepackage{arydshln}
\usepackage[table,x11names]{xcolor}

\usepackage{geometry}
\usepackage{fancyhdr}
\usepackage{authblk}
\usepackage{relsize}
\usepackage{nicefrac}
\usepackage{microtype}
\usepackage{float}
\usepackage{lipsum}             % for dummy text (remove if not needed)
\usepackage{ulem}               % underline/strikeout (remove if not needed)

\usepackage{url}
\usepackage{hyperref}

\newtheorem{theorem}{Theorem}[section]
\newtheorem{lemma}[theorem]{Lemma}

\theoremstyle{definition}

\theoremstyle{remark}
\newtheorem{remark}[theorem]{Remark}

\numberwithin{equation}{section}
\numberwithin{figure}{section}
\numberwithin{table}{section}

\DeclarePairedDelimiter{\norm}{\lVert}{\rVert}
\newbox{\bigpicturebox}

\begin{document}
\title{\bf On the Geographic Spread of Chikungunya between Brazil and Florida: A Multi-patch Model with Time Delay\thanks{Research was partially supported by the National Science Foundation  (DMS-2424605). }} 

\author{\sc Antonio Gondim$^{a}$, Xi Huo$^{a}$, Jaqueline Mesquita$^{b}$ and Shigui Ruan$^{a}$}

\affil[]{$^{a}$Department of Mathematics, University of Miami, Coral Gables, FL 33146, USA\\
$^{b}$Departamento de Matem\'atica, Universidade de Bras\'ilia, 
%Campus Universit\'ario Darcy Ribeiro, 
Bras\'ilia - DF - 70910-900, Brazil}

%\author[2]{\sc Shigui Ruan}

\maketitle
\begin{abstract}
Chikungunya (CHIK) is a viral disease transmitted to humans through the bites of {\it Aedes} mosquitoes infected with the chikungunya virus (CHIKV). CHIKV has been imported annually to Florida in the last decade due to Miami's crucial location as a hub for international travel, particularly from Central and South America including Brazil, where CHIK is endemic. This work addresses to the geographic spread of CHIK, incorporating factors such as human movement, temperature dependency, as well as vertical transmission, and incubation periods, for different patches. Central to the model is the integration of a multi-patch framework, and in the numerical analysis it is considered human movement between endemic Brazilian states and Florida. We establish crucial correlations between the mosquito reproduction number $\mathcal{R}_{m}$ and the disease reproduction number $\mathcal{R}_{0}$ with the disease dynamics in a multi-patch environment, encompassing not only a numerical analysis but also from a theoretical perspective. Through numerical simulations, validated with real population and temperature data, it is possible to understand the disease dynamics under many different scenarios and make future projections, offering insights for potential effective control strategies, as well as addressing the timing for these strategies to be adopted.


\end{abstract}

\section{Introduction}

Chikungunya (CHIK), caused by chikungunya virus (CHIKV),  is a vector-borne disease transmitted by the {\it Aedes} mosquitoes ({\it Aedes aegypti} and {\it Aedes albopictus}), the same vector responsible for the transmission of dengue and  Zika. CHIK poses a significant public health threat, particularly in tropical and subtropical regions, where it recurrently manifests, posing risks to individuals of all ages. Although these diseases share similar symptoms such as fever, headache, joint pain, nausea, and rash, certain key differences set them apart. The main symptom of CHIK is severe joint pain (polyarthralgia), which can affect multiple joints but is particularly intense in the hands and feet, including the fingers, ankles, and wrists. This pain comes from joint inflammation and is often accompanied by swelling and stiffness. While dengue and Zika can also cause joint pain, the severity differs. In dengue and Zika, joint pain tends to be mild to moderate, whereas in CHIK, it is typically intense and debilitating, significantly impacting productivity and quality of life. Despite a low mortality rate, in its subacute or chronic phase, the related joint pain in CHIK can persist for months or even years, especially in older patients. For more details, see \cite{fiocruz-1}. The transmission dynamics of CHIK are complex, with an intrinsic incubation period in humans ranging from 1 to 12 days and an extrinsic incubation period in mosquitoes about 10 days post mosquito bite. With no effective vaccine, symptomatic treatment remains the primary approach.

%\textcolor{red}{In the present paper, our main goal is to investigate CHIK models, since} chikungunya virus (CHIKV) poses a significant public health threat, particularly in tropical and subtropical regions, where it recurrently manifests, posing risks to individuals of all ages. Its transmission dynamics \textcolor{red}{\sout{, facilitated by the {\it Aedes aegypti} mosquito vector,}} are complex, with an intrinsic incubation period in humans ranging from 1 to 12 days post mosquito bite. \textcolor{red}{\sout{The acute phase typically manifests with high fever and polyarthralgia, while atypical presentations add to the diagnostic challenges. }\sout{Despite a low mortality rate, the lingering joint pain significantly impairs quality of life for months to years.}} With no effective vaccine, symptomatic treatment remains the primary approach.

CHIKV was first isolated from human serum and mosquitoes in an epidemic in Tanzania in 1952-1953 \cite{Lumsden1955}. In 2004, an outbreak originating on the coast of Kenya subsequently spread to Comoros, La R\'eunion, several other Indian Ocean islands, and India in the following two years. Viremic travelers then spread outbreaks from India to the Andaman and Nicobar Islands, Sri Lanka, the Maldives, Singapore, Malaysia, Indonesia \cite{WHOchik}.  The first evidence of autochthonous CHIK transmission in the Americas was recorded in St. Martin Island in December 2013, consequently autochthonous transmission was detected in more than 40 countries and territories of the Americas and the Pan American Health Organization (PAHO) reported a total of 1,071,696 suspected cases (including 169 deaths) in 2014. Currently, CHIK is endemic in the Americas \cite{PAHOchikungunya}. 

%Over the last decade, CHIKV has exhibited unprecedented global spread, affecting regions across Africa, Asia, Europe, the Americas, and the Indian and Pacific Oceans.  Notably, the 2013 outbreak in St. Martin Island marked the beginning of its dissemination throughout the Americas, with subsequent autochthonous transmission reported in forty-five countries. The threat magnitude is substantial \cite{Monteiro2020}, prompting its classification as a notifiable disease by the Centers for Disease Control and Prevention. 

%Recent Data on CHIKV in Brazil and Florida.

Brazil remains heavily impacted by CHIKV, with the virus maintaining its endemic status in multiple regions \cite{Ferreira2023}. PAHO reported a total of 423,898 CHIK cases in Brazil in 2024. Genetic analysis has shown that the predominant strain in Brazil belongs to the East/Central/South African (ECSA) lineage, which demonstrates high adaptability to local vectors, including {\it Aedes albopictus}. This enhanced adaptability raises concerns about the virus's epidemic potential in less urbanized regions, expanding its reach beyond traditional hotspots \cite{Costa-da-Silva2017}.

Beginning in 2014, CHIK cases were reported among U.S. travelers returning from affected areas in the Americas. While CHIKV transmission is not endemic in Florida, the state has experienced multiple travel-associated cases and Florida Department of Health reported 12 locally-transmitted cases in 2014 \cite{fdh}. Environmental conditions, including increased rainfall and elevated summer temperatures, have bolstered mosquito populations, enhancing the potential for sustained local transmission. Studies reveal the presence of CHIKV in vertical mosquito transmission cycles, further emphasizing the risk of autochthonous outbreaks. Monitoring data highlight the importance of vector control strategies to limit disease spread \cite{Honorio2018}.

%Factors such as global warming, reduced herd immunity, substandard living conditions, and increased human mobility have fueled the rapid expansion of CHIKV, particularly in the Americas \cite{12}. Rising temperatures and altered rainfall patterns directly influence {\it Aedes} mosquito dynamics, by accelerating their development rates \cite{couret2014temperature} and reproduction and feeding behaviors \cite{13, 14}. 
%We will be using \cite{couret2014temperature} also to approximate the mean development stage time in terms of the temperature to a periodic function, which will be defined as the delay for the development stage $\tau_{l}(t),$ that will be capturing the seasonality effect on the mean time of the development stage of the larvae. \textcolor{red}{Incorporating this dynamic delay into the model enhances its robustness and accuracy, allowing for a more precise representation of temperature-driven effects on the mosquito life cycle.}

Factors such as global warming, reduced herd immunity, substandard living conditions, and increased human mobility have fueled the rapid expansion of CHIKV, particularly in the Americas \cite{12}. Rising temperatures and altered rainfall patterns directly influence {\it Aedes} mosquito dynamics, accelerating development rates \cite{couret2014temperature} and modifying reproduction and feeding behaviors \cite{13, 14}. This interplay has been especially evident in both Brazil and Florida, where localized data indicate increased mosquito longevity and bite frequency. The combination of climatic and environmental changes underscores the necessity for integrated vector management and disease surveillance systems to curb future outbreaks 

Mathematical modeling has played a crucial role in understanding general disease dynamics \cite{Schultz}, and particularly vector-borne disease dynamics where
various models have been developed to explore transmission dynamics and control strategies, we refer to the reviews \cite{Aguiar2022, Ganesan2017, Ogunlade2023, RL2025} and the references cited therein. Importation and transmission dynamics of vector-borne diseases in Florida have been investigated via mathematical analysis \cite{chen2018zika, chen2016westnile, Lord2020}. Since an outbreak of CHIK originating on the coast of Kenya subsequently spread to La R\'eunion Island in 2004, mathematical models were first proposed to understand the transmission dynamics of CHIK and the reported data of CHIK in La R\'eunion Island, such as \cite{Dumont2008, Dumont2010, Moulay2011, Yakob2013}. Since then, various models have been constructed to describe several features of CHIK in different regions, see for example \cite{D-R2020} for CHIK epidemics (2016, 2018, 2019) in the city of Rio de Janeiro, \cite{Feng2019} for the CHIK outbreak in Italy, and so on. The effects of climate change, temperature, rainfall, extrinsic incubation period, and maturation delay have also been included in CHIK models \cite{Chadsuthi2016, Liu2020}. Metapopulation models including populations mobility on a large-scale network have also been studied \cite{Moulay2013}.

Given the annual importation of CHIKV to Miami from endemic regions as Latin America, this paper aims to present a {multi-patch model} incorporating some epidemiological and environmental factors— {human movement, temperature, vertical transmission, and incubation period }— to analyze CHIKV's geographic spread.  We propose metapopulation models in a periodical environment and investigate the effect of migrations of infectious hosts from one patch to another, potentially introducing a disease outbreak in a disease free patch. We will apply the results in \cite{couret2014temperature} to approximate the mean development stage time in terms of the temperature to a periodic function, which will be defined as the delay for the development stage $\tau_{l}(t),$ that will be capturing the seasonality effect on the mean time of the development stage of the larvae. More precisely, in this paper we present a periodic time delay differential model accounting for seasonal variability and periodic time-dependent delays, incorporating temperature impacts on CHIKV transmission dynamics, which plays a crucial role on the mosquitoes dynamics, specially on determining the conditions under which the disease will either remain endemic or die out. Furthermore, we estimate numerically the vertical transmission probability \cite{verticaltrans, Heath2020}, as well as the influence of various control strategies on both the basic reproduction number $\mathcal{R}_{0}$ and the infectious population. Finally, we project cumulative infection counts under different scenarios, considering alternative control strategies in Miami-Dade through 2030, and forecast the number of infectious individuals in Brazil through 2025. 

\section{The Multi-patch Model with Delay}
\setcounter{equation}{0}\setcounter{figure}{0}

Within the continuous model, humans, mosquitoes, and mosquito larvae are considered. Humans serve as hosts for chikungunya infection, acquiring the virus through bites from infectious mosquitoes and spreading the virus via the bites by uninfected mosquitoes. Acting as vectors, mosquitoes facilitate disease transmission to humans and may also infect larvae, which subsequently develop into infectious mosquitoes. In this study, we utilize temperature and population data from various regions, aggregating them into a single model with $n$ patches. Each patch represents a distinct sub-population within the broader endemic context. The deterministic model incorporates branches for humans, mosquitoes, and mosquito larvae, each replicate across the $n$ patches to capture the multi-patch framework’s complexity. 

Let $S_{h}^i(t), E_{h}^i(t), I_{h}^i(t), A_{h}^i(t), R_{h}^i(t)$ represent the number of susceptible, exposed, infected, asymptomatic, and recovered humans in patch $i$ at time $t$. Similarly, let $S_{m}^i(t), E_{m}^i(t), I_{m}^i(t)$ denote the number of susceptible, exposed, and infected mosquitoes, and $L_{s}^i(t), L_{I}^i(t)$ denote the number of susceptible and infected larvae in patch $i$ at time $t$, where $i = 1, \dots , n.$ The parameters along with their biological meanings are listed in Table \ref{Table1}.  All time dependent parameters and delays are $w$-periodic positive functions for some $w>0$. 
The schematic diagram and interactions among the subpopulations in patches are depicted in Figure \ref{fig:chikungunya_model_patch}.
\begin{figure}[ht]
\centering
\begin{tikzpicture}[scale=0.8,
compartment/.style={rectangle, draw, minimum width=1.6cm, minimum height=0.6cm, node distance=1.6cm, font=\tiny},
arrow/.style={-Stealth, thin},
crosspatch/.style={arrow, dashed, color=blue},
curved/.style={arrow, bend right=45},
every node/.style={font=\tiny}
]
    % Human compartments (top group)
\node[compartment, fill=red!20] at (0,4) (Sh) {$S_h^i$};
\node[compartment, fill=red!20] at (4,7) (Ih) {$I_h^i$};
\node[compartment, fill=red!20] at (8,4) (Ah) {$A_h^i$};
\node[compartment, fill=red!20] at (4,4) (Eh) {$E_h^i$};
\node[compartment, fill=red!20] at (8,7) (Rh) {$R_h^i$};
    % Mosquito compartments 
\node[compartment, fill=blue!20] at (-2,0) (Sm) {$S_m^i$};
\node[compartment, fill=blue!20] at (4,0) (Em) {$E_m^i$};
\node[compartment, fill=blue!20] at (8,0) (Im) {$I_m^i$};
    % Larvae compartments
\node[compartment, fill=green!20] at (2,-3) (Ls) {$L_s^i$};
\node[compartment, fill=green!20] at (6,-3) (Li) {$L_I^i$};
    % Human flows
\draw[arrow] (Eh) -- node[right] {$\tau_h$} (Ih);
\draw[arrow] (Eh) -- node[above] {$\tau_h$} (Ah);
\draw[arrow] (Ah) -- node[right] {$\eta_h$} (Rh);
\draw[arrow] (Sh) -- node[above] {$b_m^i(t)\beta_{mh}I_m^i$} (Eh);
    % Mosquito flows
\draw[arrow] (Sm) -- node[above] {$\tiny b_m^i(t)\beta_{hm}(I_h^i + A_h^i)$} (Em);
\draw[arrow] (Em) -- node[above] {$\tau^i(t)$} (Im);
    % Interaction between groups
\draw[dashed] (Ih) -- (Sm);
\draw[dashed] (Ah) -- (Sm);
\draw[dashed] (Im) -- (Sh);
    % Larvae flows
\draw[curved] (Sm) to node[left] {$o^i_2(t)$} (Ls);
\draw[arrow] (Ls) to node[right] {$\tau_l^i(t)$} (Sm);
\draw[curved] (Im) to node[right] {$o_4^i(t)$} (Li);
\draw[arrow] (Im) to node[right, below] {$o_3^i(t)$} (Ls);
\draw[arrow] (Li) to node[right] {$\tau_l^i(t)$} (Im);
\draw[arrow] (Em) to node[left] {$o_1^i(t)$} (Ls);
\draw[arrow] (Ih) to node[above] {$\eta_h$} (Rh);
%\draw[arrow] (Eh) to node[above, sloped] {$\eta_h$} (Rh);
    % Cross-patch flows (simplified labels)
\draw[crosspatch, <-] (Sh) to ++(-2,1) node[left] {$R^i + \sum_{j \neq i} m_{ji} S_h^j$};
\draw[crosspatch, <-] (Ah) to ++(2,1) node[right] {$\sum_{j\neq i} m_{ji} A_h^j$};
\draw[crosspatch, ->] (Sh) to ++(-2,0) node[left] {$\sum_{j\neq i} m_{ij} $};
\draw[crosspatch, ->] (Ah) to ++(2,0) node[right] {$\sum_{j\neq i} m_{ij}A_h^i $};
\draw[crosspatch, ->] (Rh) to ++(2,0) node[right] {$\sum_{j\neq i} m_{ij}R^i_h $};
\draw[crosspatch, <-] (Rh) to ++(2,1) node[right] {$\sum_{j\neq i} m_{ji} R_h^j$};
\draw[crosspatch, <-] (Eh) to ++(2,1) node[right] {$\sum_{j\neq i} m_{ji} E_h^j$};
\draw[crosspatch, ->] (Eh) to ++(2,-1) node[right] {$\sum_{j\neq i} m_{ij}E^i_h $};  
    % Death rates
\draw[arrow] (Sh) -- ++(-1,-1) node[left] {$\mu_h$};
\draw[arrow] (Eh) -- ++(-1,-1) node[left] {$\mu_h$};
\draw[arrow] (Ih) -- ++(0,1) node[right] {$\mu_h$};
\draw[arrow] (Ah) -- ++(1,-1) node[right] {$\mu_h$};
\draw[arrow] (Rh) -- ++(-1,1) node[left] {$\mu_h$};
\draw[arrow] (Sm) -- ++(-1,-1) node[left] {$\mu_m$};
\draw[arrow] (Em) -- ++(0,-1) node[below] {$\mu_m$};
\draw[arrow] (Im) -- ++(1,-1) node[right] {$\mu_m$};
\draw[arrow] (Ls) -- ++(-1,-1) node[left] {$\mu_l$};
\draw[arrow] (Li) -- ++(1,-1) node[right] {$\mu_l$};
\end{tikzpicture}
\caption{Schematic diagram of  the multi-patch model (\ref{0}). Arrows represent transitions between compartments, with dashed arrows representing interactions and blue dashed arrows indicating cross-patch flows and the recruitment and death rates for \( S_h^i, A_h^i \). Curved arrows depict larvae maturation to susceptible or infected mosquitoes represented by the delay \( \tau_l^i(t) \), and larvae birth from susceptible or infected mosquitoes.}
\label{fig:chikungunya_model_patch}
\end{figure}
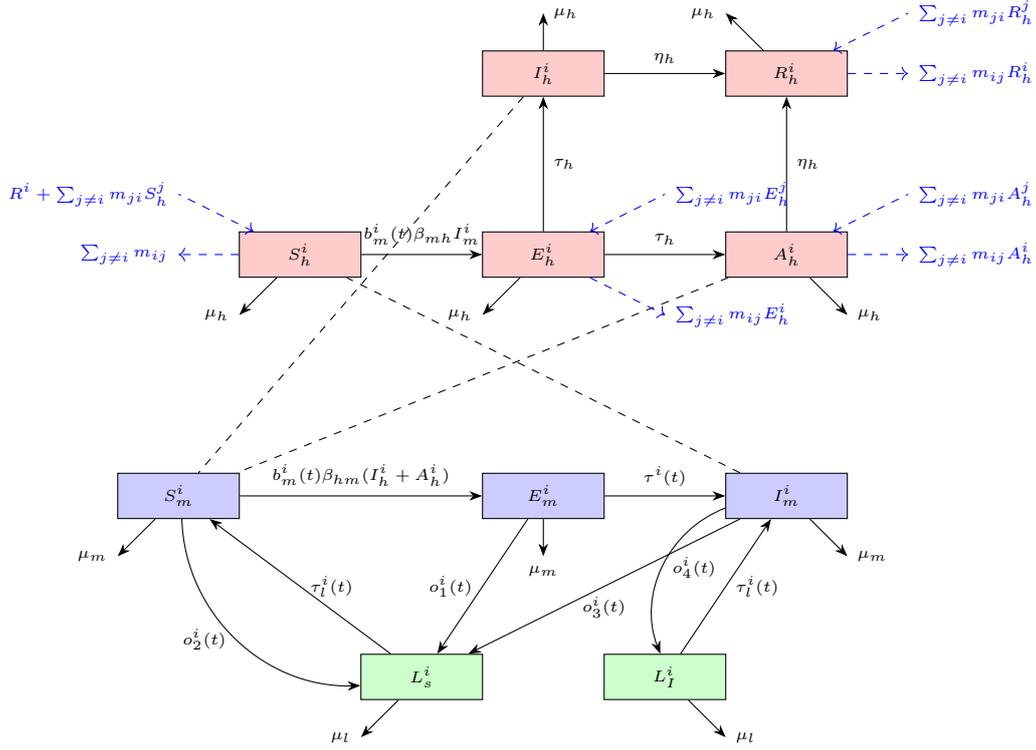
 
 {\small
\begin{table}[h!]
\centering
\caption{Parameters used in the model}
\begin{tabular}{|c|c|}
\hline
%$S_{h}^i$ & Susceptible humans in the patch $i$ \\
%$E_{h}^i$ & Exposed humans in the patch $i$\\
%$I_{h}^i$ & Infected humans in the patch $i$\\
%$A_{h}^i$ & Asymptomatic humans in the patch $i$\\
%$R_{h}^i$ & Recovered humans in the patch $i$\\
%$S_{m}^i$ & Susceptible mosquitoes in the patch $i$\\
%$E_{m}^i$ & Exposed mosquitoes in the patch $i$\\
%$I_{m}^i$ & Infected mosquitoes in the patch $i$\\
%$L_{s}^i$ & Susceptible larvae in the patch $i$\\
%$L_{I}^i$ & Infected larvae in the patch $i$\\
$R^i$ & Human population recruitment rate in patch $i$\\
$\mu_{h}$ & Human natural death rate \\
$\mu^i_{m}(t)$ & Mosquito natural death rate in patch $i$\\
$b_{m}^i(t)$ & Per capita mosquito biting rate in patch $i$\\
$\beta_{mh}$ & Human-mosquito transmission probability \\
$\beta_{hm}$ & Mosquito-human transmission probability \\
%$\eta_{h}$ & Rate of recovery for symptomatic humans \\
$\tau^i(t)$ & The extrinsic incubation period in mosquitoes in patch $i$\\
    $\tau^i_{l}(t)$ & The mean development stage time in larvae in patch $i$\\% https://journals.plos.org/plosone/article?id=10.1371/journal.pone.0087468
$\tau_{h}$ & The intrinsic incubation period in humans\\
$\mu_{b}^i(t)$ & Larvae birth rate in patch $i$\\
$a$ & Proportion of asymptomatic infectious humans \\
$p$ & Vertical transmission probability \\%https://www.ncbi.nlm.nih.gov/pmc/articles/PMC6520672/
$K^i(t)$ & Carrying capacity for mosquitoes in patch $i$\\
$\mu_{l}^i(t)$ & Larvae death rate in patch $i$\\
$m_{ij}$ & Traveling rate from patch $i$ to patch $j$
\\
\hline
\end{tabular}\label{Table1}
\end{table}
}
\clearpage

Let $\tau_h, \ \tau^i(t),\ \tau_l^i(t)$ represent the intrinsic incubation period, extrinsic incubation period and the larvae development period, respectively.  Based on the flowchart in Figure \ref{fig:chikungunya_model_patch}, the meta-population model is consisted of the following coupled subsystems:
{\small
\begin{equation}
\begin{aligned}[t]
\frac{dS_{h}^{i}}{dt} &= R^{i} - b_{m}^{i}(t)\beta_{mh} S_{h}^{i}(t) I_{m}^{i}(t) - \mu_{h} S_{h}^{i}(t) + \sum_{j \neq i} m_{ji} S_{h}^{j}(t) - \sum_{j \neq i} m_{ij} S_{h}^{i}(t) \\
\frac{dE_{h}^{i}}{dt} &= b_{m}^{i}(t) \beta_{mh} S_{h}^{i}(t) I_{m}^{i}(t) - b_{m}^{i}(t - \tau_h) \beta_{mh}(t - \tau_h) S_{h}^{i}(t - \tau_h) I_{m}^{i}(t - \tau_h) e^{-\mu_h \tau_h}  \\
& \quad - \mu_h E_{h}^{i}(t) + \sum_{{j \neq i}} m_{ji} E_{h}^{j}(t) - \sum_{{j \neq i}} m_{ij} E_{h}^{i}(t) \\
\frac{dI_{h}^{i}}{dt} &= (1-a)\beta_{mh} b_{m}^{i}(t - \tau_{h}) S_{h}^{i}(t - \tau_{h}) I_{m}^{i}(t - \tau_{h}) e^{-\mu_{h} \tau_{h}} - (\eta_h + \mu_{h}) I_{h}^{i}(t) \\
\frac{dA_{h}^{i}}{dt} &= a \beta_{mh} b_{m}^{i}(t - \tau_{h}) S_{h}^{i}(t - \tau_{h}) I_{m}^{i}(t - \tau_{h}) e^{-\mu_{h} \tau_{h}} - (\eta_{h} + \mu_{h}) A_{h}^{i}(t) + \sum_{j \neq i} m_{ji} A_{h}^{j}(t) - \sum_{j \neq i} m_{ij} A_{h}^{i}(t) \\
\frac{dR_{h}^{i}}{dt} &= \eta_h (I_{h}^{i}(t) + A_{h}^{i}(t)) - \mu_h R_{h}^{i}(t) + \sum_{j \neq i} m_{ji} R_{h}^{j}(t) - \sum_{j \neq i} m_{ij} R_{h}^{i}(t) \\
\frac{dL_{I}^{i}}{dt} &= \mu_{b}(t)\left(1-\frac{L_{s}^{i}(t)+L_{I}^{i}(t)}{K^{i}(t)}\right)p I_{m}^{i}(t) - \mu^{i}_{l}(t) L_{I}^{i}(t) \\
\frac{dL_{s}^{i}}{dt} &= \mu_{b}(t)\left(1-\frac{L_{s}^{i}(t)+L_{I}^{i}(t)}{K^{i}(t)}\right)\left(S_{m}^{i}(t) + E_{m}^{i}(t) + (1-p) I_{m}^{i}(t)\right) - \mu^{i}_{l}(t) L_{s}^{i}(t) \\
\frac{dS_{m}^{i}}{dt} &= (1-\dot{\tau}^{i}_{l}(t))\mu_{b}(t-\tau^{i}_{l}(t))\left(1-\frac{L_{s}^{i}(t-\tau^{i}_{l}(t))+L_{I}^{i}(t-\tau^{i}_{l}(t))}{K^{i}(t-\tau^{i}_{l}(t))}\right) \\
&\quad \times \left(S_{m}^{i}(t-\tau^{i}_{l}(t)) + E_{m}^{i}(t-\tau^{i}_{l}(t)) + (1-p) I_{m}^{i}(t-\tau^{i}_{l}(t))\right)e^{-\int_{t-\tau_{l}(t)}^{t} \mu^{i}_{l}(s) \, ds}  \\
&\quad  - b_{m}^{i}(t) \beta_{hm} S_{m}^{i}(t) (I_{h}^{i}(t) + A_{h}^{i}(t)) - \mu_{m}^{i}(t) S_{m}^{i}(t) \\
\frac{dE_{m}^{i}}{dt} &= b_{m}^{i}(t) \beta_{hm} S_{m}^{i}(t) (I_{h}^{i}(t) + A_{h}^{i}(t)) - \mu_{m}^{i}(t) E_{m}^{i}(t)  \\
&\quad - (1-\dot{\tau}^{i}(t)) \beta_{hm} b^{i}(t-\tau^{i}(t))(I_{h}^{i}(t-\tau^{i}(t)) + A_{h}^{i}(t-\tau^{i}(t))) S_{m}^{i}(t-\tau^{i}(t)) e^{-\int_{t-\tau^{i}(t)}^{t} \mu_{m}^{i}(s) \, ds} \\
\frac{dI_{m}^{i}}{dt} &= \left(1-\dot{\tau}^{i}_{l}(t)\right)\mu_{b}(t-\tau^{i}_{l}(t))\left(1-\frac{L_{I}^{i}(t-\tau^{i}_{l}(t)) + L_{s}^{i}(t-\tau^{i}_{l}(t))}{K^{i}(t-\tau^{i}_{l}(t))}\right)p I_{m}^{i}(t-\tau^{i}_{l}(t)) e^{-\int_{t-\tau^{i}_{l}(t)}^{t} \mu^{i}_{l}(s) \, ds}  \\
&\quad + \left(1-\dot{\tau}^{i}(t)\right)\beta_{hm} b_{m}^{i}(t-\tau^{i}(t))(I_{h}^{i}(t-\tau^{i}(t)) + A_{h}^{i}(t-\tau^{i}(t))) S_{m}^{i}(t-\tau^{i}(t)) e^{-\int_{t-\tau^{i}(t)}^{t} \mu^{i}_{m}(s) \, ds} \\
&\quad - \mu_{m}^{i}(t) I_{m}^{i}(t)
\end{aligned}\label{0}
\end{equation}
}
with initial value conditions
\begin{equation}
\begin{aligned}
& S_{h}^{i}(s)= \varphi^{i}_{1}(s), \;\; E_{h}^{i}(0)= \varphi^{i}_{2}(0), \;\;  I_{h}^{i}(s)= \varphi^{i}_{3}(s), \;\;  A_{h}^{i}(s)= \varphi^{i}_{4}(s), \;\; R_{h}^{i}(0)= \varphi^{i}_{5}(0), \\
& L_{I}^{i}(s)= \varphi^{i}_{6}(s), \;\; L_{s}^{i}(s)= \varphi^{i}_{7}(s), \;\; S_{m}^{i}(s)= \varphi^{i}_{8}(s), \;\; E_{m}^{i}(s)= \varphi^{i}_{9}(s), \;\; I_{m}^{i}(s)= \varphi^{i}_{10}(s),
\end{aligned}\label{0-1}
\end{equation}
for $s \in [-\widetilde{\tau}, 0],$ in which 
$$
\widetilde{\tau} := \max\{\tau_{0}^{i}\}_{i=1,\dots,n} \text{ where } \tau_{0}^{i} := \max \{ \sup_{0 \leq t \leq w} \tau^{i}(t), \tau^{i}_{h}, \sup_{0 \leq t \leq w} \tau^{i}_{l}(t) \},
$$ 
and $\tau^{i}(t), \tau_{l}^{i}(t)$ are strictly positive functions. 
Since the variables $E_h^i$ and $R_h^i$ are only present in their own equations, we only need to focus on a subsystem, that is, system (\ref{0})--(\ref{0-1}) without the $E_h^i$ and $R_h^i$ equations.  

To simplify the notation in Figure \ref{fig:chikungunya_model_patch} we have defined  
\begin{align*}
%d^i(t) &= \left(1-\dot{\tau}^{i}(t)\right)\beta_{hm}b_{m}^{i}(t-\tau^{i}(t))(I_{h}^{i}(t-\tau^{i}(t))+A_{h}^{i}(t-\tau^{i}(t)))S_{m}^{i}(t-\tau^{i}(t))e^{-\int_{t-\tau^{i}(t)}^{t} \mu^{i}_{m}(s) \, ds },\\ 
o_1^i(t),\ o_2^i(t) &= \mu_{b}(t)\left(1-\frac{L_{s}^{i}(t)+L_{I}^{i}(t)}{K^{i}(t)}\right),\\
o_3^i(t) &=(1-p)\mu_{b}(t)\left(1-\frac{L_{s}^{i}(t)+L_{I}^{i}(t)}{K^{i}(t)}\right),\\
o_4^i(t) &=p\ \mu_{b}(t)\left(1-\frac{L_{s}^{i}(t)+L_{I}^{i}(t)}{K^{i}(t)}\right).
\end{align*}

\section{Mathematical Analysis}
\setcounter{equation}{0}\setcounter{figure}{0}

\subsection{Well-posedness}
Now we will follow \cite{xiaoMain} to define the basic reproduction numbers for the mosquito population and human population, respectively.

%Let $\widetilde{\tau} := \max\{\tau_{0}^{i}\}_{i=1,\dots,n}$ where $\tau_{0}^{i} := \max \{ \sup_{0 \leq t \leq w} \tau^{i}(t), \tau^{i}_{h}, \sup_{0 \leq t \leq w} \tau^{i}_{l}(t) \}$, and the delays $\tau^{i}(t), \tau_{l}^{i}(t)$ are considered to be strictly positive functions. 
Define a Banach space $\Omega = C([-\widetilde{\tau}, 0], \mathbb{R}^{10n})$ for $t \geq 0$ and the space
\begin{align}\label{assump}
 \Omega(t) &= \left\{ \varphi \in C([-\widetilde{\tau}, 0], \mathbb{R}_{+}^{10n}) \ \middle| \ \varphi_{6}^{i}(s) + \varphi_{7}^{i}(s) \leq K^{i}(t+s), \ \forall s \in [-\widetilde{\tau}, 0], \right. \nonumber \\
 & \qquad \left. \varphi_{9}^{i}(0) = \int_{-\tau^{i}(0)}^{0} b_{m}^{i}(x) \beta_{hm} \varphi^{i}_{8}(x) (\varphi^{i}_{3}(x) + \varphi^{i}_{4}(x)) \, dx \right\}.
\end{align}

Assume that 
\begin{equation}\label{assumption}
K^{i}(t) \leq M^{i}, \ \forall t \in [-\widetilde{\tau}, \infty),
\end{equation}
for some $M^{i} > 0$, $i=1, \dots, n$.

\begin{lemma}\label{lemazimm}
Given $\varphi \in \Omega(0),$ system (\ref{0})--(\ref{0-1}) has a unique solution $$x(t,\varphi)=(x^{1}_{1}(t,\varphi),...,x^{1}_{10}(t,\varphi),...,x^{n}_{1}(t,\varphi),...,x^{n}_{10}(t,\varphi)),$$ defined on $[0,\infty)$, which is bounded and nonnegative, where $x(0,\varphi)=\varphi.$
\end{lemma}

\begin{proof}
 Let $\varphi \in \Omega.$ For fixed $t\in \mathbb{R}_{+},$ define a family of continuous operators  $f^{i}_{t}:\Omega \rightarrow \mathbb{R}^{10n}_{+}$ by 
$$
f^{i}_{t}(\varphi)=[{f_{t}}^{i}_{1}\varphi),{f_{t}}^{i}_{2}(\varphi),{f_{t}}^{i}_{3}(\varphi),{f_{t}}^{i}_{4}(\varphi),{f_{t}}^{i}_{5}(\varphi),{f_{t}}^{i}_{6}(\varphi),{f_{t}}^{i}_{7}(\varphi),{f_{t}}^{i}_{8}(\varphi),{f_{t}}^{i}_{9}(\varphi),{f_{t}}^{i}_{10}(\varphi)]^{T},
$$ 
\begin{align*}
\fontsize{5}{7}\selectfont
f_{t}(\varphi)=&\begin{bmatrix}
R^{i} - b_{m}^{i}(t)\beta_{mh} \varphi_{1}^{i}(t) \varphi_{8}^{i}(0) -\mu_{h}^{i} \varphi_{1}^{i}(0) + \sum_{j \neq i} m_{ji}\varphi_{1}^{j}(0)-\sum_{j \neq i} m_{ij}\varphi_{1}^{i}(0) \\
{f^{i}_{t}}_{2}(\varphi)\\
(1-a)\beta_{mh}b_{m}^{i}(t-\tau_{h}) \varphi_{1}^{i}(-\tau_{h}) \varphi_{8}^{i}(-\tau_{h})e^{-\mu_{h}\tau_{h}}-(\eta_h + \mu_{h})\varphi_{5}^{i}(0) \\
a\beta_{mh}b_{m}^{i}(t-\tau_{h})\varphi_{1}^{i}(-\tau_{h}) \varphi_{8}^{i}(-\tau_{h})e^{-\mu_{h}\tau_{h}} - (\eta_{h} + \mu_{h})\varphi_{3}^{i}(0)+ \sum_{j \neq i} m_{ji}\varphi_{3}^{j}(0)-\sum_{j \neq i} m_{ij}\varphi_{3}^{i}(0) \\
\eta_{h} (\varphi_{3}^{j}(0) + \varphi_{4}^{j}(0)) - \mu_{h} \varphi_{5}^{j}(0)+ \sum_{j \neq i} m_{ji}\varphi_{5}^{j}(0)-\sum_{j \neq i} m_{ij}\varphi_{5}^{j}(0)
\\
\mu_{b}(t)\left(1-\frac{\varphi_{6}^{i}(0)+\varphi_{7}^{i}(0)}{K^{i}(t)}\right)p\varphi_{10}^{i}(0)-\mu^{i}_{l}(t)\varphi_{6}^{i}(0) \\
\mu_{b}(t)\left(1-\frac{\varphi_{6}^{i}(0)+\varphi_{7}^{i}(0)}{K^{i}(t)}\right)((1-p)\varphi_{10}^{i}(0)+\varphi_{9}^{i}(0)+\varphi_{8}^{i}(0))-\mu^{i}_{l}(t)\varphi_{7}^{i}(0)\\
{f_{t}}^{i}_{8}(\varphi)\\
{f_{t}}^{i}_{9}(\varphi)\\
{f_{t}}^{i}_{10}(\varphi)
\end{bmatrix}_{i=1,...,n}
\end{align*}
\begin{equation}
\begin{aligned}
{f_{t}}^{i}_{2}(\varphi) =&b_{m}^{i}(t)\beta_{mh}\varphi_{1}^{i}(0)\varphi_{10}^{i}(0)-b_{m}^{i}(t-\tau_{h})\beta_{mh}(t-\tau_{h}) \varphi_{1}^{i}(-\tau_{h}) \varphi_{10}^{i}(-\tau_{h})e^{-\mu_{h}\tau_{h}} -\mu_{h}\varphi_{2}^{i}(0) \\
&+\sum_{j \neq i} m_{ji}\varphi_{2}^{j}(0) -\sum_{j \neq i} m_{ij}\varphi_{2}^{i}(0) \\
{f_{t}}^{i}_{8}(\varphi) =&(1-\dot{\tau^{i}_{l}}(t))\mu_{b}(t-\tau^{i}_{l}(t))\left(1-\frac{\varphi_{6}^{i}(-\tau^{i}_{l}(t))+\varphi_{7}^{i}(-\tau^{i}_{l}(t))}{K^{i}(t-\tau^{i}_{l}(t))}\right)\times\\
&\left(\varphi_{8}^{i}(-\tau^{i}_{l}(t)) + \varphi_{9}^{i}(-\tau^{i}_{l}(t)) + (1-p)\varphi_{10}^{i}(-\tau^{i}_{l}(t))\right)e^{-\int_{t-\tau_{l}(t)}^{t} \mu^{i}_{l}(s) \, ds } \\
&- b_{m}^{i}(t)\beta_{hm}\varphi_{8}^{i}(0) (\varphi_{2}^{i}(0)+\varphi_{3}^{i}(0))  -\mu_{m}^{i}(t)\varphi_{8}^{i}(0), \\
{f_{t}}^{i}_{9}(\varphi)=&b^{i}_{m}(t)\beta_{hm}\varphi_{6}^{i}(0) (\varphi_{4}^{i}(0)+\varphi_{3}^{i}(0)) -\mu_{m}^{i}(t)\varphi_{7}^{i}(0)\\
&-(1-\dot{\tau}^{i}(t))\beta_{hm}b^{i}(t-\tau^{i}(t))(\varphi_{4}^{i}(-\tau^{i}(t))+\varphi_{3}^{i}(-\tau^{i}(t)))\varphi_{8}^{i}(-\tau^{i}(t))e^{-\int_{t-\tau^{i}(t)}^{t} \mu_{l}^{i}(s) \, ds },\\
{f_{t}}^{i}_{10}(\varphi) &= \left(1-\dot{\tau^{i}_{l}}(t)\right)\mu_{b}(t-\tau^{i}_{l}(t))\left(1-\frac{\varphi_{7}^{i}(-\tau^{i}_{l}(t))+\varphi_{6}^{i}(-\tau^{i}_{l}(t))}{K^{i}(t-\tau^{i}_{l}(t))}\right)p\varphi_{10}^{i}(-\tau^{i}_{l}(t))e^{-\int_{t-\tau^{i}_{l}(t)}^{t} \mu^{i}_{l}(s) \, ds } \\
&+ \left(1-\dot{\tau}^{i}(t)\right)\beta_{hm}b_{m}^{i}(t-\tau^{i}(t))(\varphi_{4}^{i}(-\tau^{i}(t))+\varphi_{3}^{i}(-\tau^{i}(t)))\varphi_{8}^{i}(-\tau^{i}(t))e^{-\int_{t-\tau^{i}(t)}^{t} \mu^{i}_{l}(s) \, ds } \\
& -\mu_{m}^{i}(t) \varphi_{10}^{i}(0)
\end{aligned}
\end{equation}
Thus, the Fréchet derivative of $f_{t}$ at a fixed $\varphi \in \Omega$, $Df_{t}(\varphi)$ is clearly a continuous linear operator. So $f_{t}\in C^{1}(\Omega)$ and hence $f_{t}$ is locally Lipschitz in $\Omega.$ %see this on https://math.stackexchange.com/questions/3521228/does-c1-imply-locally-lipschitz-on-banach-spaces
It follows that the continuous mapping given by $f(t,\varphi)=f_{t}(\varphi)$ is locally Lipschitz on $\Omega(t)$ for fixed $t\in \mathbb{R}_{+}$. Hence, given $\varphi \in \Omega (0),$ system (\ref{0})--(\ref{0-1}) has a unique solution $x(t,\varphi)$ on $[0,\epsilon_{\varphi})$ for some $\epsilon_{\varphi} \in (0,\infty],$ due to Theorem 2.2.3 from \cite{hale1993introduction}.

Fix $\varphi\in\Omega(0)$ and let $x(t)=x(t,\varphi)$. Since 
$$
\frac{d(x_6^i(t)+x_7^i(t))}{dt} = \mu_b(t)\left( 1 - \frac{x_6^i(t)+x_7^i(t)}{K^i(t)}\right)(x_8^i(t)+x_9^i(t)+x_{10}^{i}(t)) - \mu_l^i(t)(x_6^i(t)+x_7^i(t)),
$$  
we have that $x_6^i(t)+x_7^i(t) \leq K^i(t) \leq \ M^i,\ \forall t \geq 0$, due to assumptions (\ref{assump}) and (\ref{assumption}). In particular, $x^{i}_{k}(t)$ is bounded in $[0,\epsilon_{\varphi}),\ k=6,7$. 
Now let $y(t)=\sum_{i=1}^{n}x^{i}_{1}(t)+x^{i}_{2}(t)+x^{i}_{3}(t)+x_{4}^{i}(t)+x_{5}^{i}(t)$, and so $\frac{dy}{dt}(t)=(\sum_{i=1}^{n}R^{i})-\mu_{h}y(t),$ and thus $0<\limsup_{t\to \epsilon_{\varphi}}y(t)< \infty.$ 
Consequently, $x^{i}_{k}(t)$ is bounded in $[0,\epsilon_{\varphi})$, $\forall \ 1\leq i\leq n$, and $k=1,2,3,4,5.$
%AGORA VOU ADD UP x_{8} e x_{6}

Suppose that $\limsup_{t \to \epsilon_{\varphi}}x_{10}^{i}(t)=\infty$. By continuity, $\lim_{t \to \epsilon_{\varphi}}x_{10}^{i}(t-\tau_{l}^{i}(t))<\infty$ and $\lim_{t \to \epsilon_{\varphi}}x_{8}^{i}(t-\tau_{l}^{i}(t))<\infty$. Therefore, $\liminf_{t \to \epsilon_{\varphi}}\frac{dx^{i}_{10}}{dt}(t)=-\infty,\ \exists \epsilon_{\varphi}>0,$ from system (\ref{0})--(\ref{0-1}), which is a contradiction as this would imply that there exists $0<\epsilon<\epsilon_{\varphi}$ such that $x^{i}_{k}(t)$ is strictly decreasing on the interval $(\epsilon, \epsilon_{\varphi})$. Hence, $x^{i}_{10}(t)$ is bounded. Analogously, $x^{i}_{k}(t)$ is bounded, for $k=8, 9.$ From Theorem 2.3.1 of \cite{hale1993introduction}, $\epsilon_{\varphi}=\infty.$ %pois em qualquer compacto dentro de [0,\epsilon_{\varphi}] x_{k} sao bounded por continuidade

Lastly, as $\varphi\geq 0,$ suppose that there exists a sequence $\{ t_{k} \}_{k=1}^{K}$, $t_{k}\in\mathbb{R}_{+}$, where $t_k$ is the $k^{th}$ least positive real number where  $x_{m}^{i}(t_{k})=0$ for the first time, for some $m=1,\dots ,10,\ i=1,\dots,n.$ Thus, $1\leq K\leq 10 n$, and by analyzing system (\ref{0})--(\ref{0-1}), we have that $\frac{d x_{m}^{i}}{d t}(t_{k})\geq 0,\ \forall 1\leq m \leq 10,\ m \neq 9,\ 1\leq i\leq n.$ This guarantees that $x_{m}^{i}$ will not assume negative values. If $m=9,$ then we have \begin{align*}\frac{dx_{9}^{i}}{dt} &=  b_{m}^{i}(t) \beta_{hm} x_{8}^{i}(t) (x_{3}^{i}(t) + x_{4}^{i}(t)) - \mu_{m}^{i}(t) x_{9}^{i}(t)\\&
\quad - (1-\dot{\tau}^{i}(t)) \beta_{hm} b^{i}(t-\tau^{i}(t))(x_{3}^{i}(t-\tau^{i}(t)) + x_{4}^{i}(t-\tau^{i}(t))) x_{8}^{i}(t-\tau^{i}(t)) e^{-\int_{t-\tau^{i}(t)}^{t} \mu_{m}^{i}(s) \ ds}\ ,
\end{align*}

$\forall 1\leq i\leq n.$ Let
$$h_{i}(t):=\int_{t-\tau^{i}(t)}^{t} b_{m}^{i}(s) \beta_{hm} x^{i}_{8}(s) (x^{i}_{3}(s) + x^{i}_{4}(s))e^{-\int_{s}^{t} \mu_{l}^{i}(x) \, dx }ds \geq 0,$$ and consider $$ \tilde{x}_{9}^{i}(t) : = e^{\int_{0}^{t} \mu_m(s)ds} x_{9}^{i}(t), \ \forall 1 \leq i \leq n. $$Therefore, $$\frac{d\tilde{x}_{9}^{i}}{dt} = \frac{dh_{i}}{dt} e^{\int_{0}^{t}\mu_m(s)ds} \geq \frac{dh_{i}}{dt},$$ what implies that $$\tilde{x}_{9}^{i}(t) \geq h_{i}(t)\geq 0,\forall t \geq 0, $$since $\tilde{x}_{9}^{i}(0)=x_{9}^{i}(0)= h_{i}(0),\ \forall 1 \leq i \leq n.$
Thus, we have that $\tilde{x}_{9}^{i}(t) \geq 0$, implying that $x_{9}^{i}(t) \geq 0,\ \forall t \in [0,\infty).$

Hence, $x(t,\varphi)$ is a nonnegative solution. Finally, we obtain that $x(t,\varphi)$ is a unique nonnegative solution defined on $[0,\infty)$.
\end{proof}

\subsection{Mosquito Population Dynamics}
Considering the disease-free state in system (\ref{0})--(\ref{0-1}) without the $E_h^i$ and $R_h^i$ equations, for $i=1,..,n$, we obtain the following system:
\begin{equation}
\begin{aligned}[t]
\frac{dS_{h}^{i}}{dt} =& R^{i} -\mu_{h} S_{h}^{i}(t) + \sum_{j \neq i} m_{ji}S_{h}^{j}(t)-\sum_{j \neq i} m_{ij}S_{h}^{i}(t)\\
\frac{dL_{s}^{i}}{dt} =& \mu_{b}(t)\left(1-\frac{L_{s}^{i}(t)}{K^{i}(t)}\right)S_{m}^{i}(t)-\mu^{i}_{l}(t)L_{s}^{i}(t) \\
\frac{dS_{m}^{i}}{dt} =& ((1-\dot{\tau}^{i}_{l}(t))\mu^{i}_{b}(t-\tau^{i}_{l}(t))\left(1-\frac{L_{s}^{i}(t-\tau^{i}_{l}(t))}{K^{i}(t-\tau^{i}_{l}(t))}\right)\\&\times\left(S_{m}^{i}(t-\tau^{i}_{l}(t))\right)e^{-\int_{t-\tau_{l}(t)}^{t} \mu_{l}^{i}(s) \, ds } - \mu_{m}^{i}(t)S_{m}^{i}(t).
\end{aligned}\label{pem}
\end{equation}
The disease-free equilibrium in \( \mathbb{R}^{3n} \) is given by
\[ 
({{\bar{S}_{h}}^{i}}, 0, 0)_{1 \leq i \leq n}, 
\]
where
\[ 
{{\bar{S}_{h}}^{i}} = \frac{R^{i} + \sum_{i \neq j} m_{ji}{{\bar{S}_{h}}^{j}}}{\mu_{h} + \sum_{i \neq j} m_{ij}}. 
\]

Now, since \( S^{i}_{h} \) only appear in their own equations for \( i = 1, \dots, n \), we can simplify this system to
\begin{equation}
\begin{aligned}[t]
\frac{dL_{s}^{i}}{dt} &= \mu^{i}_{b}(t)\left(1-\frac{L_{s}^{i}(t)}{K^{i}(t)}\right)S_{m}^{i}(t)-\mu^{i}_{l}(t)L_{s}^{i}(t) \\
\frac{dS_{m}^{i}}{dt} &=((1-\dot{\tau}^{i}_{l}(t))\mu^{i}_{b}(t-\tau^{i}_{l}(t))\left(1-\frac{L_{s}^{i}(t-\tau^{i}_{l}(t))}{K^{i}(t-\tau^{i}_{l}(t))}\right)\\&\times\left(S_{m}^{i}(t-\tau^{i}_{l}(t))\right)e^{-\int_{t-\tau_{l}(t)}^{t} \mu_{l}^{i}(s) \, ds } - \mu_{m}^{i}(t)S_{m}^{i}(t).
\end{aligned}\label{3}
\end{equation}
Then, linearizing (\ref{pem}) at the given equilibria, we have the system 
\begin{equation}
\begin{aligned}[t]
\frac{dS_{h}^{i}}{dt} &=-(\mu_{h}+\sum_{j \neq i} m_{ij} )S_{h}^{i}(t) + \sum_{j \neq i} m_{ji}S_{h}^{j}(t)\\
\frac{dL_{s}^{i}}{dt} &= \mu_{b}^{i}(t)S_{m}^{i}(t)-\mu_{l}^{i}(t)L_{s}^{i}(t) \\
\frac{dS_{m}^{i}}{dt} &=(1-\dot{\tau}^{i}_{l}(t))\mu^{i}_{b}(t-\tau^{i}_{l}(t))e^{-\int_{t-\tau_{l}(t)}^{t} \mu_{l}^{i}(s) \, ds }S_{m}^{i}(t-\tau^{i}_{l}(t)) - \mu_{m}^{i}(t)S_{m}^{i}(t),
\end{aligned}\label{4}
\end{equation}
which further simplifies to the following system:
\begin{equation}
\begin{aligned}[t]
\frac{dL_{s}^{i}}{dt} &= \mu_{b}^{i}(t)S_{m}^{i}(t)-\mu_{l}^{i}(t)L_{s}^{i}(t) \\
\frac{dS_{m}^{i}}{dt} &=(1-\dot{\tau^{i}_{l}}(t))\mu^{i}_{b}(t-\tau^{i}_{l}(t))e^{-\int_{t-\tau_{l}(t)}^{t} \mu_{l}^{i}(s) \, ds }S^{i}_{m}(t-\tau^{i}_{l}(t)) - \mu_{m}^{i}(t)S_{m}^{i}(t).
\end{aligned}\label{5}
\end{equation}
Define linear operators $\Bar{V}_{i}: \mathbb{R} \rightarrow \mathcal{L}(\mathbb{R}^{2}, \mathbb{R}^{2})$ by
$$\Bar{V}_{i}(t)=
\begin{bmatrix}
\mu_{l}^{i}(t) &  0\\
0 & \mu_{m}^{i}(t)\\
\end{bmatrix},
$$and 
$\Bar{F}_{i}: \mathbb{R} \rightarrow \mathcal{L}(C_{i}, \mathbb{R}^{2})$ by $$\Bar{F}_{i}(t)\begin{bmatrix}
\varphi^{i}_{1} \\
\varphi^{i}_{2} \\
\end{bmatrix}:=\begin{bmatrix}
\mu_{b}^{i}(t) \varphi^{i}_{2}(0) \\
(1-\dot{\tau}^{i}_{l}(t))\mu^{i}_{b}(t-\tau^{i}_{l}(t)) e^{-\int_{t-\tau^{i}_{l}(t)}^{t} \mu_{l}^{i}(s) \, ds } \varphi^{i}_{2}(-\tau^{i}_{l}(t)) \\
\end{bmatrix}
$$
for $C_{i}:=C([-\bar{\tau}^{i},0],\mathbb{R}^{2}),\ C_{i}^{+}:=C([-\bar{\tau}^{i},0],\mathbb{R}_{+}^{2})$ and $\bar{\tau}^{i}:=\max_{0 \leq t\leq w}\tau_{l}^{i}(t).$ 

Now we are in the position to define $\mathcal{R}_{m},$ the basic reproduction number for the mosquito population. Consider the system 
$$
\frac{dx}{dt}=-\Bar{V}_{i}(t)x.
$$ 
Let 
$$
\Bar{\Phi}_{i}(t,k):=\begin{bmatrix}
e^{-\int_{k}^{t} (\eta_{l}^{i}(t) + \mu_{l}^{i}(t))dt} &  0\\
0 & e^{-\int_{k}^{t}(\eta_{l}^{i}(t) + \mu_{l}^{i}(t))dt}\\
\end{bmatrix}, \;\; \forall t\geq k, 
$$
be its fundamental solution, where $\Bar{\Phi}(k,k)= I$.
Consequently, analogous to \cite{xiaoMain}, consider the following linear operator which, similar to the proof of Lemma \ref{R0}, is a compact operator, defined by $\Bar{L}_{i}:C_{w} \rightarrow C(\mathbb{R},\mathbb{R}^{2}),$ 
 $$\Bar{L}_{i}(\varphi)(t):= \int_{0}^{\infty}\Bar{\Phi}_{i}(t,t-s)\Bar{F}_{i}(t-s)\varphi(t-s+. )ds,$$ where $C_{w}$ is the set of $w$-periodic functions from $\mathbb{R}$ to $\mathbb{R}^{2}$. $\bar{L}_{i}$ can be proved to be compact, similarly to Lemma \ref{lema9}.

Therefore, the mosquito reproduction number is defined as the spectral radius of $\Bar{L}_{i}$ by 
\begin{equation}\label{Rm}
\mathcal{R}_{m}^{i}:=r(\Bar{L}_{i}).
\end{equation}
Hence, as $\Bar{L}_{i}$ is a compact operator, $\mathcal{R}^{i}_{m}<\infty$ is an eigenvalue of $\Bar{L}_{i},$ whose eigenvector lies in the positive cone $C_{w}^{+},$ by Krein-Rutman Theorem.

Consider the sets $\mathcal{I}^{i}_{1}:=  C([-\bar{\tau}^{i}_{l},0], \mathbb{R}^{2})$ and $ \mathcal{I}^{i}_{2}:= C([-\tau^{i}_{l}(0),0],\mathbb{R}^{2}),$ and their respective positive cones ${\mathcal{I}^{i}_{1}}^{+}:=   C([-\bar{\tau}^{i}_{l},0], \mathbb{R}_{+}^{2}), $ $ {\mathcal{I}^{i}_{2}}^{+}:= C([-\tau^{i}_{l}(0),0], \mathbb{R}_{+}^{2}).$ 

%SERA Q PODEMOS NAO USAR $C_{i}$, E APENAS DEFINIR ESSES DOMINIOS $\mathcal{I}^{i}_{k}$?  PERGUNTO ISSO PQ NAO ME PARECEU ESSENCIAL NA PROVA DE NENHUM DOS LEMAS USAR ESSE DEMINIO MAIOR. TALVEZ ISSO SEJA ESSENCIAL NO PAPER DO \cite{xiaoMain} MAS PRECISO OLHAR BEM!!! FALTA PROVAR O ULTIMO THM PRA RESPONDER ESSA PERGUNTA E OUTRAS!!

%DEVIDO AO THM DE KREIN-RUTMAN, NAO POSSO DEFINIR ${\mathcal{I}^{i}_{2}}^{+}(t)$ DA FORMA DE ANTES, POIS 
%${\mathcal{I}^{i}_{2}}^{+}(t)-{\mathcal{I}^{i}_{2}}^{+}(t)$ %NAO %EH DENSO EM ${\mathcal{I}^{i}_{2}}$ q eh o banach space.

Let us state the following lemmas regarding the existence and uniqueness of the initial value problems of (\ref{3}), defined over $\mathcal{I}^{i}_{1},$ $\mathcal{I}^{i}_{2}$ respectively, whose proofs will be omitted, as they are analogous to the proofs of Lemmas \ref{lemazimm} and \ref{lema7}, respectively.
%%More generally, the space C(K) of continuous functions on a compact metric space K equipped with the sup-norm is a Banach space.

\begin{lemma}\label{lemma1}
For $i=1,...,n$, given $\mathcal{\phi} \in \mathcal{I}^{i}_{1},$ system (\ref{3}) has a unique nonnegative solution $\bar{{P}}^{i}(t):\mathcal{I}^{i}_{1} \rightarrow \mathcal{I}^{i}_{1}$ 
 with $\bar{{P}}^{i}(0)=\mathcal{\phi}$ for all $t \in [0, \infty)$, for i=1,...,n.
\end{lemma}

\begin{lemma}\label{lemma2}
For $i=1,...,n$, given $\mathcal{\phi} \in \mathcal{I}_{2}^{i},$ system (\ref{3}) has a unique nonnegative solution ${P}^{i}(t):\mathcal{I}^{i}_{2} \rightarrow \mathcal{I}^{i}_{2}$ 
 with ${P}^{i}(0)=\mathcal{\phi}$ for all $t \in [0, \infty)$, for i=1,...,n.
\end{lemma}

By Lemmas \ref{lemma1} and \ref{lemma2}, $\bar{{P}}^{i}(t)|_{{\mathcal{I}^{i}_{1}}^{+}}: {\mathcal{I}^{i}_{1}}^{+} \rightarrow {\mathcal{I}^{i}_{1}}^{+}$ and ${P}^{i}(t)|_{{\mathcal{I}^{i}_{2}}^{+}}: {\mathcal{I}^{i}_{2}}^{+} \rightarrow {\mathcal{I}^{i}_{2}}^{+}$.
Also as a consequence of the last two Lemmas, we have the following remarks.

\begin{remark}\label{rmk00}
The same results established in Lemmas \ref{lemma1} and \ref{lemma2} will remain valid for system (\ref{5}). So the linear operators $\bar{Q}^{i}(t)\colon \mathbb{R}\times C([-\bar{\tau}^{i}_{l},0],\mathbb{R})\rightarrow \mathbb{R}\times C([-\bar{\tau}^{i}_{l},0],\mathbb{R})$ and ${Q}^{i}(t) \colon \mathbb{R}\times C([-{\tau}^{i}_{l}(0),0],\mathbb{R})\rightarrow \mathbb{R}\times C([-{\tau}^{i}_{l}(0),0],\mathbb{R})$ are solutions to system (\ref{5}) and both are invariant in their respective positive cones.
\end{remark}

\begin{remark}\label{rmk1}
By the uniqueness of solutions in Lemmas \ref{lemma1} and \ref{lemma2}, it follows that for any $\phi \in \mathcal{I}^{i}_{1}$ and $\bar{\phi} \in \mathcal{I}^{i}_{2},$ with $\phi^{i}_{1}=\bar{\phi}^{i}_{1}$, $\phi^{i}_{2}=\bar{\phi}^{i}_{2}$  for all $\theta$ in $[-\tau^{i}_{l}(0),0]$ and $1 \leq i \leq n$, one has $\bar{{P}}^{i}(t) = {P}^{i}(t)$ for all t in $[0, \infty],$ $1\leq i\leq n.$
\end{remark}

Similar to the proof of Lemma 3.5 from \cite{lou2017}, $\bar{{P}}^{i}(t)$ and ${P}^{i}(t)$  are $w-$periodic semiflows. 

\begin{remark}\label{rmk2}
Fix $1\leq i \leq n$ and let $x^{i}(t,x)$ be the solution of (\ref{3}) such that $x^{i}(0,x)=x \in \mathcal{I}^{i}_{1}$, and for fixed $t \geq 0,$ let $Dx^{i}$ be its Fréchet derivative at the equilibria $0 \in \mathbb{R}^{2}$ of system (\ref{3}), and define 
$$
f_{i}(x^{i}(t,x),t):=\begin{bmatrix}
\mu^{i}_{b}(t)\left(1-\frac{x_{1}^{i}(t,x)}{K^{i}(t)}\right)x_{2}^{i}(t,x)-\mu_{l}(t)x^{i}_{1}(t,x) \\
((1-\dot{\tau}^{i}_{l}(t))\mu^{i}_{b}(t-\tau^{i}_{l}(t))\left(1-\frac{x_{1}^{i}(t-\tau^{i}_{l}(t),x)}{K^{i}(t-\tau^{i}_{l}(t))}\right)\times\\
\left(x_{2}^{i}(t-\tau^{i}_{l}(t),x)\right)e^{-\int_{t-\tau_{l}(t)}^{t} \mu(s) \, ds } - \mu_{m}^{i}(t)x_{2}^{i}(t,x)\\
\end{bmatrix},
$$
with Fréchet derivative at the origin $Df_{i}(x^{i}(t,.),t):\mathcal{I}^{i}_{1} \rightarrow \mathcal{I}^{i}_{1} $, for $t\geq 0$ fixed, given by 
$$
Df_{i}\begin{bmatrix}
x^{i}_{1}(t,x) \\
x^{i}_{2}(t,x) \\
\end{bmatrix}:= \begin{bmatrix}
\mu_{b}^{i}(t)x^{i}_{1}(t,x)-\mu_{l}^{i}(t) x^{i}_{2}(t,x) \\
(1-\dot{\tau}^{i}_{l}(t))\mu^{i}_{b}(t-\tau^{i}_{l}(t)) e^{-\int_{t-\tau^{i}_{l}(t)}^{t} \mu^{i}(s) \, ds } x^{i}_{2}(t-\tau^{i}_{l}(t),x)-\mu_{m}^{i}(t)x^{i}_{2}(t,x) \\
\end{bmatrix},$$
for $i=1,\dots,n.$
Thus, we can take the Fréchet Derivative at the origin and use the chain rule, considering that $x^{i}(t,0)=0$ is an equilibria of (\ref{5}):
$$
D\dot{x}^{i}=Df_{i}(x^{i}(t,x),t)Dx^{i}, \ \forall i=1,..,n.
$$
Since the maps $(t,x) \mapsto D\dot{x}^{i}(t,x)$ and $(t,x) \mapsto \frac{d}{dt}Dx^{i}(t,x)$ are continuous, we have $D\dot{x}^{i}=\frac{d}{dt}Dx^{i}$, for any $i=1,\dots,n$.\label{rmk} 
\end{remark}

Fix $i=1,\dots,n$. Let $\bar{P}^{i}(t,\varphi): \mathbb{R}\times \mathcal{I}^{i}_{1} \rightarrow \mathcal{I}^{i}_{1}$, where $\bar{P}^{i}(t,\varphi)|_{\mathbb{R}_{+}\times {\mathcal{I}^{i}_{1}}^{+}}$ is the solution map of system (\ref{3}), which is invariant in the positive cone $\mathbb{R_{+}}\times {\mathcal{I}^{i}_{1}}^{+}$.
Therefore by Lemma \ref{lemma2}, Remarks \ref{rmk00}, \ref{rmk1} and \ref{rmk}, we obtain that $D\bar{P}^{i}|_{\mathbb{R}_{+}\times C([-\bar{\tau}^{i}_{l},0],\mathbb{R}_{+})}(t,\varphi)=\bar{Q}^{i}(t,\varphi)$ is solution map of system (\ref{5}), once restricted to the positive cone \small{${\mathbb{R}_{+}\times C([-\bar{\tau}^{i}_{l},0],\mathbb{R}_{+})}$}. 

In view of Remark \ref{rmk1}, the same results hold for the solution map ${P}^{i}(t,\varphi): \mathbb{R}\times \mathcal{I}^{i}_{2} \rightarrow \mathcal{I}^{i}_{2}.$ Therefore, we define 
$$
Q^{i} := DP^{i}(w, \cdot),\ \bar{Q}^{i} := D\bar{P}^{i}(w, \cdot)
$$ 
to be the Poincar\'e maps of system (\ref{5}), for $1\leq i \leq n$ fixed. 
%Define $P := \prod_{i=1}^{n} P^{i}$ and $Q := \prod_{i=1}^{n} Q^{i}$.

 %IN FACT THE FRECHET DERIVATIVE OF A PRODUCT OF OPERATORS EUQALS TO THE PRODUCT OF THE FRECHET DERIVATIVES. ISTO TEM NALGUM PDF, TEM Q PROCURAR DNV
%FATO: https://math.stackexchange.com/questions/533063/frechet-derivative-of-compact-operator-is-compact 
%COMO ESSE LINK JA GARANTE $Q^{i}$ SAO COMPACTOS, PELA PROVA DE \cite{lou2017} NAO PRECISAMOS PROVAR Q SAO STRONGLY POSITIVE, SE FORMOS REALMENTE USAR OS DOIS DOMINIOS $C_{i},\mathcal{I}^{i}_{1,2}$ COMO DISCUTIDO ACIMA...
 
\begin{remark}\label{rmkk}
Since $Q^{i},\ \bar{Q}^{i}$ are Fr\'{e}chet derivatives of compact operators, they are compact linear operators, so by Krein-Rutman Theorem and similarly to \cite[Lemma 3.8]{lou2017}, %and LEMA Q FALA Q OS DOIS SPECTRAL RADIUS DE $Q^{i},\bar{Q}^{i}$ SAO IGUAIS, 
$r(\bar{Q}^{i})=r(Q^{i})<\infty,$ where $r(\bar{Q}^{i}), \ r(Q^{i})$ are eigenvalues of $Q^{i},\ \bar{Q}^{i},$ respectively, whose eigenvector lies on $\mathbb{R}_{+}\times C([-\bar{\tau}^{i}_{l},0],\mathbb{R}_{+})-\{0\}$.
\end{remark}
%Now let $\bar{\mathcal{I}}_{k}:=
%\mathbb{R}\times\prod_{i=1}^{n} {\mathcal{I}_{k}}^{i}$, and ${{\bar{\mathcal{I}}}^{+}}_{k}:=
%\mathbb{R}_{+}\times \prod_{i=1}^{n}
%{{{\mathcal{I}_{k}}}^{i}}^{+},\ k=1,2$. Given $\phi = (\phi_{1}, \dots, \phi_{n}) \in \bar{\mathcal{I}}_{1}$, we define $P(t, \phi) = \prod_{i=1}^{n} \bar{P}^{i}(t, \phi_{i})$, which represents the solution map of the full system (\ref{3}) by varying $i = 1, \dots, n$. ACHO Q SE O LEMA ABAIXO N PRECISAR PEGAR O SISTEMA GLOBAL, ENTAO NAO PRECISA DEFINIR ESSES CARAS VEI...

\begin{lemma}\label{precompact}
Let $1\leq i\leq n$ and $\varphi \in {{\mathcal{I}}_{1}^{{i}^{+}}}.$ Then the orbit $t \mapsto \bar{P}^{i}(t,\varphi)$ is precompact. 
\end{lemma}

\begin{proof}
Fix $1\leq i \leq n,$ and suppose $\exists \ \phi^{i} \in {\mathcal{I}}_{1}^{i}$ such that there exists a time sequence $0 \leq t_{1}<t_{2}<\dots<t_{k}<\dots$ where the sequence $(\bar{P}^{i}(t_{k},\phi))_{k\geq 1}$ has no convergent subsequences. 
By uniqueness of solutions for system (\ref{3}), since $\bar{P}^{i}(0,\phi) = \phi$, we have 
$$
Q^{i}\bar{P}^{i}(t,\phi) = \bar{P}^{i}(t,\bar{Q}^{i}\phi)
$$ 
for all $t \geq 0$. Similar to the proof of Lemma 3.1, it is possible to prove that the trajectory $t \mapsto \bar{P}^{i}(t,\bar{Q}^{i}\phi)$ is bounded. Therefore, since $\bar{Q}^{i}$ is compact, due to the above equation, this trajectory is precompact. Consequently, the sequence $(\bar{P}^{i}(t_{k},\bar{Q}^{i}\phi))_{k\geq 1}$ has a convergent subsequence, i.e., the limit $$\lim_{m\rightarrow \infty}\bar{P}^{i}(t_{j_{m}},\bar{Q}^{i}\phi)=\lim_{m\rightarrow \infty}\bar{Q}^{i}\bar{P}^{i}(t_{j_{m}},\phi)$$ exists. 

By the continuity of $\bar{Q}^{i}$, $\lim_{m\rightarrow \infty}\bar{P}^{i}(t_{j_{m}},\phi)$ exists, which is a contradiction. Therefore, $\forall 1\leq i \leq n, \ \phi \in \mathcal{I}_{1}^i,$ the orbit $t \mapsto \bar{P}^{i}(t,\varphi)$ is precompact.
\end{proof}

According to \cite{xiaoMain}, we have the following result.

\begin{lemma}
$\mathcal{R}^{i}_{m} - 1\leq 0$ if and only if $r(Q^{i}) - 1\leq 0$ for $i=1,...,n$.
\end{lemma}

%\begin{remark}
%Lastly, according to the results above and in particular Remark \ref{rmk1}, VOU FAZER TUDO PROVADO PARA $Q^{i}.$
%\end{remark}

So we are in a position to prove the following threshold dynamics for the mosquito population.
%ESSE LEMA EH EM $\mathcal{I}_{2}$ MSM, AGR TENHO Q USAR OS RESULTADOS ACIMA PARA DE $\mathcal{I}_{1}$ PASSAR PRA $\mathcal{I}_{2}$. TENHO Q LER O PAPER DO LIU ANTES PRA ENTENDER MLR COMO Q VOU FAZER

%ATE AGR ME PARECE Q N PRECISO DEFINIR OS $C_{i}$, TEM Q  VER SE O PAPER DO \cite{xiaoMain} REQUER OU SE O THM PRINCIPAL REQUER...

\begin{lemma}\label{mainlemma}
If $\mathcal{R}_{m}^{i}\leq 1$ for all $i=1,..., n,$ then 
$(\bar{S}^{i}_{h},0,0)_{i=1,..,n}\in \mathbb{R}^{3n}$ is globally stable in $\prod_{i=1}^{n}{\mathcal{I}^i_{2}}^{+}$. Otherwise, consider without loss of generality that $1< i_{1}<..<i_{k}< n$, such that $\mathcal{R}^{i_{j}}_{m}>1$ for $j=1,…,k, \ 1\leq k< n$. Thus, by letting
$x_{i}:=({\bar{S}^{i}_{h}},0,0)$, and
$y_{i}(t):=({\bar{S}^{i}_{h}},{L^{i}_{s}}^{*}(t),{S^{i}_{m}}^{*}(t))$,
there exists a unique $w$-periodic solution $\gamma(t)$ of system (\ref{3}), given by $\gamma(t)=$
$(x_{1},...,y_{i_{1}}(t),…,y_{i_{k}}(t),...,x_{n}),$ which is globally stable in $\prod_{i=1}^{n}{\mathcal{I}^i_{2}}^{+}-\{(x_{1},..,x_{n})\}$. \label{lemmaRm}
\end{lemma}
%O FATO EH O SEGUINTE: NA PROVA DESSE LEMA, EU PRECISO PROVAR Q O MAPA $P^{i}$ EH EVENTUALLY STRONGLY MONOTONE???SE NAO FOR PRECISO, ENTAO OMITO ISSO E EXCLUO O LEMMA 8. Nao sera preciso contanto q $P^{i}$ seja strongly subhomogeneous. TO COMENTANDO PARTE DA SOLUCAO Q SERA DESNECESSARIA CASO EU PROVE A STRONGLY SUBHOMOGENEITY... EH O PAPER Q WALTER Q VAI SER DAR O XEQUE-MATE... OTHER HYPOTHESIS I SHOULD PROVE: Df(0) is compact, f(0)=0, f is assymptotically smooth(??), and every positive orbit of f in V is bounded-proof analogous to in lemma 1-
\begin{proof}
In system (\ref{3}) fix $1 \leq i \leq n$. Due to Lemma \ref{lemma2}, $\bar{P}^{i}(t)$ is monotone for all $t\in \mathbb{R}_{+}$. Now we claim that $\bar{P}^{i}(t)$ is strongly subhomogeneous. 
Let $v\in$ ${\mathcal{I}^{i}_{1}}^{+}-\{(0,0)\}$, $\lambda \in (0,1),$ and consider the solution maps $\bar{P}^{i}(t,v), \ \bar{P}^{i}(t,\lambda v)$ of system (\ref{3}), where $\bar{P}^{i}(0,v)=v,$ $\bar{P}^{i}(0,\lambda v)=\lambda v.$ 
Consider the maps $h^{i}(t):=\lambda \bar{P}^{i}(t,v), \ g^{i}(t):=\bar{P}^{i}(t,\lambda v).$ 
Let $b>0,\ a=0$. Using the same notation as in Theorem 4 of \cite{walter1997} and Chapter 2 of \cite{zhao2017}, notice that $0=Pg^{i} \gg Ph^{i}$ in $(a,b]$, and hence, we obtain $h\ll g$ in $(a,b],$ by Theorem 4 in \cite{walter1997}. Since $b>0$ is arbitrary and $h(t), g(t)$ are defined over $\mathbb{R}_{+}$ by Lemma \ref{precompact}, we find that $h \ll g$ in $(0,\infty)$. 
Therefore, $\bar{P}^{i}(t,v)$ is strongly subhomogeneous. 
%https://math.stackexchange.com/questions/3218973/compactness-in-normed-vector-spaces

Analogous to Lemma \ref{lemazimm}, we can prove that the map $t\mapsto \bar{P}^{i}(t,v)$ is bounded. %Now, we should prove $\bar{P}^{i}(t)$ is asymptotically stable, i.e, given an invariant closed bounded set $B\subset{\mathcal{I}^{i}_{2}}^{+},$ there exists a totally bounded set $K\subset {\mathcal{I}^{i}_{2}}^{+}$ such that $d(\bar{P}^{i}(nt)(B),K)\rightarrow 0,$ as $n\rightarrow \infty.$ Indeed, consider the w-limit set $w(B)$, which is closed, by definition. Suppose $w(B)$ is not totally bounded. So, $\exists \epsilon>0$ where any finite collection of balls $B(x_{k},\epsilon)$ does not cover $w(B).$
Moreover, similar to the proof of Lemma \ref{lema9}, $D\bar{P}^{i}(t)$ is a compact linear operator.
Lastly, according to the proof of Theorem 2.3.4 \cite{zhao2017}, it suffices to show that every positive orbit of $t\rightarrow \bar{P}^{i}(t,v)$ is precompact, i.e, has a compact closure. This fact is a consequence of Lemma \ref{precompact}.

%MAS EU NAO SEI SE ISSO VAI SER MAIS DIFICIL NA REAL.

%since its $w-limit$ set its internally transitive.

%We should prove it is totally bounded. I SHOULD USE INTERNALLY TRANSITIVE CHAIN!!!!!!!

%Furthermore, as $(0,0)$ is an equilibrium, $\bar{P}^{i}(t)((0,0))=(0,0), \ \forall t\geq 0.$ And in addition, using the same notation of Remark \ref{rmk2}, consider the family of linear bounded operators $\{D\bar{P}^{i}(t)\}_{t\in \mathbb{R}_{+}}$, where for each $t\geq 0,$ $D\bar{P}^{i}(t)$ is the Fréchet Derivative at the origin (0,0) of the operator $P^{i}(t)$. Thus, given $\varphi \in {\mathcal{I}_{2}^{i}}^{+}$, $\{DP^{i}(t)(\varphi)\}_{t\in\mathbb{R}_{+}}$ is uniformly bounded. Also, since $\mu_{m}^{i}(t), \mu_{b}^{i}(t), \mu_{l}^{i}(t), \tau_{l}^{i}(t)$ are all $w-$periodic functions, it is clear that given $\varepsilon>0,$ there exists $\delta>0$ such that if $v,w \in {\mathcal{I}^{i}}_{2}$ where $sup \norm{v-w} <\delta$, then $\norm{D\bar{P}^{i}(t)v-D\bar{P}^{i}(t)w} <\varepsilon.$ Therefore, $\{DP^{i}(t)\}_{t\in\mathbb{R}_{+}}$ equicontinuous family. Consequently, $D\bar{P}^{i}(t)$ is compact operator, by Arzela-Ascoli Theorem.

%ASCOLI-ARZELA: https://web.stanford.edu/class/stats300b/Notes/arzela-ascoli.pdf

In view of Remark \ref{rmk1}, all the above results are valid for the solution ${P}^{i}:\mathbb{R}\times\mathcal{I}^{i}_{2} \rightarrow \mathcal{I}^{i}_{2}$. Now we are in position to use Theorem 2.3.4 of \cite{zhao2017} to obtain that, by letting $t=w$, $\bar{P}^{i}:=\bar{P}^{i}(w),$ we find that $\mathcal{R}_{m}^{i}-1=r(\bar{Q}^{i})-1=r(Q^{i})-1\leq 0 \implies x_{i}$ is globally stable in ${\mathcal{I}^{i}}_{2}^{+}$, and $\mathcal{R}_{m}^{i}-1=r(\bar{Q}^{i})-1=r(Q^{i})-1>0 \implies y_{i}(t)$ is globally stable in $ {\mathcal{I}^{i}}_{2}^{+}-\{(0,0) \}$ for all $i=1,\dots, n$, where $r(Q^{i})$ is a positive eigenvalue of $Q^{i}$, whose eigenvector lies on the positive cone ${\mathcal{I}}^{i}_{2}{}^{+}$, by Krein-Rutman Theorem.

Without loss of generality, let $\mathcal{R}_{m}^{i_{j}}>1$ for $j=1,...,k.$ Consider a trajectory $\Gamma(t) \in \bar{\mathcal{I}}^{+}_{2}-\{(x_{1},..,x_{n})\}$, and let $\Pi:\prod_{i=1}^{n}\mathcal{I}_{2}^{i}\rightarrow \prod_{j=1}^{k}\mathcal{I}_{2}^{i_{j}}$ and $\bar{\Pi}:\prod_{i=1}^{i=n}\mathcal{I}_{2}^{i} \rightarrow \prod_{j \in \{1,...,n \}-\{i_{1},...,i_{k} \}}\mathcal{I}_{2}^{j}$ be natural projections. Since as we vary $i=1,...,n$, system (\ref{3}) is uncoupled, using the results in the last paragraph, we find that $\Pi(\Gamma(t))\rightarrow (\gamma_{i_{1}}(t),...,\gamma_{i_{k}}(t)),$ $\bar{\Pi}(\Gamma(t))\rightarrow (0,...,0)\in \mathbb{R}^{2(n-k)}$ as $t\rightarrow \infty$. Without loss of generality consider that $\mathcal{R}_{m}^{1},\mathcal{R}_{m}^{n}\leq 1.$ Thus we obtain that $\gamma(t):=(x_{1},...,y_{i_{1}}(t),…,y_{i_{k}}(t),...,x_{n})$ a periodic solution of system (\ref{pem}), as we vary $i=1,...,n,$ is globally stable in $\prod_{i=1}^{n}{\mathcal{I}^i_{2}}^{+}-\{(x_{1},..,x_{n})\}$. On the other hand, analogously, if $\mathcal{R}_{m}^{i}\leq 1,\ \forall 1\leq i\leq n,$ then $(x_{1},...,x_{n})$ is globally stable in $\prod_{i=1}^{n}{\mathcal{I}^i_{2}}^{+}.$
\end{proof}
%INFECTIVE SYSTEM IS THE SYSTEM CAPABLE FOR CAUSING INFECTION.

\subsection{Human Population Dynamics}
In the following, we shall proceed according to \cite{xiaoMain}. We assume, without loss of generality, that $\mathcal{R}^{1}_{m},\mathcal{R}_{m}^{n}\leq 1$, $\mathcal{R}_{m}^{i_{j}}>1$ for $j=1,...,k$ with $1\leq k<n$ and $1\leq i_{1}<...<i_{k}<n$, and then linearize the infective system of (\ref{0})--(\ref{0-1}) over the disease-free periodic trajectory established in Lemma \ref{mainlemma}, but now regarding it as a trajectory of system (\ref{0})--(\ref{0-1}) without the $E_h^i$ and $R_h^i$ equations given by $\gamma(t):=(x_{1},...,y_{i_{1}}(t),...,y_{i_{k}}(t),...,x_{n})$ with $x_{i}=({S_{h}^{i}}^{*},0,0,0,0,0,0,0)$ and $y_{i_{j}}(t)=({S_{h}^{i_{j}}}^{*},0,0,{L_{s}^{i_{j}}}^{*}(t),0,{S_{m}^{i_{j}}}^{*}(t),0,0)$ for $j=1,...,k.$
For $1\leq j \leq k<n$ such that $\mathcal{R}_{m}^{i_{j}}> 1,$
\begin{equation}
\begin{aligned}[t]
\frac{dI_{h}^{{i}_{j}}}{dt} =&(1-a)\beta_{mh}b_{m}^{i_{j}}(t-\tau_{h}) {S_{h}^{i_{j}}}^{*} I_{m}^{i_{j}}(t-\tau_{h})e^{-\mu_{h}\tau_{h}}-(\eta_h + \mu_{h})I^{i_{j}}_{h}(t)\\
\frac{dA_{h}^{i_{j}}}{dt} =& a\beta_{mh}b_{m}^{i_{j}}(t-\tau_{h}) {S_{h}^{i_{j}}}^{*}I_{m}^{i_{j}}(t-\tau_{h})e^{-\mu_{h}\tau_{h}} - (\eta_{h} + \mu_{h}+\sum_{l \neq i_{j}} m_{i_{j}l})A_{h}^{i_{j}}(t)+ \sum_{l \neq i_{j}} m_{li_{j}}A_{h}^{l}(t) \\
\frac{dL_{I}^{i_{j}}}{dt} =&\mu^{i_{j}}_{b}(t)\left(1-\frac{{L_{s}^{i_{j}}}^{*}(t)}{K^{i_{j}}(t)}\right)pI_{m}^{i}(t)-\mu_{l}(t)L_{I}^{i_{j}}(t)\\
\frac{dI_{m}^{i_{j}}}{dt} =& \left(1-\dot{\tau}^{i_{j}}_{l}(t)\right)\mu_{b}^{i_{j}}(t-\tau^{i_{j}}_{l}(t))\left(1-\frac{{{L_{s}^{i_{j}}}^{*}}(t-\tau^{i_{j}}_{l}(t))}{K^{i_{j}}(t-\tau^{i_{j}}_{l}(t))}\right)pI_{m}^{i_{j}}(t-\tau^{i_{j}}_{l}(t))e^{-\int_{t-\tau^{i_{j}}_{l}(t)}^{t} \mu^{i_{j}}_{l}(s) \, ds } \\
&+ \left(1-\dot{\tau}^{i_{j}}(t)\right)\beta_{hm}b^{i_{j}}(t-\tau^{i_{j}}(t))(I_{h}^{i_{j}}(t-\tau^{i_{j}}(t))+A_{h}^{i_{j}}(t-\tau^{i_{j}}(t))){S_{m}^{i_{j}}}^{*}(t-\tau^{i_{j}}(t))e^{-\int_{t-\tau^{i_{j}}(t)}^{t} \mu^{i_{j}}_{l}(s) \, ds } \\
&-\mu_{m}^{i_{j}}(t) I_{m}^{i_{j}}(t), 
\end{aligned}\label{6}
\end{equation}
and for $1\leq i\leq n$ such that $\mathcal{R}_{m}^{i}\leq 1$,
\begin{equation}
\begin{aligned}[t]
\frac{dI_{h}^{i}}{dt} &=(1-a)\beta_{mh}b_{m}^{i}(t-\tau_{h}) {S_{h}^{i}}^{*} I_{m}^{i}(t-\tau_{h})e^{-\mu_{h}\tau_{h}}-(\eta_h + \mu_{h})I^{i}_{h}(t)\\
\frac{dA_{h}^{i}}{dt} &= a\beta_{mh}b_{m}^{i}(t-\tau_{h}) {S_{h}^{i}}^{*}I_{m}^{i}(t-\tau_{h})e^{-\mu_{h}\tau_{h}} - (\eta_{h} + \mu_{h}+\sum_{l \neq i} m_{il})A_{h}^{i}(t)+ \sum_{l \neq i} m_{li}A_{h}^{l}(t) \\
\frac{dL_{I}^{i}}{dt} &=\mu^{i}_{b}(t)pI_{m}^{i}(t)-\mu^{i}_{l}(t)L_{I}^{i}(t)\\
\frac{dI_{m}^{i}}{dt} &= \left(1-\dot{\tau}^{i}_{l}(t)\right)\mu_{b}^{i}(t-\tau^{i}_{l}(t))pI_{m}^{i}(t-\tau^{i}_{l}(t))e^{-\int_{t-\tau^{i}_{l}(t)}^{t} \mu_{l}^{i}(s) \, ds } -\mu_{m}^{i}(t) I_{m}^{i}(t).
\end{aligned}\label{7}
\end{equation}

Let $\tau^{i}_{0}:=\max \{ \sup_{0\leq t\leq w}\tau^{i}(t), \ \tau^{i}_{h},\ \sup_{0\leq t\leq w}\tau^{i}_{l}(t)\}$ for $i=1,...,n$, and define $C:=\prod_{i=1}^{n}C([-\tau^{i}_{0},0],\mathbb{R}^{4})$ and $C^{+}:=\prod_{i=1}^{n}C([-\tau^{i}_{0},0],\mathbb{R}^{4}_{+}).$
Define a mapping $F:\mathbb{R}\rightarrow \mathcal{L}(C^{+},\mathbb{R}^{4n})$ by 
\begin{equation}\nonumber
F(t)\begin{bmatrix}
\varphi^{1}_{1} \\
\varphi^{1}_{2} \\
\varphi^{1}_{3} \\
\varphi^{1}_{4}\\
\vdots \\
\varphi^{i_{j}}_{1} \\
\varphi^{i_{j}}_{2} \\
\varphi^{i_{j}}_{3} \\
\varphi^{i_{j}}_{4} \\
\vdots \\
\varphi^{n}_{1} \\
\varphi^{n}_{2} \\
\varphi^{n}_{3} \\
\varphi^{n}_{4}\\
\end{bmatrix}:=\begin{bmatrix}
(1-a)\beta_{mh}b_{m}^{1}(t-\tau_{h}) ({S_{h}^{1}}^{*}) \varphi^{1}_{4}(-\tau_{h})e^{-\mu_{h}\tau_{h}}\\
 a\beta_{mh}b_{m}^{1}(t-\tau_{h}) ({S_{h}^{1}}^{*})\varphi^{1}_{4}(-\tau_{h})e^{-\mu_{h}\tau_{h}} + \sum_{l \neq 1} m_{l1}\varphi^{l}_{2}(0)\\
\mu^{1}_{b}(t)p\varphi^{1}_{4}(0)\\
 \left(1-\dot{\tau}^{1}_{l}(t)\right)\mu_{b}^{1}(t-\tau^{1}_{l}(t))p\varphi_{4}^{1}(-\tau^{1}_{l}(t))e^{-\int_{t-\tau^{1}_{l}(t)}^{t} \mu_{l}^{1}(s) \, ds }\\
 \vdots \\
(1-a)\beta_{mh}b_{m}^{i_{j}}(t-\tau_{h}) ({S_{h}^{i_{j}}}^{*}) \varphi^{i_{j}}_{4}(-\tau_{h})e^{-\mu_{h}\tau_{h}}\\
 a\beta_{mh}b_{m}^{i_{j}}(t-\tau_{h}) ({S_{h}^{i_{j}}}^{*})\varphi^{i_{j}}_{4}(-\tau_{h})e^{-\mu_{h}\tau_{h}} + \sum_{l \neq i_{j}} m_{li_{j}}\varphi^{l}_{2}(0)\\
\mu^{i_{j}}_{b}(t)\left(1-\frac{{L_{s}^{i_{j}}}^{*}(t)}{K^{i_{j}}(t)}\right)p\varphi^{i_{j}}_{4}(0)\\
\left(1-\dot{\tau}^{i_{j}}_{l}(t)\right)\mu_{b}^{i_{j}}(t-\tau^{i_{j}}_{l}(t))\left(1-\frac{({{L_{s}^{i_{j}}}^{*}}(t-\tau^{i_{j}}_{l}(t)))}{K^{i_{j}}(t-\tau^{i_{j}}_{l}(t))}\right)p\varphi_{4}^{i_{j}}(-\tau^{i_{j}}_{l}(t))e^{-\int_{t-\tau^{i_{j}}_{l}(t)}^{t} \mu^{i_{j}}_{l}(s) \, ds } +\\
+ \left(1-\dot{\tau}^{i_{j}}(t)\right)\beta_{hm}b^{i_{j}}(t-\tau^{i_{j}}(t))(\varphi_{1}^{i_{j}}(-\tau^{i_{j}}(t))+\varphi_{2}^{i_{j}}(-\tau^{i_{j}}(t))({S_{m}^{i_{j}}}^{*}(t-\tau^{i_{j}}(t)))e^{-\int_{t-\tau^{i_{j}}(t)}^{t} \mu^{i_{j}}_{l}(s) \, ds }\\
 \vdots \\
(1-a)\beta_{mh}b_{m}^{n}(t-\tau_{h}) ({S_{h}^{n}}^{*}) \varphi^{n}_{4}(-\tau_{h})e^{-\mu_{h}\tau_{h}}\\
 a\beta_{mh}b_{m}^{n}(t-\tau_{h}) ({S_{h}^{n}}^{*})\varphi^{n}_{4}(-\tau_{h})e^{-\mu_{h}\tau_{h}} + \sum_{l \neq n} m_{ln}\varphi^{l}_{2}(0)\\
\mu^{n}_{b}(t)p\varphi^{n}_{4}(0)\\
 \left(1-\dot{\tau}^{n}_{l}(t)\right)\mu_{b}^{n}(t-\tau^{n}_{l}(t))p\varphi_{4}^{n}(-\tau^{n}_{l}(t))e^{-\int_{t-\tau^{n}_{l}(t)}^{t} \mu_{l}^{n}(s) \, ds }
\end{bmatrix},
\end{equation}
where $j=1,...,k<n.$ In fact $F(t)C^{+} \subset C^{+}$.
Moreover, for $t\geq 0$, let us define 
\begin{eqnarray*}
& V(t)=-\text{diag}(-(\eta_{h}+\mu_{h}),-(\eta_{h}+\mu_{h}+\sum_{l \neq 1}m_{1l}), -\mu_{l}^{1}(t), -\mu_{m}^{1}(t),..., \\
&\qquad\quad -(\eta_{h}+\mu_{h}),-(\eta_{h}+\mu_{h}+\sum_{l \neq n}m_{nl}), -\mu_{l}^{n}(t), -\mu_{m}^{n}(t)).
\end{eqnarray*}

Consider the system $\frac{dx}{dt}=-V(t)x(t)$ and let $\Phi(t,k)$ be its fundamental matrix. Hence, $\frac{d\Phi(t,k)}{dt}=-V(t)\Phi(t,k),\ \forall t \geq k,$ where $\Phi(k,k)=I,\ \forall k\in \mathbb{R}.$
In fact, $l(\Phi)<0,$ where $l(\Phi):=\inf \{\tilde{l}: \exists M\geq 1;\| \Phi(t,t-s)  \| \leq M e^{\tilde{l}s},\ \forall s \in \mathbb{R}, t \geq 0      \},$ since 
\begin{align}
\tiny
\Phi(t,k) = & \text{diag}\left( e^{-(\eta_{h}+\mu_{h})(t-k)}, e^{-(\eta_{h}+\mu_{h}+\sum_{l \neq 1}m_{1l})(t-k)}, e^{-\int_{k}^{t}\mu_{l}^{1}(t)dt}, e^{-\int_{k}^{t}\mu_{m}^{1}(t)dt}, \dots, \right. \nonumber \\
&\quad  \left. e^{-(\eta_{h}+\mu_{h})(t-k)},e^{-(\eta_{h}+\mu_{h}+\sum_{l \neq n}m_{nl})(t-k)}, e^{-\int_{k}^{t}\mu_{l}^{n}(t)dt}, e^{-\int_{k}^{t}\mu_{m}^{n}(t)dt} \right)
\end{align}
%XIAO MAIN PAPER
Consequently, analogous to \cite{xiaoMain}, consider the following linear operator $L$ on $\bar{C}_{w}$, the space of all continuous $w-$periodic functions from $\mathbb{R}$ to $\mathbb{R}^{4n}$:
$$ 
L(\varphi(.))(t):= \int_{0}^{\infty}\Phi(t,t-s)F(t-s)\varphi(t-s+ \cdot )ds.
$$ 
The basic reproduction number is defined as the spectral radius of $L(\varphi(\cdot))$ by 
\begin{equation}\label{R0}
\mathcal{R}_{0}:=r(L).
\end{equation}

%O FATO DE TER ESSE PONTO EH PQ t\geq 0 eh so um tempo fixado, e o ponto eh um elemento do dominio de \varphi. entao nos dois lados da eq de definicao de L temos funcoes de C_{w}.

\begin{lemma}\label{R0}
The linear operator $L:C_{w} \rightarrow C_{w}$ is compact.  $\mathcal{R}_{0}< \infty$ is an eigenvalue of $L$.
\end{lemma}
 
\begin{proof}
%FAZ SENTIDO DOMINIO DE L SER C_{w}? de fato range de F(t) eh em C_{w}
Let $B(0;1) \subset C_{w}$ be the unitary ball centered at the origin. For $\{ f_{n} \}_{n\geq 1} \subset B(0;1),$ we have 
$$
\| L(f_{n})(t)-L(f_{n})(t_{0}) \| \leq M_{n}\int_{0}^{\infty} e^{w(\Phi) s}(\|f_{n}(t)-f_{n}(t_{0})\|)ds
$$ 
with $M_{n}>0,$ since $F(t)$ is a linear bounded operator.
As $f_{n}$ is continuous, $$t\rightarrow t_{0} \implies \|L(f_{n})(t)-L(f_{n})(t_{0})\| \rightarrow 0 \ \forall \ t_{0} \in \mathbb{R}, n \in \mathbb{N}.$$ Therefore, $\{ f_{n} \}_{n\geq 1}$ is an equicontinuous sequence. By Ascoli-Arzel\`{a} Theorem, there exists a subsequence $\{ f_{n_{k}} \}_{k\geq 1}$ such that $L(f_{n_{k}})$ converges in the $C_{w}$ norm.
Since $\{ f_{n} \}_{n\geq 1} \subset B(0;1)$ is arbitrary, $L$ is a compact operator.
Thus, by Krein-Rutman Theorem, $\mathcal{R}_{0}$ is a positive eigenvalue of $L$.
\end{proof}

\begin{remark}
As noticed in systems (\ref{6})--(\ref{7}), the values of $\text{sgn}(\mathcal{R}_{m}^{i}-1)$, $i=1,...,n$, are of utmost importance for the basic reproduction number $\mathcal{R}_{0}.$ While defining $\mathcal{R}_{0}$, we considered that without loss of generality the assumptions $\mathcal{R}_{m}^{1},\ \mathcal{R}_{m}^{n}>1$ and $\mathcal{R}_{m}^{i_{j}}\leq 1$ for $j=1,...,k<n.$
\end{remark} 

Define $\Theta_{i}:=C([-\tau^{i}(0),0],\mathbb{R}_{+}^{2}) \times \mathbb{R}_{+} \times C([-\bar{\tau}^{i},0],\mathbb{R}_{+})$, where $\bar{\tau}^{i}:=\max\{ \tau_{l}^{i}(0), \tau_{h} \}$ and $\Theta:=\prod_{i=1}^{n}\Theta_{i}.$ Moreover, define $\Gamma_{i}=C([-\tilde{\tau},0], \mathbb{R}_{+}^{4})$ and $\Gamma=\prod_{i=1}^{n}\Gamma_{i}.$ 

\begin{lemma}\label{lema7}
Given $\theta^{i}=(\theta_{1}^{i},...,\theta_{4}^{i})\in \Theta_{i},$ then system (\ref{6})--(\ref{7}) admits a unique nonnegative solution $x(t,\theta)$, defined on $[0, \infty)$, such that $x(0, \theta)=\theta$, where $\theta=(\theta^{1},...,\theta^{n}) \in \Theta$.
\end{lemma}

\begin{proof}
Let $\hat{\tau}:=\min\{\tau_{h},\tau_{l}^{i}(0), \tau^{i}(0)\}_{i=1,...,n}$. Since $1-\dot{\tau}^{i}(t),\ 1-\dot{\tau_{l}}^{i}(t)\geq 0$, we have $-{\tau}^{i}(0)\leq \hat{\tau}-{\tau}^{i}(\hat{\tau}) \leq 0$ and $-{\tau_{l}}^{i}(0)\leq \hat{\tau}-{\tau_{l}}^{i}(\hat{\tau}) \leq 0, \ \forall 1\leq i \leq n.$ Thus, given $\theta \in \Theta,$ by the Variation of Parameters  Formula, the solution $x(t,\theta)=(x_{1}(t,\theta),...,x_{n}(t,\theta))$ with $x(0,\theta)=\theta$ exists and is unique on $[0,\hat{\tau}]$, since for $t\in [0,\hat{\tau}]$, $x_{m}^{i_{j}}(t-\tau^{i_{j}}(t))=\theta_{m}^{i_{j}}(t-\tau^{i_{j}}(t))$ for $m=1,2$ and $x_{4}^{i}(t-\tau^{i}(t))=\theta_{4}^{i}(t-\tau^{i}(t))$ for all $i=1,...,n,\ 1\leq j\leq k<n.$ Analogously, by taking $t \in [\hat{\tau},2\hat{\tau}],$ we can prove existence and uniqueness for the interval $[0,2\hat{\tau}].$ Thus, by induction, $\forall k \in \mathbb{N},$ the solution is unique on $[0,k\hat{\tau}]$. Hence, $x(t,\theta)$ is unique on $[0,\infty )$. The nonnegativity of $x(t,\theta)$ comes from the fact that if $x_{i}(t,\theta)=0$ for some $t> 0,$ then $\frac{dx_{i}(t,\theta)}{dt}\geq0,\ \forall i$.
\end{proof}

\begin{remark}
Let $\varphi=(\varphi_{1},\dots,\varphi_{n}) \in \Gamma$, $\theta=(\theta_{1},\dots,\theta_{n}) \in \Theta$, where $\theta_{i}=\varphi_{i}$ on the intervals where both are defined, for $i=1,\dots,n$. So, by Lemma \ref{lema7} and analogous to Lemma \ref{lemazimm}, the unique nonnegative solution maps $x(t,\varphi), y(t,\theta)$ of system (\ref{6})--(\ref{7}) satisfy $x(t,\varphi)=$ $y(t,\theta),\ \forall t\geq 0$.\end{remark}
%DEVO NALGUM MOMENTO PROVAR Q O MAPA DE POINCARE P_{\EPSILON} EH COMPACT LIN OPERATOR PRA O LINEARIZED INFECTIVE SYSTEM
%QUANDO FOR PROVAR O THM PRINCIPAL, NAO DEVO DIVIDIR (i) e (ii), ACHO QUE TUDO SERA UMA UNICA PROVA PARA O CASO EM QUE R_{0}<1 e R_{m}^{i_{k}}\geq 1, R_{m}^{j}>1 para os determinados patches no Lema 5.

Let $P(t)\theta:=x(t,\theta)$ be the solution map of system (\ref{6})--(\ref{7}), defined on $[0,\infty)$. Thus, analogous to Lemma 3.5 in \cite{lou2017}, the solution $P(t)$ is a $\omega-$periodic semiflow on $\Theta$. Define its Poincar\'e map $\mathcal{P}:=\mathcal{P}(w)$.
\begin{remark}

Given $t \geq 0,$ 
$\theta_1, \theta_2 \in \Theta,$ such that $\theta_1 \geq \theta_2$, then $\theta_1-\theta_2 \in \Theta$, and due to Lemma \ref{lema7}, 
 $\mathcal{P}(t)(\theta_1 ) - \mathcal{P}(t)(\theta_2 )=\mathcal{P}(t)(\theta_1 - \theta_2) \geq 0.$ 
 Thus, $\forall t\geq 0$, $\mathcal{P}(t)$ is monotone over $\Theta.$
\end{remark}

Consider $\tilde{\mathcal{P}}:\Gamma \rightarrow \Gamma, \ \mathcal{P}: \Theta \rightarrow \Theta$ are Poincar\'e maps of system (\ref{6})--(\ref{7}). Similar to Lemma \ref{lema9}, it follows that both are compact operators. In addition, by \cite{xiaoMain} and \cite{lou2017} %E O OUTRO ARTIGO Q FALA $r(\tilde{\mathcal{P}})=r(\mathcal{P})!!!!!!!!!!!!!!!!!!!!!!!$ 
we have the following conclusion.

\begin{lemma}
$r(\tilde{\mathcal{P}})=r(\mathcal{P}).$ Hence, ${\rm sgn}(\mathcal{R}_{0}-1)={\rm sgn}(r(\mathcal{P})-1).$
\end{lemma}

Given $\epsilon > 0,$ let us consider the following pertubation of system (\ref{6})--(\ref{7}).

For $\ j \in {1,...,k}$, where $1<i_{1}<...<i_{k}<n$ are such that $\mathcal{R}_{m}^{i_{j}}>1$:
\begin{equation}
\begin{aligned}[t]
\frac{dI_{h}^{{i}_{j}}}{dt} =&(1-a)\beta_{mh}b_{m}^{i_{j}}(t-\tau_{h}) ({S_{h}^{i_{j}}}^{*}+\epsilon) I_{m}^{i_{j}}(t-\tau_{h})e^{-\mu_{h}\tau_{h}}-(\eta_h + \mu_{h})I^{i_{j}}_{h}(t)\\
\frac{dA_{h}^{i_{j}}}{dt} =& a\beta_{mh}b_{m}^{i_{j}}(t-\tau_{h}) ({S_{h}^{i_{j}}}^{*}+\epsilon)I_{m}^{i_{j}}(t-\tau_{h})e^{-\mu_{h}\tau_{h}} - (\eta_{h} + \mu_{h}+\sum_{l \neq i_{j}} m_{i_{j}l})A_{h}^{i_{j}}(t)+ \sum_{l \neq i_{j}} m_{li_{j}}A_{h}^{l}(t) \\
\frac{dL_{I}^{i_{j}}}{dt} =&\mu^{i_{j}}_{b}(t)\left(1-\frac{{L_{s}^{i_{j}}}^{*}(t)-\epsilon}{K^{i_{j}}(t)}\right)pI_{m}^{i}(t)-\mu_{l}(t)L_{I}^{i_{j}}(t)\\
\frac{dI_{m}^{i_{j}}}{dt} =& \left(1-\dot{\tau}^{i_{j}}_{l}(t)\right)\mu_{b}^{i_{j}}(t-\tau^{i_{j}}_{l}(t))\left(1-\frac{{{L_{s}^{i_{j}}}^{*}}(t-\tau^{i_{j}}_{l}(t))-\epsilon}{K^{i_{j}}(t-\tau^{i_{j}}_{l}(t))}\right)pI_{m}^{i_{j}}(t-\tau^{i_{j}}_{l}(t))e^{-\int_{t-\tau^{i_{j}}_{l}(t)}^{t} \mu^{i_{j}}_{l}(s) \, ds } \\
&+ \left(1-\dot{\tau}^{i_{j}}(t)\right)\beta_{hm}b^{i_{j}}(t-\tau^{i_{j}}(t))(I_{h}^{i_{j}}(t-\tau^{i_{j}}(t))+A_{h}^{i_{j}}(t-\tau^{i_{j}}(t)))({S_{m}^{i_{j}}}^{*}(t-\tau^{i_{j}}(t))+\epsilon)e^{-\int_{t-\tau^{i_{j}}(t)}^{t} \mu^{i_{j}}_{l}(s) \, ds }\\
& -\mu_{m}^{i_{j}}(t) I_{m}^{i_{j}}(t), 
\end{aligned}\label{lastsystem}
\end{equation}
and for $1\leq i\leq n$ such that $\mathcal{R}_{m}^{i}\leq 1$:
\begin{equation}
\begin{aligned}[t]
\frac{dI_{h}^{i}}{dt} =&(1-a)\beta_{mh}b_{m}^{i}(t-\tau_{h}) ({S_{h}^{i}}^{*} +\epsilon)I_{m}^{i}(t-\tau_{h})e^{-\mu_{h}\tau_{h}}-(\eta_h + \mu_{h})I^{i}_{h}(t)\\
\frac{dA_{h}^{i}}{dt} =& a\beta_{mh}b_{m}^{i}(t-\tau_{h}) ({S_{h}^{i}}^{*}+\epsilon)I_{m}^{i}(t-\tau_{h})e^{-\mu_{h}\tau_{h}} - (\eta_{h} + \mu_{h}+\sum_{l \neq i} m_{il})A_{h}^{i}(t)+ \sum_{l \neq i} m_{li}A_{h}^{l}(t) \\
\frac{dL_{I}^{i}}{dt} =&\mu^{i}_{b}(t)pI_{m}^{i}(t)-\mu^{i}_{l}(t)L_{I}^{i}(t)\\
\frac{dI_{m}^{i}}{dt} =& \left(1-\dot{\tau}^{i}_{l}(t)\right)\mu_{b}^{i}(t-\tau^{i}_{l}(t))pI_{m}^{i}(t-\tau^{i}_{l}(t))e^{-\int_{t-\tau^{i}_{l}(t)}^{t} \mu_{l}^{i}(s) \, ds }   \\
& +\left(1-\dot{\tau}^{i_{j}}(t)\right)\beta_{hm}b^{i_{j}}(t-\tau^{i_{j}}(t))(I_{h}^{i_{j}}(t-\tau^{i_{j}}(t))+A_{h}^{i_{j}}(t-\tau^{i_{j}}(t))) \ \epsilon \ e^{-\int_{t-\tau^{i_{j}}(t)}^{t} \mu^{i_{j}}_{l}(s) \, ds }\\
&   -\mu_{m}^{i}(t) I_{m}^{i}(t).
\end{aligned}\label{lastsystem1}
\end{equation}
Analogous to Lemma \ref{lema7}, for any fixed $\theta \in \Theta,$ this system admits a unique nonnegative solution $x^{\epsilon}(t)$ in $[0,\infty),$ where $x^{\epsilon}(0)=\theta$.

Fix $\epsilon > 0$, $t \geq 0$ let $\hat{\Theta}:=\prod_{i=1}^{n}\hat{\Theta}_{i},$ where $\hat{\Theta}_{i}=C([-\tau^{i}(0),0],\mathbb{R}^{2}) \times \mathbb{R} \times C([-\bar{\tau}^{i},0],\mathbb{R})$ and consider the solution map $((\mathcal{P}_{\epsilon})_{2}(t),(\mathcal{P}_{\epsilon})_{3}(t),(\mathcal{P}_{\epsilon})_{4}(t),(\mathcal{P}_{\epsilon})_{8}(t))=\mathcal{P}_{\epsilon}(t):\hat{\Theta} \rightarrow \hat{\Theta},$ which is a linear operator and leaves the space $\Theta \subset \hat{\Theta}$ invariant. Let us prove the following lemma regarding the compactness of the linear operator $\mathcal{P}_{\epsilon}.$ 

\begin{lemma}\label{lema9}
Fix $t \geq 0,\ \epsilon >0$. Then the linear operator $\mathcal{P}_{\epsilon}(t)$ is compact.
\end{lemma}

\begin{proof}
Let $\epsilon\geq 0$, and consider the solution map of system (\ref{lastsystem})--(\ref{lastsystem1}) $x^{\epsilon} \colon \mathbb{R}_{+}\times \hat{\Theta}\rightarrow \hat{\Theta} $, given by $x^{\epsilon}(t,\theta)=\mathcal{P}_{\epsilon}(t)(\theta).$ Let $\mathcal{F}_{\epsilon}=\{\mathcal{P}_{\epsilon}(t)\}_{t \in \mathbb{R}^{+}}$ be a family of continuous linear operators. Given $\theta \in \hat{\Theta},$ similar to the proof of Lemma \ref{lemazimm}, it is clear that $\{\mathcal{P}_{\epsilon}(t)(\theta)\}_{t \in \mathbb{R}^{+}}$ is a bounded set in $\hat{\Theta}$. Therefore, by the Uniform Boundedness Principle, $M:=\sup_{t\in \mathbb{R}_{+}, \norm{\theta}\leq 1}\norm{\mathcal{P}_{\epsilon}(t)(\theta)}_{\hat{\Theta}} <\infty$. Now, it is clear that $\mathcal{F}_{\epsilon}$ is an equicontinuous family, since given $\varepsilon>0,$ $$\norm{\mathcal{P}_{\epsilon}(t)(\theta_{1})-\mathcal{P}_{\epsilon}(t)(\theta_{2})}\leq M\norm{\theta_{1}-\theta_{2}}<\varepsilon,$$ if $\norm{\theta_{1}-\theta_{2}}<\varepsilon /M.$
Lastly, by Ascoli-Arzel\`{a} Theorem, for $t\geq 0$, $\mathcal{P}_{\epsilon}(t)$ is a compact linear operator.
\end{proof}

Now we are in position to prove the Threshold Dynamics Theorem, using arguments presented in \cite{xiaoMain}. 
Let $\tilde{\tau}^{i}:=\max \{\tau_{h}, \tau^{i}(0), \tau_{l}^{i}(0)\}$, $\hat{\tau}^{i}:=\max \{\tau^{i}(0), \tau_{l}^{i}(0)\}$, and let $t \geq 0$. Therefore, we define
\begin{align*}
\Delta(t) := \Bigg\{ & \varphi \in \prod_{i=1}^{n}\mathbb{R}_{+}\times C([-\tau^{i}(0),0],\mathbb{R}_{+}^{2})\times C([-\tau_{l}^{i}(0),0],\mathbb{R}_{+}^{2}) \\
& \times C([-\hat{\tau}^{i},0],\mathbb{R}_{+}) \times \mathbb{R}_{+} \times  C([-\tilde{\tau}^{i},0],\mathbb{R}_{+}) \;|\; \varphi_{4}^{i}(s)+\varphi_{5}^{i}(s)\leq K^{i}(t+s), \forall s\in [-\tilde{\tau}^{i},0], t \in \mathbb{R}_{+}, \\ & \varphi_{7}^{i}(0)=\int_{-\tau^{i}(0)}^{0}b_{m}^{i}(x)\beta_{hm}S_{m}^{i}(x) (I_{h}^{i}(x)+A_{h}^{i}(x))dx
\Bigg\}.
\end{align*}%este final eh para match o initial value problem de E_{m} que eh E_{m}(0) com o que dever ser a solucao em t=0.
%a outra restricao dos varphi_{4,5} eh pra impose q nem o IVP pode superar a K^{i}

\begin{theorem}\label{principalthm} 
If $\mathcal{R}_{0}< 1$, then $\gamma(t)=(x_{1},...,y_{i_{1}}(t),...,y_{i_{k}}(t),...,x_{n})$, a periodic solution of system (\ref{0})--(\ref{0-1}) without the $E_h^i$ and $R_h^i$ equations, is globally stable in $\Delta(0)-\{(x_{1},...,x_{n})  \} $. If $\mathcal{R}_{0}> 1,$ then there exists $\delta>0$ such that  
$$ 
\forall \varphi \in \{ v \in  \Delta(0):\ v_l^{i_j}(0)>0, \forall \ 1\leq j  \leq k, \ l\ =\ 2,3,4,8 \},
$$ 
we have 
$$ 
\liminf_{t\to\infty}(I_{h}^{i_j}(t), A_{h}^{i_j}(t),L_{I}^{i_j}(t),I_{m}^{i_j}(t))\geq (\delta, \delta, \delta, \delta), \forall \ 1\leq j\leq k,
$$

where $(S_h^1(t), \dots , I_m^1(t), \dots, S_h^n(t), \dots, I_m^n(t) ) = x(t, \varphi)$ is the unique solution of system (\ref{0}). 
\end{theorem}

\begin{proof}
Let $\mathcal{R}_{0}<1,\ \epsilon > 0,$  $\tilde{P}\colon\mathbb{R}_{+} \times \Delta(0)\rightarrow \Delta(t)$ and ${P}_{\epsilon} \colon \mathbb{R}_{+} \times \hat{\Theta}^{+}\rightarrow \hat{\Theta}^{+}.$ According to Lemma \ref{lema9}, $\forall t\geq 0,$ ${P}_{\epsilon}(t)$ is a compact operator. Defining ${P}_{\epsilon}:={P}_{\epsilon}(w),$ by Krein-Rutman Theorem we know that $r({P}_{\epsilon})<\infty$ is an eigenvalue of $P_{\epsilon}$ whose eigenvector lies on $\hat{\Theta}^{+}\setminus \{0\}.$ %O EIGENVECTOR NAO PRECISA SER POSITIVO EM TDS AS COORD. ISSO SO GARANTIRIA Q A CURVA PERIODICA v(t) DO LEMMA 1 DESSE PAPER Q MENCIONAREI A SEGUIR SEJA POSITIVA EM TODAS AS COORD TB, OQ N EH NECESSIDADE

As $\forall \epsilon \geq 0$, $P_{\epsilon}$ is a compact operator, by the continuity of the spectral radius of compact operators and by \cite{xiaoMain}, take $\epsilon>0$ small enough such that $r(P_{\epsilon})<1.$ By Lemma 1 of \cite{wang2017dynamics}, given $\varphi \in \hat{\Theta}^{+},$ we have 
\begin{equation}\label{epsilon}
P_{\epsilon}(t,\varphi)\rightarrow 0\in \hat{\Theta}^{+} \text {as } t\rightarrow \infty.
\end{equation}

On the other hand, in system (\ref{0})--(\ref{0-1}) without the $E_h^i$ and $R_h^i$ equations define the change of variables $\mathcal{M}^{i}(t)=(S_{m}^{i}+E_{m}^{i}+I_{m}^{i})(t),$ $L^{i}(t)=(L^{i}_{s}+L^{i}_{I})(t)$, that will produce the following system: \begin{equation}
\begin{aligned}[t]
\frac{{dL}^{i}}{dt} &= \mu^{i}_{b}(t)\left(1-\frac{L^{i}(t)}{K^{i}(t)}\right)\mathcal{M}^{i}(t)-\left(\eta^{i}_{l}(t)+\mu^{i}_{l}(t)\right)L^{i}(t) \\
\frac{d\mathcal{M}^{i}}{dt} &=((1-\dot{\tau}^{i}_{l}(t))\mu^{i}_{b}(t-\tau^{i}_{l}(t))\left(1-\frac{L^{i}(t-\tau^{i}_{l}(t))}{K^{i}(t-\tau^{i}_{l}(t))}\right) \\
&\quad \times \left(\mathcal{M}^{i}(t-\tau^{i}_{l}(t))\right)e^{-\int_{t-\tau_{l}(t)}^{t} \mu_{l}^{i}(s) \, ds } - \mu_{m}^{i}(t)\mathcal{M}^{i}(t),
\end{aligned}\label{SISTEMA}
\end{equation}which is topologically equivalent to system (\ref{3}). Thus, similar to Lemma \ref{lemmaRm}, as $t \rightarrow \infty$, $$(\mathcal{M}^{1}(t),L^{1}(t),\dots,\mathcal{M}^{n}(t),L^{n}(t)) \rightarrow \tilde{\gamma}=(\tilde{\gamma}^{1},\dots, \tilde{\gamma}^{n}),$$for some $w$--periodic curve $s\mapsto \tilde{\gamma}(s)$, where $({\mathcal{M}^{i}}^{*}(s), {L^{i}}^{*}(s)):=\tilde{\gamma}^{i}(s), \ \forall 1\leq i\leq n.$

Therefore, $\forall 1\leq i \leq n, \ \phi \in \Delta (0),$ as $t\rightarrow \infty$ we have 
$$
\left\|\left(\sum_{k=6}^{8}{\tilde{P}_{k}^{i}(t,\phi)}, \sum_{k=4}^{5}{\tilde{P}_{k}^{i}(t,\phi)} \right) - \tilde{\gamma}^{i}\right\| \rightarrow (0,0).
$$
Notice in system (\ref{0})--(\ref{0-1}) without the $E_h^i$ and $R_h^i$ equations that the disease free plane is invariant under the solution map $\tilde{P}^{i}$. Consider the natural inclusion $\iota_t:{\mathcal{I}^{i}_{2}}^{+} \hookrightarrow \Delta(t)$ as well as the natural projection $\pi: \Delta(0) \rightarrow {\mathcal{I}_{2}^{i}}^{+}$, $\iota:=\iota_0.$ Therefore, $\forall \ \phi \in \Delta(0),$ 
$$
\tilde{P}^{i}(t,\iota(\pi(\phi)))=\iota_t (\bar{P}^{i}(t, \pi(\phi))).
$$

Now suppose without loss of generality that $\iota(\bar{P}^{i}(t, \pi(\phi)))=\bar{P}^{i}(t, \pi(\phi)),$ i.e, ${\mathcal{I}_{2}^{i}}^{+}$ is naturally included in $\Delta(t), \ \forall t \geq 0.$ In fact, as $t \rightarrow \infty$, we have that $\forall \phi \in \Delta(0)$,
\[
\norm{\mathcal{M}^{i}(t,\phi)-\tilde{\gamma}^{i}} = \left\|{\sum_{k=6}^{8}\tilde{P}_{k}^{i}(t,\phi)-\tilde{\gamma}^{i}} \right\| \rightarrow 0,
\]
and thus in particular as $t \rightarrow \infty$, 
\[
\norm{\mathcal{M}^{i}(t,\pi(\phi))-\tilde{\gamma}^{i}} \rightarrow 0.
\]However, due to the invariance of the disease free plane and uniqueness of solutions, $\mathcal{M}^{i}(t,\iota(\pi(\phi)))=\bar{P}_{6}^{i}(t,\pi(\phi))$ and, consequently, we obtain 
\[
\norm{\bar{P}^{i}(t,\pi(\phi))-\tilde{\gamma}^{i}} \rightarrow 0 \quad \text{as} \quad t \rightarrow \infty.
\]
Therefore, due to Lemma \ref{lemmaRm}, $({\mathcal{M}^{i_{j}}}^{*}(s),{L^{i_{j}}}^{*}(s))=({L_{s}^{i_{j}}}^{*}(s),{S_{m}^{i_{j}}}^{*}(s))$ for all $s \geq 0$, $1 \leq j \leq k < n$, and $({\mathcal{M}^{i}}^{*}(s),{L^{i}}^{*}(s))=(0,0)$ for all $s \geq 0$, where $i \in \{1,\dots,n\} - \{i_{1},\dots,i_{k}\}$.

In particular, $\forall  \epsilon>0,$ there exists $t>0$ sufficiently large such that $\forall 1 \leq j \leq k,$ $$ \tilde{P}^{i_{j}}_{6}(t,\iota(\pi(\phi)))\leq \epsilon + {S_{m}^{i_{j}}}^{*}(t),$$ 
$$ 
{L_{s}^{i_{j}}}^{*}(t)-\epsilon  \leq \tilde{P}^{i_{j}}_{4}(t,\iota(\pi(\phi)))+\tilde{P}^{i_{j}}_{5}(t,\iota(\pi(\phi)))\leq {L_{s}^{i_{j}}}^{*}(t)+\epsilon,
$$
and $\forall i \in \{1,\dots,n   \}- \{  i_{1},\dots, i_{k}\},\ \text{and}\ m=4,5,6, $ 
$$ 
\norm{\tilde{P}^{i}_{m}(t,\iota(\pi(\phi)))}\leq \epsilon.
$$ 
Now, for $t>0$ sufficiently large we obtain 
\begin{align*}
\frac{d\tilde{P}^{i_{j}}_{m}(t,\phi)}{dt} &\leq \frac{d(P_{\epsilon}^{i_{j}})_{m}(t,\bar{\iota}(\bar{\pi}(\phi)))}{dt},\  1 \leq j \leq k,\ m=2, 4,8, \\
\frac{d\tilde{P}^{i}_{m}(t,\phi)}{dt} &\leq \frac{d(P_{\epsilon}^{i})_{m}(t,\bar{\iota}(\bar{\pi}(\phi)))}{dt},\  i \in \{1, \dots, n\} - \{i_{1}, \dots, i_{k}\},\ m=2,4,8,
\end{align*}
where ${\bar{\pi}}(\phi)=(\phi^{i}_{2},\phi^{i}_{3},\phi^{i}_{4},\phi^{i}_{8})_{i=1,\dots,n}$ and $\bar{\iota}({\bar{\pi}}(\phi))=(\phi_{2}^{i},\phi_{3}^{i},\phi_{4}^{i},\phi_{8}^{i})_{i=1,\dots,n}\in \hat{\Theta},$ in which  $\bar{\pi} \colon \Delta(0)\rightarrow \bar{\pi}(\Delta(0)),\ \bar{\iota} \colon \bar{\pi}(\Delta(0))\rightarrow \hat{\Theta}$ are natural projection and natural inclusion, respectively. By uniqueness of solutions of systems (\ref{0})--(\ref{0-1}) without the $E_h^i$ and $R_h^i$ equations and (\ref{lastsystem})--(\ref{lastsystem1}), we conclude by (\ref{epsilon}) that there exists \( t_{0} > 0 \) such that
\[
t \geq t_{0} \implies \tilde{P}^{i_{j}}_{m}(t,\phi) \leq (P_{\epsilon})_{m}^{i_{j}}(t,\bar{\iota}(\bar{\pi}(\phi))) \implies \lim_{t \to \infty} \tilde{P}^{i_{j}}_{m}(t,\phi) = 0 \quad  \ 1 \leq j \leq k, \ m=2,4,8,
\]and for all \( i \in \{1,\dots,n\} - \{i_{1},\dots,i_{k}\} \),\  \begin{equation}
    t \geq t_{0} \implies \tilde{P}^{i}_{m}(t,\phi) \leq (P_{\epsilon})_{m}^{i}(t,\bar{\iota}(\bar{\pi}(\phi))) \implies \lim_{t \to \infty} \tilde{P}^{i}_{m}(t,\phi) = 0,\ m=2,4,8. \label{man}\end{equation} Similarly, we can conclude that \begin{equation}
    t \geq t_{0} \implies \sum_i\tilde{P}^{i}_{3}(t,\phi) \leq \sum_i(P_{\epsilon})_{3}^{i}(t,\bar{\iota}(\bar{\pi}(\phi))) \implies \lim_{t \to \infty} \sum_i\tilde{P}^{i}_{3}(t,\phi) = 0\implies \lim_{t \to \infty} \tilde{P}^{i}_{3}(t,\phi) = 0, \ \forall \ i=1,\dots,n. \label{manAh}\end{equation}

This concludes the proof for the case $\mathcal{R}_{0}< 1,$ as $\phi \in \Delta(0)$ was chosen arbitrarily.

Let $\mathcal{R}_{0}>1,\ \phi \in \Delta(0)$. Analogous to the case $\mathcal{R}_{0}<1,$ we have
$$
\frac{d\tilde{P}^{i}_{m}(t,\phi)}{dt} \leq \frac{d(P_{\epsilon}^{i})_{m}(t,\bar{i}(\bar{\pi}(\phi)))}{dt},\ \ \forall 1 \leq i \leq n,\ m=2,3,4,8, \ \forall \ t>0\ \text{sufficiently large},
$$ 
and $\tilde{P}^{i}_{m}(t,\phi)$ are ultimately bounded for $m=1,5,6,\ \forall 1\leq i \leq n$, particularly due to the arguments established by system (\ref{SISTEMA}) in the case $\mathcal{R}_{0}<1$.%JA QUE ESSES COMPARTIMENTOS IRAO SE APROXIMAR DA ORBITA PERIODICA Q EH UNIFORMLY BOUNDED.

%EVERY COMPACT OPERATOR IS BOUNDED: https://math.stackexchange.com/questions/2802862/show-that-a-compact-operator-is-bounded

According to Lemma \ref{lema9}, the operator $P_{\epsilon}$ is compact. Hence, it is a bounded operator. Therefore, by the above inequality and the Comparison Principle, we find that the solution map $\tilde{P}$ is ultimately bounded, and thus it is point dissipative. Due to Theorem 3.6.1 of \cite{hale1993introduction}, the compositions of the Poincar\'e map, defined in $\Delta(0)$,  
$$
P^{(n)}(\phi):=\tilde{P}(\phi,nw)
$$
are compact maps, $\forall n\in\mathbb{N}$ sufficiently large. By Theorem 1.1.3 from \cite{zhao2017}, $P:=P^{(1)}$ has a strong global attractor in $\Delta(0).$ Furthermore, analogous to Lemma \ref{precompact}, every orbit is precompact. Henceforth, due to Lemma 1.2.1 in \cite{zhao2017}, the $\omega-$limit set of any orbit is an internally transitive chain.

Define 
$$
\mathcal{S}_{0}=\{ v\in \Delta(0): v^{i}_{l}(0)>0,\ \forall \ 1\leq i\leq n,\ l=2,3,4,8     \},
$$
and therefore 
$$
\partial \mathcal{S}_{0}=\{ v\in \Delta(0): v^{i}_{l}(0)=0,\ \exists \  1\leq i\leq n,\ l=2,3,4,8   \}. 
$$
Also define 
$$
\mathcal{M}_{\partial}:=\{ \phi \in  \partial\mathcal{S}_{0}: P^{(n)}(\phi)\in \partial\mathcal{S}_{0},\ \forall n \in \mathbb{N} \}.
$$
We should prove that 
\begin{equation}
\label{ultimamah}  
\{  \gamma(t):t\in[0,w]\}\supset\bigcup_{x\in \mathcal{M}_{\partial} }\omega(x).
\end{equation} 
Given $x\in \mathcal{M}_{\partial},$ $\omega(x)$ is invariant under $P$ and thus $P^{(n)}(x),\ \forall n>0.$ Consequently,
$$
\mathcal{M}_{\partial}=\bigcap_{n\in \mathbb{N}}(P^{(n)})^{-1}(\partial \mathcal{S}_{0}),
$$
which implies that $\mathcal{M}_{\partial} $ is a closed set in $\Delta(0)$. We obtain $\omega(x) \subset \mathcal{M}_{\partial},$ and hence, 
$$
\bigcup_{x\in \mathcal{M}_{\partial} }\omega(x)\subset\mathcal{M}_{\partial}. 
$$ 

Define a set 
$$
\mathcal{A}:=\{\phi \in \partial\mathcal{S}_{0}:\phi_{l}^{j}=0,\ \forall \ 1\leq j\leq n, \ l=2,3,8      \}.
$$
We claim that $\mathcal{M}_{\partial}\subset \mathcal{A}.$ If $\phi\in \partial \mathcal{S}_{0}$ is such that $\phi^{i_j}_{2}(0)>0,\ \exists \ 1\leq j \leq k,$ then by analyzing system  (\ref{0})--(\ref{0-1}) without the $E_h^i$ and $R_h^i$ equations, there exists $t_{0}>0$ such that $\tilde{P}^{i_j}_{8}(t,\phi)>0,\ \forall \ t\geq t_{0},$ and hence there exists a $t_1 > 0$ such that $\tilde{P}^{i_j}_{2}(t,\phi)>0,\ \forall t\geq t_{1},$  and there exists $t_2 >0$ such that $\tilde{P}^{i_j}_{3}(t,\phi)>0$ for all $t \geq t_2$.

If $\exists  \phi \in \partial \mathcal{S}_{0}$ where $\phi^{i_j}_{l}(0)>0$ for $  l=3,  4,  8$ and  $1$ $\leq j \leq k$  then the result is analogous. In any of these cases, we conclude that $\phi \notin \mathcal{M}_{\partial}.$ This proves the claim. Particularly, we find that 
$$
\bigcup_{x\in \mathcal{M}_{\partial} }\omega(x)\subset\bigcup_{x\in \mathcal{A} }\omega(x).
$$ 
%REESCREVER ISSO E VER SE TA CERTO: 
Now given $x \in \mathcal{A}$ and $\phi \in \omega(x),$ we find that $\phi^{i_j}_{l}=0,\ \forall 1\leq j\leq k,\ l = 2,3,4,8$. Therefore, by Lemma \ref{lemmaRm}, we have that $(\tilde{P}_{1}^{i_{j}}(t,\phi),\dots,,\tilde{P}_{8}^{i_{j}}(t,\phi))=y_{i_{j}}(t),$ and $ (\tilde{P}_{1}^{l}(t,\phi),\dots,\tilde{P}_{8}^{l}(t,\phi))=x_{l}(t),$ $\forall \ 1\leq j\leq k,\ l \in \{ 1,\dots, n \}-\{ i_{1},\dots, i_{k}  \}$. Thus, $\omega(x)\subset \{\gamma(t):t\in[0,w]  \}.$ Finally, we have 
$$\bigcup_{x\in \mathcal{A} }\omega(x) \subset \{\gamma(t):t\in[0,w]  \},
$$ 
which proves (\ref{ultimamah}).
     
%there is no cycle for \gamma(t) in \partial\mathcall{S}_{0} b/c this would require that W^{U}({\gamma}) is non empty in \partial\mathcall{S}_{0}. IMPOSSIBLE!!
%TENHO Q REVER DAQ PRA BAIXO!

In the following, we prove that $W^{s}(\{\gamma(t): t\in [0,w]  \})\cap \mathcal{S}_{0}=\emptyset$ and $\{\gamma(t): t\in [0,w]  \}$ is an isolated invariant set for $P$ in $\Delta(0).$ Suppose for every $\varepsilon>0$, there exist a $\varphi\in \mathcal{S}_{0}$ and a sufficiently large $N \in \mathbb{N}$ such that
$$
\norm{P^{(n)}\varphi - \gamma(t)}<\varepsilon,\ \forall n\geq N.
$$
Due to the continuity and boundedness of $\tilde{P}$, and invariance of $\{\gamma(t): t \in [0,w]   \}$, we can consider $\varphi\in \mathcal{S}_{0}$ such that for $\ n\in \mathbb{N}$ sufficiently large, we have
$$
\norm{ \tilde{P}(\hat{t}+nw,\varphi) - \tilde{P}(\hat{t},\gamma({t}))}\leq \norm{\tilde{P}(\hat{t})}\norm{P^{(n)}\varphi - \gamma(t)} <\varepsilon, \ \forall \hat{t}\in [0,w]. 
$$

Therefore, for $1\leq j \leq k,$ and for $t>0$ sufficiently large we obtain that
$$ \norm{{S_{h}^{i_{j}}}^{*} - \tilde{P}_{1}^{i_{j}}(t,\varphi)}<\varepsilon$$
$$ \norm{{S_{m}^{i_{j}}}^{*} - \tilde{P}_{6}^{i_{j}}(t,\varphi)}<\varepsilon \ \text{and} $$
$$  \norm{{L_{s}^{i_{j}}}^{*} - \tilde{P}_{5}^{i_{j}}(t,\varphi)}<\varepsilon. 
$$ 

Thus, similar to system (\ref{lastsystem})--(\ref{lastsystem1}), we have the system for $\ j \in \{1,...,k\}$, where $1<i_{1}<...<i_{k}<n$ are such that $\mathcal{R}_{m}^{i_{j}}>1$:
\begin{equation}
\begin{aligned}[t]
\frac{dI_{h}^{{i}_{j}}}{dt} =&(1-a)\beta_{mh}b_{m}^{i_{j}}(t-\tau_{h}) ({S_{h}^{i_{j}}}^{*}-\varepsilon) I_{m}^{i_{j}}(t-\tau_{h})e^{-\mu_{h}\tau_{h}}-(\eta_h + \mu_{h})I^{i_{j}}_{h}(t)\\
\frac{dA_{h}^{i_{j}}}{dt} =& a\beta_{mh}b_{m}^{i_{j}}(t-\tau_{h}) ({S_{h}^{i_{j}}}^{*}-\varepsilon)I_{m}^{i_{j}}(t-\tau_{h})e^{-\mu_{h}\tau_{h}} - (\eta_{h} + \mu_{h}+\sum_{l \neq i_{j}} m_{i_{j}l})A_{h}^{i_{j}}(t)+ \sum_{l \neq i_{j}} m_{li_{j}}A_{h}^{l}(t) \\
\frac{dL_{I}^{i_{j}}}{dt} =& \mu^{i_{j}}_{b}(t)\left(1-\frac{{L_{s}^{i_{j}}}^{*}(t)+\varepsilon}{K^{i_{j}}(t)}\right)pI_{m}^{i}(t)-\mu_{l}(t)L_{I}^{i_{j}}(t)\\
\frac{dI_{m}^{i_{j}}}{dt} =& \left(1-\dot{\tau}^{i_{j}}_{l}(t)\right)\mu_{b}^{i_{j}}(t-\tau^{i_{j}}_{l}(t))\left(1-\frac{{{L_{s}^{i_{j}}}^{*}}(t-\tau^{i_{j}}_{l}(t))+\varepsilon}{K^{i_{j}}(t-\tau^{i_{j}}_{l}(t))}\right)pI_{m}^{i_{j}}(t-\tau^{i_{j}}_{l}(t))e^{-\int_{t-\tau^{i_{j}}_{l}(t)}^{t} \mu^{i_{j}}_{l}(s) \, ds } \\
&+ \left(1-\dot{\tau}^{i_{j}}(t)\right)c\beta_{hm}b^{i_{j}}(t-\tau^{i_{j}}(t))(I_{h}^{i_{j}}(t-\tau^{i_{j}}(t))+A_{h}^{i_{j}}(t-\tau^{i_{j}}(t)))({S_{m}^{i_{j}}}^{*}(t-\tau^{i_{j}}(t))-\varepsilon)e^{-\int_{t-\tau^{i_{j}}(t)}^{t} \mu^{i_{j}}_{l}(s) \, ds } \\
& -\mu_{m}^{i_{j}}(t) I_{m}^{i_{j}}(t).
\end{aligned}
\end{equation}
Hence we have the inequalities 
\begin{equation}
\begin{aligned}[t]
\frac{d\tilde{P}_{2}^{i_{j}}}{dt}(t,\varphi) \geq& (1-a)\beta_{mh}b_{m}^{i_{j}}(t-\tau_{h}) ({S_{h}^{i_{j}}}^{*}-\varepsilon) I_{m}^{i_{j}}(t-\tau_{h})e^{-\mu_{h}\tau_{h}}-(\eta_h + \mu_{h})I^{i_{j}}_{h}(t)\\
\frac{d\tilde{P}_{3}^{i_{j}}}{dt}(t,\varphi) \geq& a\beta_{mh}b_{m}^{i_{j}}(t-\tau_{h}) ({S_{h}^{i_{j}}}^{*}-\varepsilon)I_{m}^{i_{j}}(t-\tau_{h})e^{-\mu_{h}\tau_{h}} - (\eta_{h} + \mu_{h}+\sum_{l \neq i_{j}} m_{i_{j}l})A_{h}^{i_{j}}(t)+ \sum_{l \neq i_{j}} m_{li_{j}}A_{h}^{l}(t) \\
\frac{d\tilde{P}_{4}^{i_{j}}}{dt}(t,\varphi) \geq& \mu^{i_{j}}_{b}(t)\left(1-\frac{{L_{s}^{i_{j}}}^{*}(t)+\varepsilon}{K^{i_{j}}(t)}\right)pI_{m}^{i}(t)-\mu_{l}(t)L_{I}^{i_{j}}(t)\\
\frac{d\tilde{P}_{8}^{i_{j}}}{dt}(t,\varphi) \geq& \left(1-\dot{\tau}^{i_{j}}_{l}(t)\right)\mu_{b}^{i_{j}}(t-\tau^{i_{j}}_{l}(t))\left(1-\frac{{{L_{s}^{i_{j}}}^{*}}(t-\tau^{i_{j}}_{l}(t))+\varepsilon}{K^{i_{j}}(t-\tau^{i_{j}}_{l}(t))}\right)pI_{m}^{i_{j}}(t-\tau^{i_{j}}_{l}(t))e^{-\int_{t-\tau^{i_{j}}_{l}(t)}^{t} \mu^{i_{j}}_{l}(s) \, ds } \\
&+ \left(1-\dot{\tau}^{i_{j}}(t)\right)c\beta_{hm}b^{i_{j}}(t-\tau^{i_{j}}(t))(I_{h}^{i_{j}}(t-\tau^{i_{j}}(t))+A_{h}^{i_{j}}(t-\tau^{i_{j}}(t)))({S_{m}^{i_{j}}}^{*}(t-\tau^{i_{j}}(t))-\varepsilon)e^{-\int_{t-\tau^{i_{j}}(t)}^{t} \mu^{i_{j}}_{l}(s) \, ds } \\
&-\mu_{m}^{i_{j}}(t) I_{m}^{i_{j}}(t),
\end{aligned}
\end{equation}
for $t>0$ sufficiently large. Thus, similar to the case $\mathcal{R}_{0}<1,$ due to the Comparison Principle and Lemma 1 from \cite{wang2017dynamics}, we conclude that for $\mathcal{R}_{0}> 1,$ 
$$ 
\tilde{P}_{l}^{i_{j}}(t,\varphi)\rightarrow\infty \ \text{as}\ t\rightarrow \infty \ \text{for } l\text{ = 2, 3, 4, 8 and }1\leq j \leq k,
$$ 
that is a contradiction to the boundedness of the solution map $\tilde{P}$ for sufficiently small $\varepsilon>0$. Hence, there exists $\varepsilon>0$ such that 
$$
\limsup_{n\to\infty}\norm{P^{(n)}\varphi - \gamma(t)}\geq \varepsilon,\ \forall \varphi \in \mathcal{S}_{0},\ t\in [0,w].
$$

Consequently, $\{ \gamma(t):t \in [0,w]   \}$ is isolated for $P$ in $\Delta(0)$ and $W^{s}(\{\gamma(t): t\in [0,w]  \})\cap \mathcal{S}_{0}=\emptyset$. Thus, as $\{\gamma(t):t\in [0,w] \}$ forms no cycle in $\partial\mathcal{S}_{0},$ by Theorem 1.3.1 from \cite{zhao2017}, $P$ is uniformly persistent with respect to $(\mathcal{S}_{0},\partial\mathcal{S}_{0})$. %TEM Q PROVAR ISSO DE SER ISOLADA E N FORMAR CYCLES..O LANCE DE N FORMAR CYCLES EH COMPREENSIVEL, PQ SE FORMASSE, O UNSTABLE MANIFODS DE \gamma(t) SERIA NAO NULO, O QUE IMPLICARIA Q NO SISTEMA (14) O UNSTABLE MANIFOLD SERIA NAO NULO TBM, OQ EH CONTRADICAO AO Rm LEMMA. 
By Theorem 3.6.1 in \cite{hale1993introduction}, for each fixed $t,$ there exists a $t_0 > 0$ sufficiently large such that the map $\tilde{P}(t)$ is a $k-$contraction and a compact map for $t \geq t_0$. Due to Theorem 4.5 in \cite{magal2005global}, by letting $\rho(x):=d(x,\partial \mathcal{S}_{0})$, it follows that $P|_{\partial\mathcal{S}_{0}}: \mathcal{S}_{0}\rightarrow \mathcal{S}_{0}$ has a compact invariant global attractor $S\subset \mathcal{S}_{0}.$
    
Let us define $\mathcal{J}:=[t_{0},t_{0}+w]\times S$. Therefore, $\mathcal{J}$ is compact in $\mathbb{R}_{+}\times \mathcal{S}_{0}$, and since $\tilde{P}(t)$ is compact operator by Theorem 3.6.1 in \cite{hale1993introduction}, $\text{for }t>t_{0},$ we obtain that $B:=\cup_{t\in [t_{0},t_{0}+w]}\tilde{P}(t,S)$ is compact. Hence, 
$$
\lim_{t\rightarrow \infty}d(\tilde{P}(t)\varphi,C)=0,\ \forall \ \varphi \in \mathcal{S}_{0}.
$$
Define a continuous function $p \colon \Delta(0)\rightarrow \mathbb{R}_{+}$ by $p(\varphi)=\min_{\ 1\leq i\leq n,\ j=2,3,4,8}(\varphi^{i}_{j}(0)).$ Then $p$ is continuous and thus $\exists \phi \in B$ such that $\inf_{\varphi \in B}p(\varphi)=\phi>0.$ Consequently, there exists $\delta>0$ such that in system (\ref{0})--(\ref{0-1}) without the $E_h^i$ and $R_h^i$ equations, we find 
$$
\liminf_{t\to\infty}\tilde{P}^{{i}_{j}}_{l}(t,\varphi)\geq \delta>0,\  \forall \varphi \in \Delta(0),\ l=2,3,4,8,\ 1\leq j\leq k.
$$
This completes the proof.    
\end{proof} 

%IMPORTANTE, NAO APAGA!!!!!!!!!!!!!!!: Now let $R_{0}=1.$ We will prove that there will exist orbits uniformly persistent while others belong to the stable manifold of $\{ \gamma(t): t\in [0,w]  \}$. We will consider a sequence of $r(P_{\epsilon_{n}})$ where $\epsilon_{n}\rightarrow 0$ and some $\epsilon_{k}>0$ while others $\epsilon_{j}<0,$ and create two systems which bounds the original systems from above and below, made from the subsequences of positive values and negative values $\epsilon_{k_{n}},\ \epsilon_{j_{n}},$ respectively, both approaching zero.

%IMPORTANTE, NAO APAGA!!!!!!!!!!!!!!!: Now let $R_{0}=1.$ We will prove that there will exist orbits uniformly persistent while others belong to the stable manifold of $\{ \gamma(t): t\in [0,w]  \}$. We will consider a sequence of $r(P_{\epsilon_{n}})$ where $\epsilon_{n}\rightarrow 0$ and some $\epsilon_{k}>0$ while others $\epsilon_{j}<0,$ and create two systems which bounds the original systems from above and below, made from the subsequences of positive values and negative values $\epsilon_{k_{n}},\ \epsilon_{j_{n}},$ respectively, both approaching zero.

\section{Numerical Analysis}
\setcounter{equation}{0}\setcounter{figure}{0}

\subsection{Data Collection and Parameters}
In this analysis, we focus on those states in Brazil where CHIKV is endemic and where at least one of its cities has direct flights to Miami. We consider the union of the endemic Brazilian regions as one patch. The total population of Miami is denoted as $N_{h0M}$, while the combined population of these Brazilian states is represented as $N_{h0B}$. The Brazilian states considered in the present work are Ceara, Pernambuco, Bahia, Amazonas, Para, Rio de Janeiro, Sao Paulo, Distrito Federal and Minas Gerais. The number of infectious cases data from these specific Brazilian states were collected from \href{
http://tabnet.datasus.gov.br/cgi/tabcgi.exe?sinannet/cnv/chikunbr.def}{Data SUS} \cite{tabnet_chikungunya}. 

The following table lists the temperature dependent parameters and constants used in the present model.
\begin{table}[htbp]\label{table}
\centering
\caption{Features and Parameters in the Model}
\label{table}
\begin{tabular}{|c|c|}
\hline
\textbf{Parameter} & \textbf{Constants and Temperature Dependent Formula} \\
\hline
$\mu_{h}$ & $\frac{1}{74.5*365} $ ${\rm day}^{-1}$ \cite{WorldBank2024}\\ 
$N_{h0M}$ & $2700000$ \cite{WorldPopulationReview2024}
\\
$N_{h0B}$ & $109966000$ \cite{IBGE2024}
\\
$R^{BR}$ & $\frac{109966000}{77*365}\ {\rm day}^{-1}$
\\
$R^{MIA}$ & $\frac{2700000}{80*365}\ {\rm day}^{-1}$
\\
$\mu_{m}(T)$ & $0.8692 - 0.1590T + 0.01116(T^{2}) -0.0003408(T^{3})+3.809(10^{-6})T^{4}\ {\rm day}^{-1} $ \cite{McLean1974}
\\
$b_{m}(T)$ & $\frac{(0.0943 + 0.0043 * T)}{N} \ {\rm day}^{-1}$, where N is the patch's total human population \cite{Focks1995} \\
$\beta_{mh}$ & $0.13$ \cite{Liu2020} \\
$\beta_{hm}$ & $0.3$ \cite{Liu2020} \\ 
$\mu_{l}(T)$ & $2.130 - 0.3797 \cdot T + 0.02457 \cdot T^2 - 0.0006778 \cdot T^3 + 0.000006794 \cdot T^4\ {\rm day}^{-1}$ \cite{Valdez2018} \\ 
$\tau(T)$ & $4 + e^{5.15 - 0.123 \cdot T}\ {\rm day}$ \cite{Focks1995} \\ 
$\tau_{h}$ & $5 \  {\rm day},$ \href{https://wwwnc.cdc.gov/travel/yellowbook/2024/infections-diseases/chikungunya} {CDC} \cite{CDC} \\

$\mu_{b}(T)$ & $  -0.32 + 0.11 \cdot T - 0.12 \cdot T^{2} + 0.01 \cdot T^{3} - 0.00015 \cdot T^{4} \ {\rm day}^{-1}$ \cite{Valdez2018}
\\ 
$a$ & 0.25 \cite{Mourad2022} \\
\hline
\end{tabular}
\end{table}
Given the temperature data collected from \cite{power_dav}, it was utilized Fast Fourier Transform from SciPy.FFT library in Python to estimate the periodic and time dependent average temperature functions in the Brazilian states  considered on this work and in Miami as well, $T_{BR}(t) $ and $T_{MIA}(t)$, respectively, defined as:
\[
T_{BR}(t) = 2.85 \cdot \sin\left( \frac{2\pi (t+59)}{365} + 2.2\right) + 24.182,
\]
\[
T_{MIA}(t) =
3.6375\cdot\sin\left(\frac{2\pi (t+59)}{365} + 3.7\right)+25.91,
\] %somo 59 dias pq foi quando a covid comecou, esse eh o ida em q inicializo o modelo!!!!!
where the unit of the time $t$ is $day$, and the temperature unit is in Celsius ($^{\circ}C$).
Based on Table \ref{table} we define the per-capita biting rate time dependent function for Brazil as
$$
b_m^{BR}(t) = 
\frac{0.0943 + 0.0043 \cdot T_{BR}(t)}{N_{BR}}.
$$

The use of air conditioning not only makes the environment less accessible to mosquitoes but also affects their survival and distribution, leading to lower biting rates and reduced disease transmission in the U.S. compared to other countries, as addressed in \cite{reiter2003texas}. Thus, following the lifting of lockdown restrictions in Florida in September 2020 as seen in \cite{PhaseThreeFlorida2020}, we consider the per-capita biting rate in Miami as: $$
b_m^{MIA}(t) = 
\begin{cases} 
\frac{{0.0335 + 0.0015 \cdot {T_{MIA}(t)}}}{N_{MIA}}& \text{if } \text{January 1, 2020} \leq t \leq \text{ July 27, 2020}, \\ \\
\frac{(0.0335 + 0.0015 \cdot T_{MIA}(210)) + (t-209)(b_m^{MIA}(209)-b_m^{MIA}(210))}{N_{MIA}} & \text{if } \text{ July 27, 2020}<t<\text{ July 28, 2020}, \\ \\
\frac{{0.07544 + 0.00344 \cdot T_{MIA}(t)}}{N_{MIA}} & \text{if } \text{ July 28, 2020} \leq t.

\end{cases}
$$Notice that, in opposite to Miami-Dade county, the per-capita biting rate in Brazil was not regarded as a piecewise defined function throughout the course of the COVID-19 pandemic during 2020. This is specially because in Brazil the transmission of Dengue virus and CHIKV mainly occurs through Aedes mosquitoes living nearby one's own home, as stated in \cite{goias_dengue}. In Brazil, most households do not have air conditioning, which results in greater reliance on open windows and doors for ventilation \cite{agenciabrasil2024dengue}. This practice increases human exposure to mosquitoes, especially in urban and peri-urban areas. This suggests that even under restricted lockdown policies, as the ones that occurred throughout the year of 2020, the biting rate was not much affected.

Now, utilizing linear model from Scikit-Learn library in Python, we fitted a Linear Regression model in the data provided in \cite{couret2014temperature}, as shown in Figure \ref{fig:larvae_development_rate}. Thus, we provide a linear formula for the temperature dependent Larvae Development Rate $\eta_{l}(T) = 0.004173T-0.007629$, whose inverse is the time delay for the larvae development stage $\tau_{l}(T) := \frac{1}{\eta_{l}(T)} = \frac{1}{0.004173T-0.007629}$.
\begin{figure}[htp]
    \centering
\includegraphics[height=6cm, width=11cm]{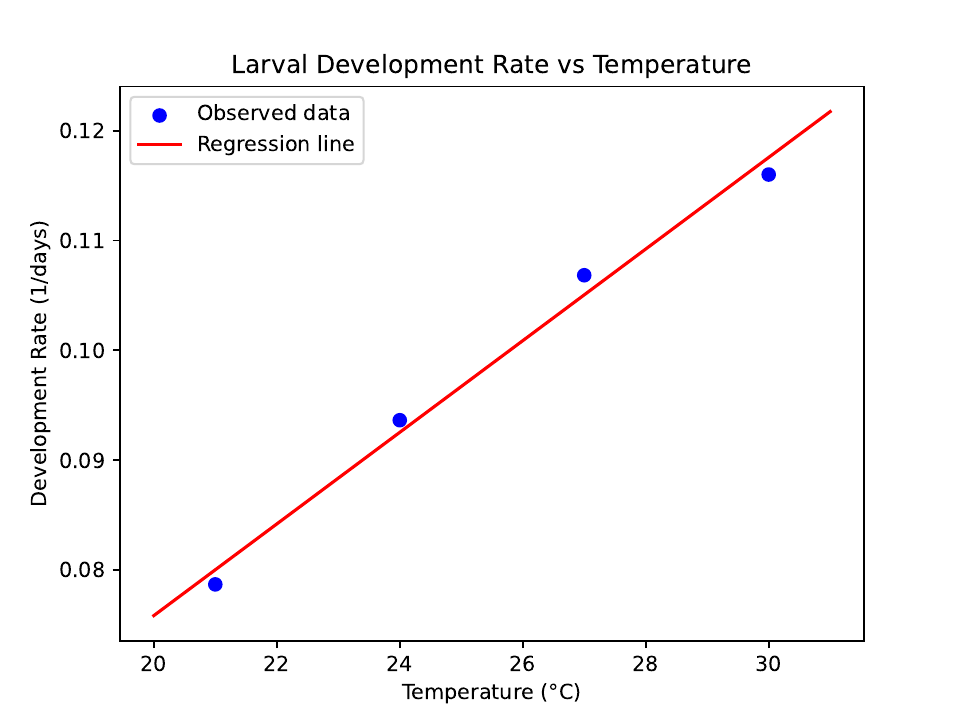}
    \caption{Larval Development Rate vs Temperature.}
\label{fig:larvae_development_rate}
\end{figure}
%\vspace{-7.0cm}

In light of the U.S. reopening its borders to fully vaccinated international travelers, including Brazilians, on November 8, 2021, as stated in \cite{cnn2021} and following \cite{arino2015epidemiological}, we define the traveling rates from Miami to the selected regions in Brazil and from these regions to Miami, respectively, as:
\[
m_{MB}(t) =
\begin{cases}
0 & \text{if } \text{January 1, 2020} \leq t \leq \text{November 7, 2021} \\
(t - 676) \cdot \left(-\log\left(1 - \frac{700000}{365 \cdot N_{h0M}}\right)\right) & \text{if }\  \text{November 7, 2021} < t < \text{November 8, 2021 } \\
-\log\left(1 - \frac{700000}{365 \cdot N_{h0M}}\right) & \text{if } t \geq \text{November 8, 2021},
\end{cases}
\]
\[
m_{BM}(t) =
\begin{cases}
0 & \text{if } \text{January 1, 2020} \leq t \leq \text{November 7, 2021} \\
(t - 676) \cdot \left(-\log\left(1 - \frac{700000}{365 \cdot N_{h0B}}\right)\right) & \text{if }  \text{November 7, 2021} < t < \text{November 8, 2021 }\\
-\log\left(1 - \frac{700000}{365 \cdot N_{h0B}}\right) & \text{if } t \geq \text{November 8, 2021}.
\end{cases}
\]

According to \cite{visitflorida}, around 1 million travelers arrived in Florida from Brazil in 2023, out of those we consider that in average 700.000 Brazilians visit Miami-Dade per year since November 2021.

The days in the functions definitions in the model, as well as the initial day of the model correspond to the following dates in chronological order:  $ t=0, \ t=76,\ t=209,\ t=677, $ correspond, respectively, to January 1st, 2020; March 17, 2020; July 27, 2020; November 8, 2021.

Utilizing Scipy.optimize.minimize from Python, it was minimized the Root Mean Squared Error for the infectious humans compartment $I_h$ based on the provided data and the numerical solution computed in subsection \ref{4.2}, and then computed the optimum values for the carrying capacities for both patches to be approximately $K^{BR}=1.9\cdot 10^{8}$, $K^{MIA} = 2.47 \cdot 10^6$, and the for vertical probability to be approximately $p = 0.05$.

\subsection{Data Fitting and Sensitive Analysis}\label{4.2}

Define the average formula for any $w-$periodic real function $h(t)$ for $w>0$, as 
$$
h:=\frac{1}{w}\int_{0}^{w}h(t)dt.
$$ 
From now on, by letting $w=365,$ we consider the time delay averages for both patches, thus by using the average formula, we calculate $\tau^{BR}  = 13.08$ days,
$\tau^{MIA} = 11.48$ days,
$\tau_{l}^{BR}  = 10.80 $ days, and 
$\tau_{l}^{MIA}  = 10.06 $ days. These time frames are close to the values shown in \cite{queensland_health_chikungunya} and \cite{cdc_aedes_lifecycle}.

Henceforth, we arrive in the system that we will be solving numerically:
\begin{footnotesize}
\begin{align*}
\frac{dS_h^{MIA}}{dt} &= R^{MIA} - b_m^{MIA}(t)\beta_{mh} S_h^{MIA}(t) I_m^{MIA}(t) - \mu_h S_h^{MIA}(t) + m_{MB}S_h^{j}(t) - m_{BM}S_h^{MIA}(t) \\
\frac{dI_h^{MIA}}{dt} &= (1-a)\beta_{mh}b_m^{MIA}(t-\tau^{MIA}) S_h^{MIA}(t-\tau^{MIA}) I_m^{MIA}(t-\tau^{MIA}) e^{-\mu_h \tau^{MIA}} - (\eta_h + \mu_h) I_h^{MIA}(t) \\
\frac{dA_h^{MIA}}{dt} &= a\beta_{mh}b_m^{MIA}(t-\tau^{MIA}) S_h^{MIA}(t-\tau^{MIA}) I_m^{MIA}(t-\tau^{MIA}) e^{-\mu_h \tau^{MIA}} - (\eta_h + \mu_h) A_h^{MIA}(t) \\
&\quad + m_{BM} A_h^{BR}(t) - m_{MB} A_h^{MIA}(t) \\
\frac{dL_I^{MIA}}{dt} &= \mu_b^{MIA}(t) \left(1 - \frac{L_s^{MIA}(t) + L_I^{MIA}(t)}{K^{MIA}}\right) p I_m^{MIA}(t) - \mu_l^{MIA}(t) L_I^{MIA}(t) \\
\frac{dL_s^{MIA}}{dt} &= \mu_b^{MIA}(t) \left(1 - \frac{L_s^{MIA}(t) + L_I^{MIA}(t)}{K^{MIA}}\right)\left(S_m^{MIA}(t) + E_m^{MIA}(t) + (1-p) I_m^{MIA}(t)\right) - \mu_l^{MIA}(t) L_s^{MIA}(t) \\
\frac{dS_m^{MIA}}{dt} &= \mu_b^{MIA}(t-\tau_l^{MIA}) \left(1 - \frac{L_s^{MIA}(t-\tau_l^{MIA}) + L_I^{MIA}(t-\tau_l^{MIA})}{K^{MIA}}\right)\\
&\quad \times\left(S_m^{MIA}(t-\tau_l^{MIA}) + E_m^{MIA}(t-\tau_l^{MIA}) + (1-p) I_m^{MIA}(t-\tau_l^{MIA})\right)e^{-\int_{t-\tau_l^{MIA}}^{t} \mu_l^{MIA}(s) \, ds}\\
&\quad - b_m^{MIA}(t)\beta_{hm}S_m^{MIA}(t) (I_h^{MIA}(t) + A_h^{MIA}(t)) - \mu_m^{MIA}(t) S_m^{MIA}(t) \\
\frac{dE_m^{MIA}}{dt} &= b_m^{MIA}(t)\beta_{hm} S_m^{MIA}(t) (I_h^{MIA}(t) + A_h^{MIA}(t)) - \mu_m^{MIA}(t) E_m^{MIA}(t) \\
&\quad -  \beta_{hm} b_m^{MIA}(t-\tau^{MIA}) (I_h^{MIA}(t-\tau^{MIA}) + A_h^{MIA}(t-\tau^{MIA})) S_m^{MIA}(t-\tau^{MIA}) e^{-\int_{t-\tau^{MIA}}^{t} \mu_l^{MIA}(s) \, ds} \\
\frac{dI_m^{MIA}}{dt} &= \mu_b^{MIA}(t - \tau_l^{MIA}) \left(1 - \frac{L_s^{MIA}(t - \tau_l^{MIA}) + L_I^{MIA}(t - \tau_l^{MIA})}{K^{MIA}}\right) p I_m^{MIA}(t - \tau_l^{MIA}) e^{-\int_{t - \tau_l^{MIA}}^{t} \mu_l^{MIA}(s) \, ds} \\
&\quad +   \beta_{hm} b_m^{MIA}(t - \tau^{MIA}) (I_h^{MIA}(t - \tau^{MIA}) + A_h^{MIA}(t - \tau^{MIA})) S_m^{MIA}(t - \tau^{MIA}) e^{-\int_{t - \tau^{MIA}}^{t} \mu_l^{MIA}(s) \, ds} \\
&\quad - \mu_m^{MIA}(t) I_m^{MIA}(t)\\
\frac{dS_h^{BR}}{dt} &= R^{BR} - b_m^{BR}(t)\beta_{mh} S_h^{BR}(t) I_m^{BR}(t) - \mu_h S_h^{BR}(t) + m_{MB}S_h^{j}(t) -  m_{BM}S_h^{BR}(t) \\
\frac{dI_h^{BR}}{dt} &= (1-a)\beta_{mh}b_m^{BR}(t-\tau^{BR}) S_h^{BR}(t-\tau^{BR}) I_m^{BR}(t-\tau^{BR}) e^{-\mu_h \tau^{BR}} - (\eta_h + \mu_h) I_h^{BR}(t) \\
\frac{dA_h^{BR}}{dt} &= a\beta_{mh}b_m^{BR}(t-\tau^{BR}) S_h^{BR}(t-\tau^{BR}) I_m^{BR}(t-\tau^{BR}) e^{-\mu_h \tau^{BR}} - (\eta_h + \mu_h) A_h^{BR}(t) + m_{MB} A_h^{BR}(t) -  m_{BM} A_h^{BR}(t) \\
\frac{dL_I^{BR}}{dt} &= \mu^{BR}_b(t) \left(1 - \frac{L_s^{BR}(t) + L_I^{BR}(t)}{K^{BR}}\right) p I_m^{BR}(t) - \mu_l^{BR}(t) L_I^{BR}(t) \\
\frac{dL_s^{BR}}{dt} &= \mu_b^{BR}(t) \left(1 - \frac{L_s^{BR}(t) + L_I^{BR}(t)}{K^{BR}}\right)\left(S_m^{BR}(t) + E_m^{BR}(t) + (1-p) I_m^{BR}(t)\right) - \mu_l^{BR}(t) L_s^{BR}(t) \\
\frac{dS_m^{BR}}{dt} &= \mu_b^{BR}(t - \tau_l^{BR}) \left(1 - \frac{L_s^{BR}(t - \tau_l^{BR}) + L_I^{BR}(t - \tau_l^{BR})}{K^{BR}}\right) [S_m^{BR}(t - \tau_l^{BR}) + E_m^{BR}(t - \tau_l^{BR}) \\
&\quad + (1 - p) I_m^{BR}(t - \tau_l^{BR})] e^{-\int_{t - \tau_l^{BR}}^{t} \mu_l^{BR}(s) \, ds} - b_m^{BR}(t) \beta_{hm} S_m^{BR}(t) (I_h^{BR}(t) + A_h^{BR}(t)) - \mu_m^{BR}(t) S_m^{BR}(t) \\
\frac{dE_m^{BR}}{dt} &= b_m^{BR}(t) \beta_{hm} S_m^{BR}(t) (I_h^{BR}(t) + A_h^{BR}(t)) - \mu_m^{BR}(t) E_m^{BR}(t) \\
&\quad - \beta_{hm} b_m^{BR}(t - \tau^{BR}) (I_h^{BR}(t - \tau^{BR}) + A_h^{BR}(t - \tau^{BR})) S_m^{BR}(t - \tau^{BR}) e^{-\int_{t - \tau^{BR}}^{t} \mu_l^{BR}(s) \, ds} \\
\frac{dI_m^{BR}}{dt} &= \mu_b^{BR}(t - \tau_l^{BR}) \left(1 - \frac{L_s^{BR}(t - \tau_l^{BR}) + L_I^{BR}(t - \tau_l^{BR})}{K^{BR}}\right) p I_m^{BR}(t - \tau_l^{BR}) e^{-\int_{t - \tau_l^{BR}}^{t} \mu_l^{BR}(s) \, ds} \\
&\quad + \beta_{hm} b_m^{BR}(t - \tau^{BR}) (I_h^{BR}(t - \tau^{BR}) + A_h^{BR}(t - \tau^{BR})) S_m^{BR}(t - \tau^{BR}) e^{-\int_{t - \tau^{BR}}^{t} \mu_l^{BR}(s) \, ds} - \mu_m^{BR}(t) I_m^{BR}(t).
\end{align*}
\end{footnotesize}
Now we shall consider the initial value problem for the delay differential system as a constant vector in $\mathbb{R}^{16}$, calculated at the model starting date March $1^{st}$, 2020, taking into account that in Miami there was a DFE equilibrium at this specific date, so we can calculate the initial number of larvae and susceptible mosquitoes in Miami. Also, according to \cite{tabnet_chikungunya}, the cumulative number of infectious humans due to CHIKV from January $1^{st}$ 2017 until March $1^{st}$ 2020 is 863127. 
We use the same equations to approximate the numbers of susceptible mosquitoes and larvae in Brazil, assuming that due to a high mosquito population number, the numbers of susceptible mosquitoes and non-infected larvae are sufficiently larger than the numbers of infected mosquitoes and infected larvae, respectively, so we can estimate these numbers by the DFE equilibrium
\begin{align*}
L^{BR}_{s0} &= K^{BR} \left(1 - \frac{\mu_{m}(T_{BR}(0))e^{\int^{0}_{-\tau_{l}^{BR}(0)\mu_{l}^{BR}(s)ds}}}{\mu_{b}^{BR}(-\tau^{BR}_{l}(0))}\right), \\
    S_{m0}^{BR}&= \frac{L^{BR}_{s0}\cdot(\eta_{l}(T_{BR}(0))+\mu_{l}(T_{BR}(0))}{\mu_{b}(t)\left(1 -  \frac{L^{BR}_{s0}}{K^{BR}} \right) },\\
    L^{MIA}_{s0} &= K^{MIA} \left(1 - \frac{\mu_{m}(T_{MIA}(0))e^{\int_{t-\tau_{l}}^{t}\mu_{l}(s)\,ds}}{\mu_{b}(t-\tau_{l})}\right),\\
    S_{m0}^{MIA} &= \frac{L^{MIA}_{s0} \cdot (\eta_{l}(T_{MIA}(0)) + \mu_{l}(T_{MIA}(0)))}{\mu_{b}(t) \left(1 - \frac{L^{MIA}_{s0}}{K^{MIA}}\right)}.
\end{align*}
%OQ TEMOS [DShB,  DIhB, DAhB,DLIB, DLSB, DSmB, DEmB, DImB,  DShM,  DIhM, DAhM,  DSmM, DEmM, DImM, DLIM, DLSM]
%OQ DEVEMOS TER[DShB, DIhB, DAhB,DLIB, DLSB, DSmB, DEmB, DImB, DShM, DIhM, DAhM,DLIM, DLSM, DSmM, DEmM, DImM,]
Hence, we arrive on the initial condition:
\begin{align*}
\varphi =& \left(  N_{h0M}, 0, 0, 0 , L_{s0}^{MIA}, S_{m0}^{MIA}, 0, 0, N_{h0B} - 863127, 860, 215, \right.\\
&\quad \left. 15000 , L_{s0}^{BR}, S_{m0}^{BR}, 1000, 5000  \right)\in\mathbb{R}^{16},
\end{align*}following 
In Figure \ref{eita mah}, we plot the number of infectious human population in Brazil from March $1^{st}$ 2020 to July 2024, and project into August 2025. It is worth mentioning that the COVID-19 pandemic likely disrupted arbovirus surveillance in Brazil. As suggested in \cite{fiocruz2021pandemia}, during COVID-19 pandemic, limited healthcare capacity and redirected attention may have led to underreporting of Arbovirus diseases.
\begin{figure}[htp]
\centering
\includegraphics[height=6cm, width=11cm]{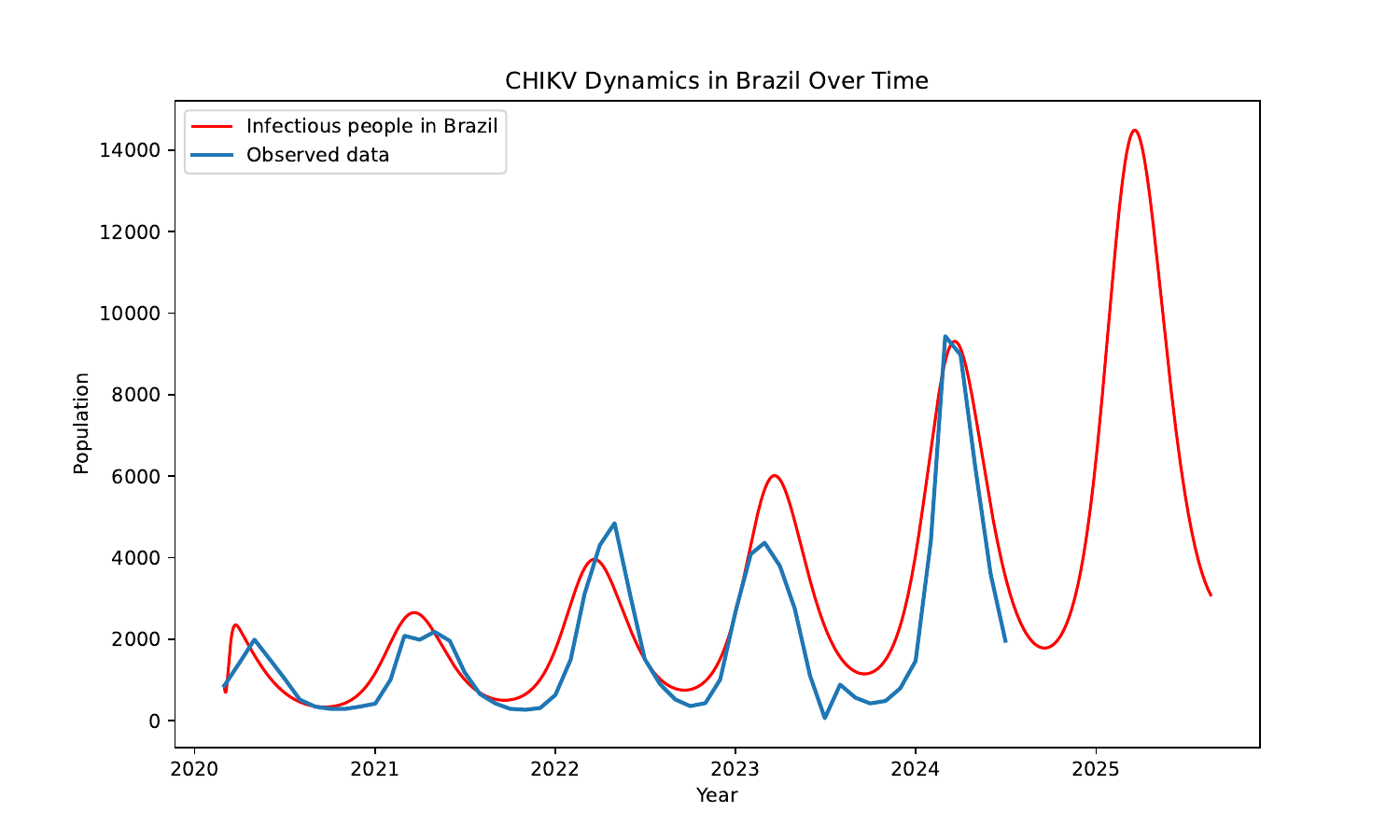}
\caption{Simulation of the number of symptomatic infectious human population in Brazil from March 2020 to August 2025. The blue curve represents the data and the red curve is the simulation of the model.}
\label{eita mah}
\end{figure}

Now let us define the per capita biting rate multiplier, mosquito death rate multiplier and larvae death rate multiplier as positive constants $1\geq b_{mult}>0,\ \mu_{m_{mult}}\geq 1,\ \mu_{l_{mult}}\geq 1$ that multiplies $b_{m}(t)$, $\mu_{m}(t),\ \mu_l (t)$, respectively, and which will play the role of control strategies aiming to reduce the biting rate, and increase the mortality rates, producing the new time dependent rates $b_{mult} \cdot b_{m}(t),\ \mu_{m_{mult}} \cdot \mu_{m}(t), \ \mu_{m_{mult}} \cdot \mu_{l}(t) $ at each patch. In Figure \ref{manuu} we do a sensitive analysis on $b_{mult}$ to emphasize the importance of reducing the mosquitoes biting rates, through  proving that effective reduction of mosquito bites can significantly decrease the transmission of CHIKV, as the virus spreads primarily through mosquito bites. In Figure \ref{miami} we project the cumulative number of infectious human population in Miami, proving that the use of multiple techniques to reduce the per capita biting rate will play an important role on the spread of the disease.
\begin{figure}[htp]
\centering
\begin{subfigure}[b]{0.49\textwidth}  % Adjust the width as needed
\centering
\includegraphics[width=\textwidth]{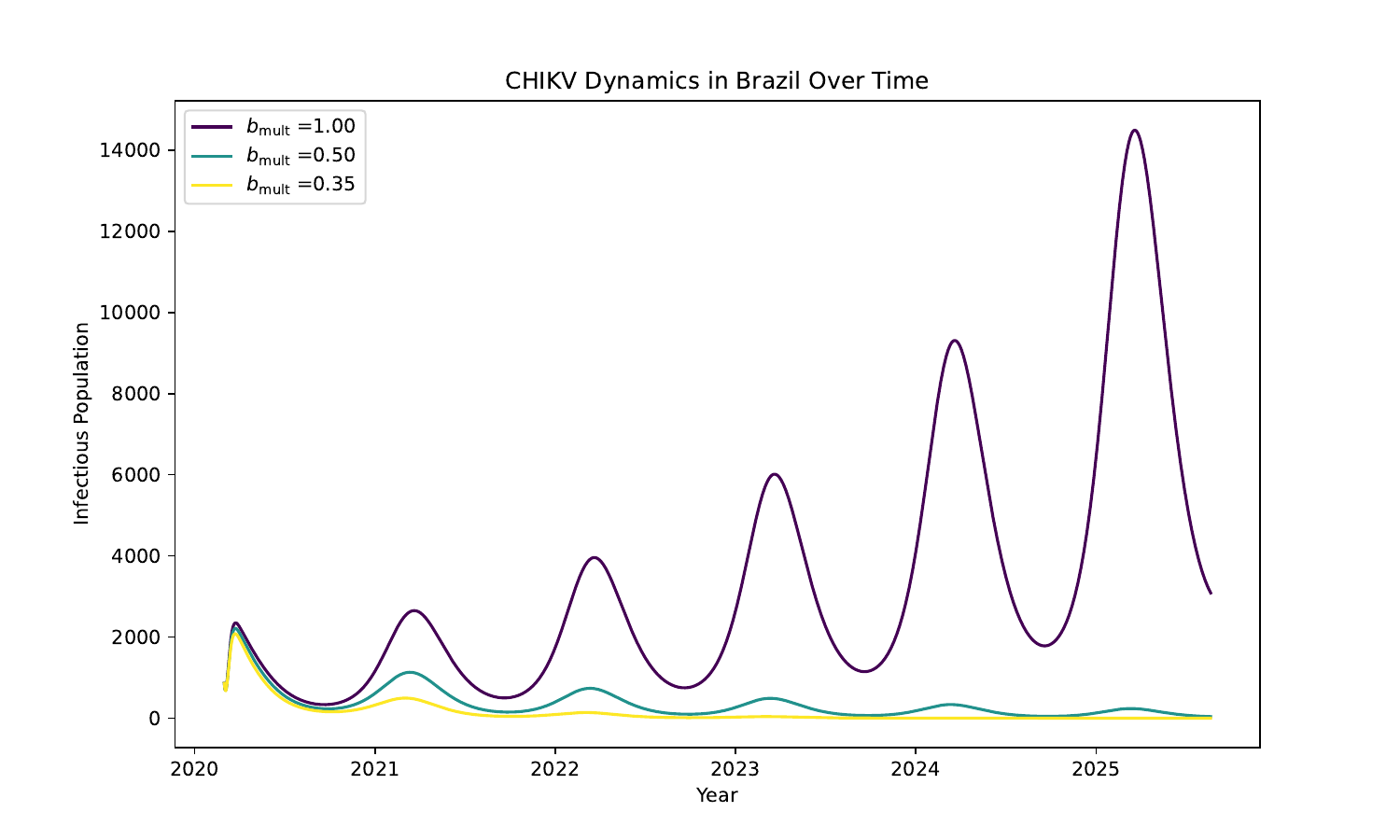}
\caption{Per capita biting rate sensitive analysis of Infectious and Symptomatic Human Population in Brazil ($I_h^{BR}$), projected until August 2025.}
\label{manuu}
\end{subfigure}
\hfill  % Creates space between the two subfigures
\begin{subfigure}[b]{0.49\textwidth}  % Adjust the width as needed
\centering 
\includegraphics[width=\textwidth]{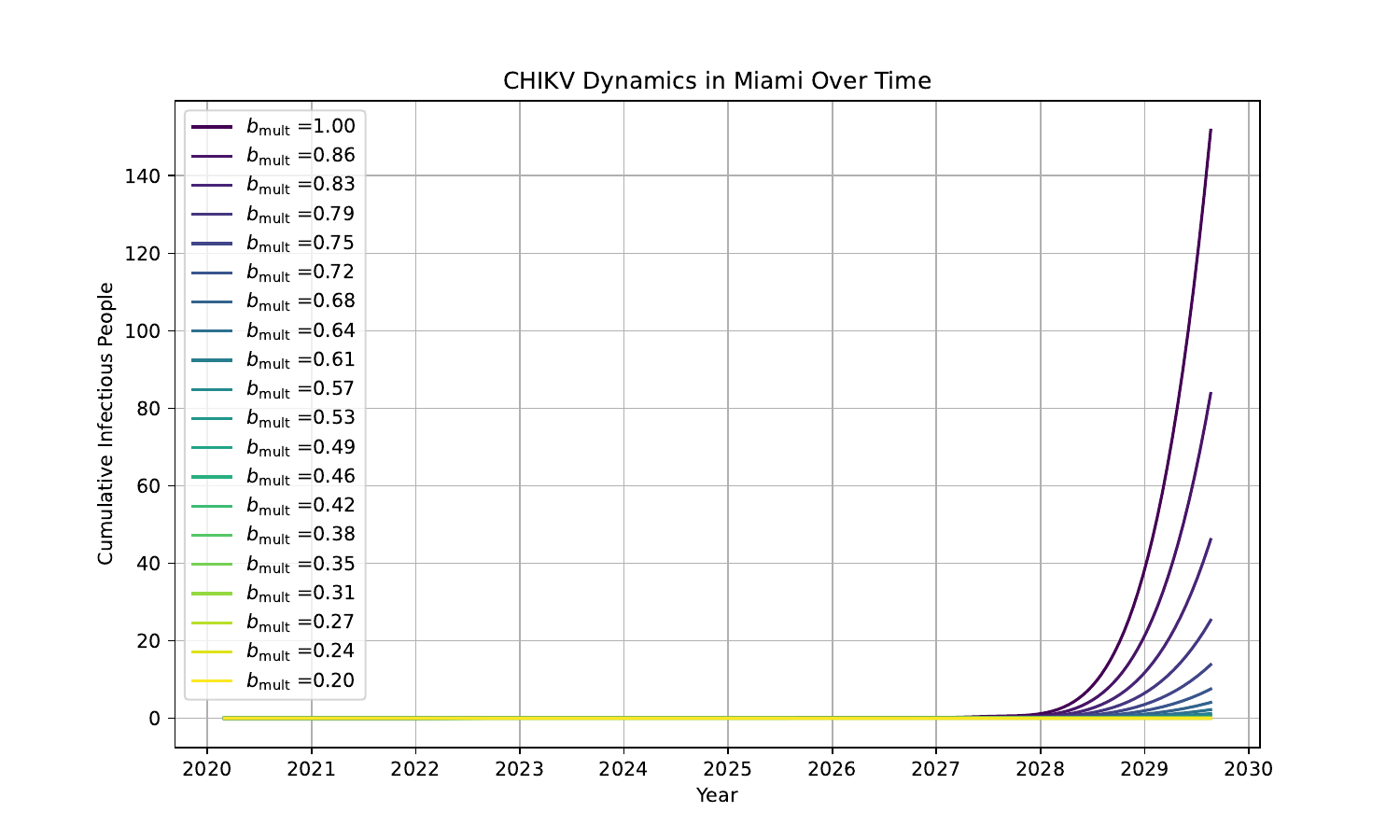}
\caption{Sensitive analysis of the per capita biting rate  of Infectious and Symptomatic Human Population in Miami ($I_h^{MIA}$), projected until 2030.}
\label{miami}
\end{subfigure}
\caption{Sensitive analysis of the per capita biting rates.}
\label{combined_figures}
\end{figure}

In Figure \ref{BRAZILIU} we conclude that both control strategies aiming to decrease the number of mosquitoes and larvae, e.g, through larvicides or pesticides, are significant to reduce the number of infectious cases. Also notice that in Brazil, by comparing Figure \ref{manuu} and Figure \ref{BRAZILIU}, we can understand  that small reductions in the per capita biting rate produce a similar change in the disease dynamics as in larger increases in the death rates of mosquitoes and larvae.
\begin{figure}[htp]
\centering
\begin{subfigure}[a]{0.48\textwidth}
\centering
(a)\includegraphics[width=\textwidth]{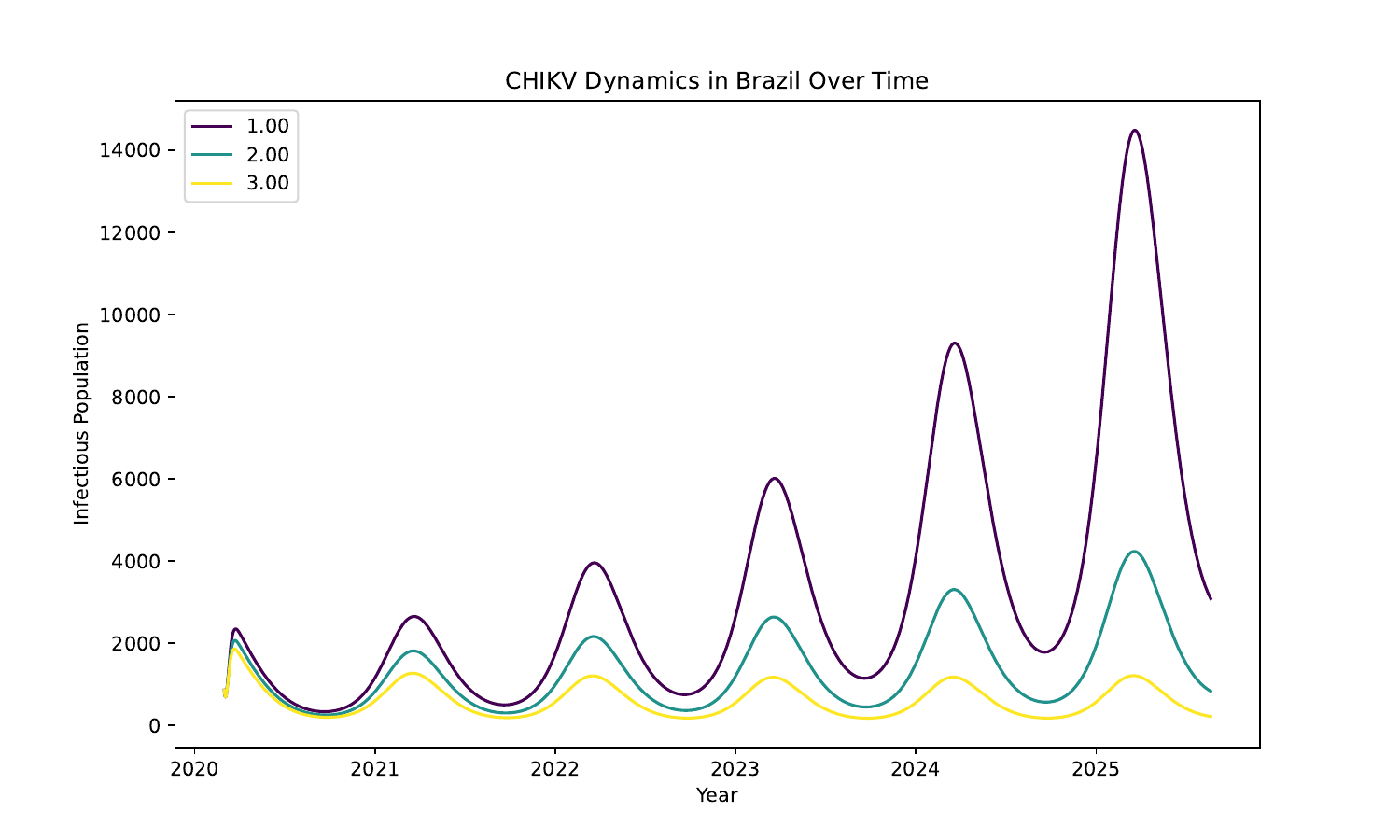}
\end{subfigure}
\hfill  % Creates space between the two subfigures
\begin{subfigure}[a]{0.48\textwidth}
\centering
(b)\includegraphics[width=\textwidth]{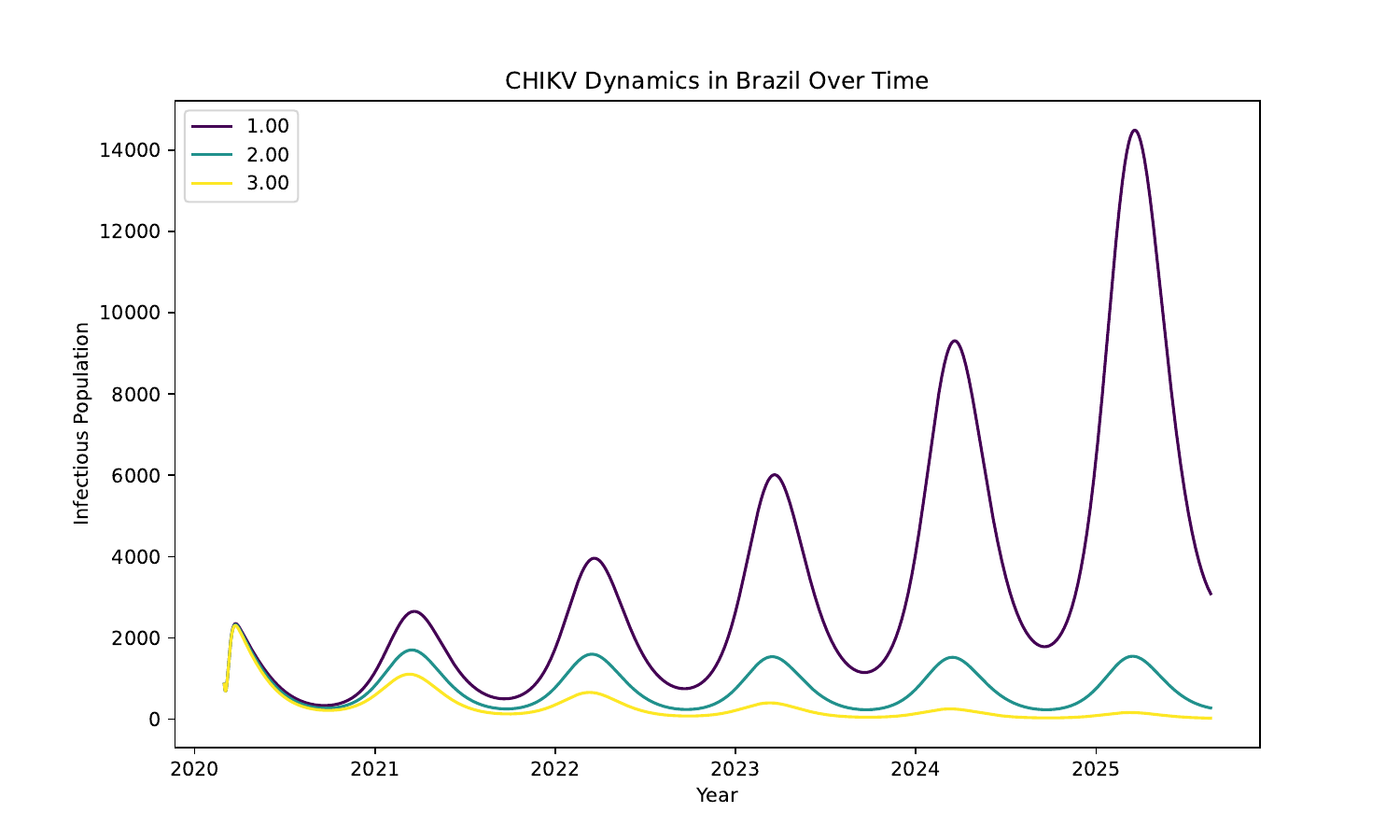}
\end{subfigure}
\caption{Sensitive analysis on (a) larvae and (b) mosquitoes death rate multiplier of Infectious and Symptomatic Human Population in Brazil ($I_h^{BR}$), projected until August 2025.}\label{BRAZILIU}
\end{figure}

The dynamics in Miami, much like in Brazil, are highly sensitive to biting rate multipliers. This suggests that the use of insecticidal sprays and spatial repellents is essential, along with potential advancements from new studies aimed at reducing the per capita biting rate. Furthermore, as evidenced by the projections presented below and in Figure \ref{miami}, controlling the mosquito biting rate and mosquito mortality rate may prove to be more effective than focusing on larvae death rates. This fact is essentially analogous to the situation in Brazil, as mentioned earlier.
\begin{figure}[htp]
\centering
\begin{subfigure}[a]{0.48\textwidth}
\centering
(a)\includegraphics[width=\textwidth]{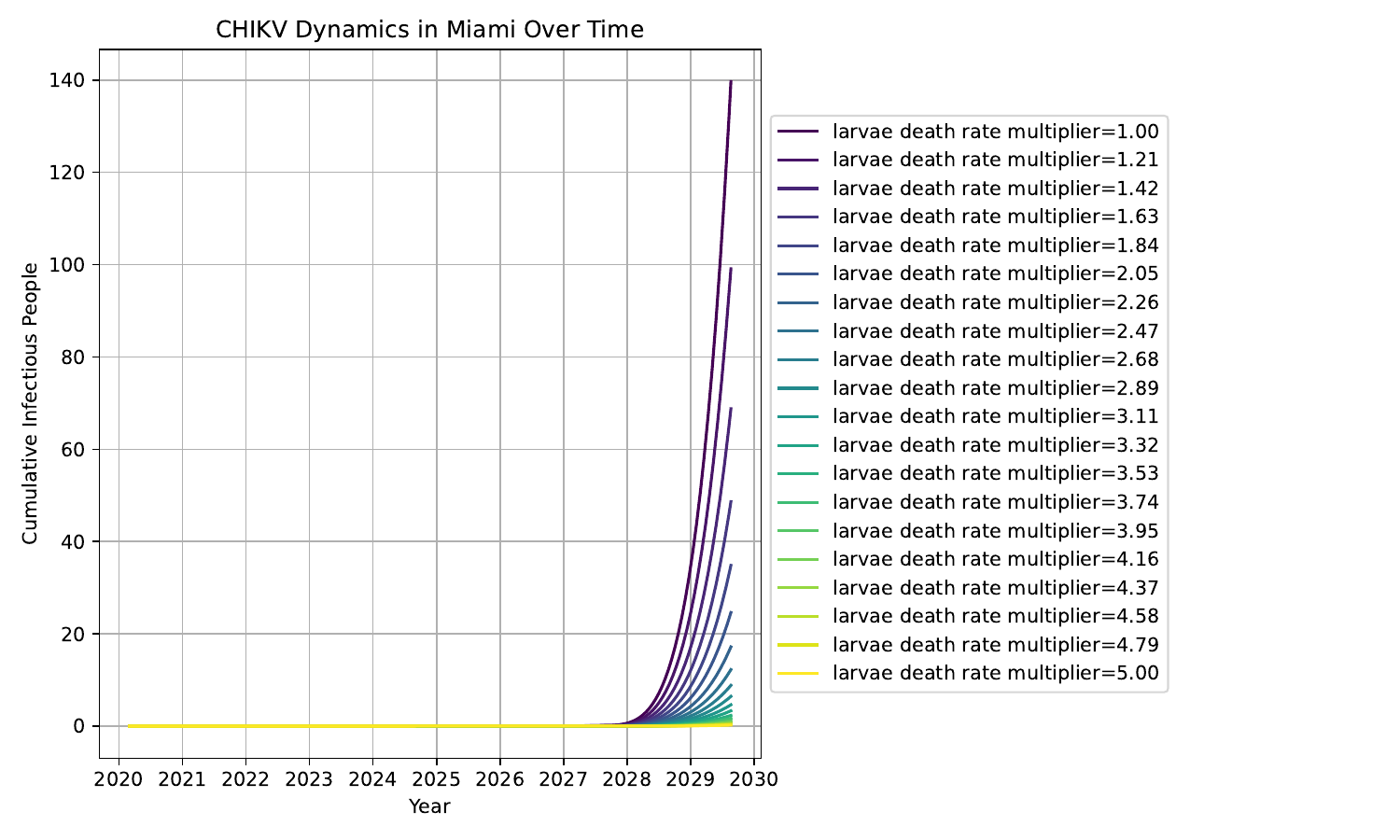}
\end{subfigure}
\hfill  % Creates space between the two subfigures
\begin{subfigure}[a]{0.48\textwidth}
\centering
(b)\includegraphics[width=\textwidth]{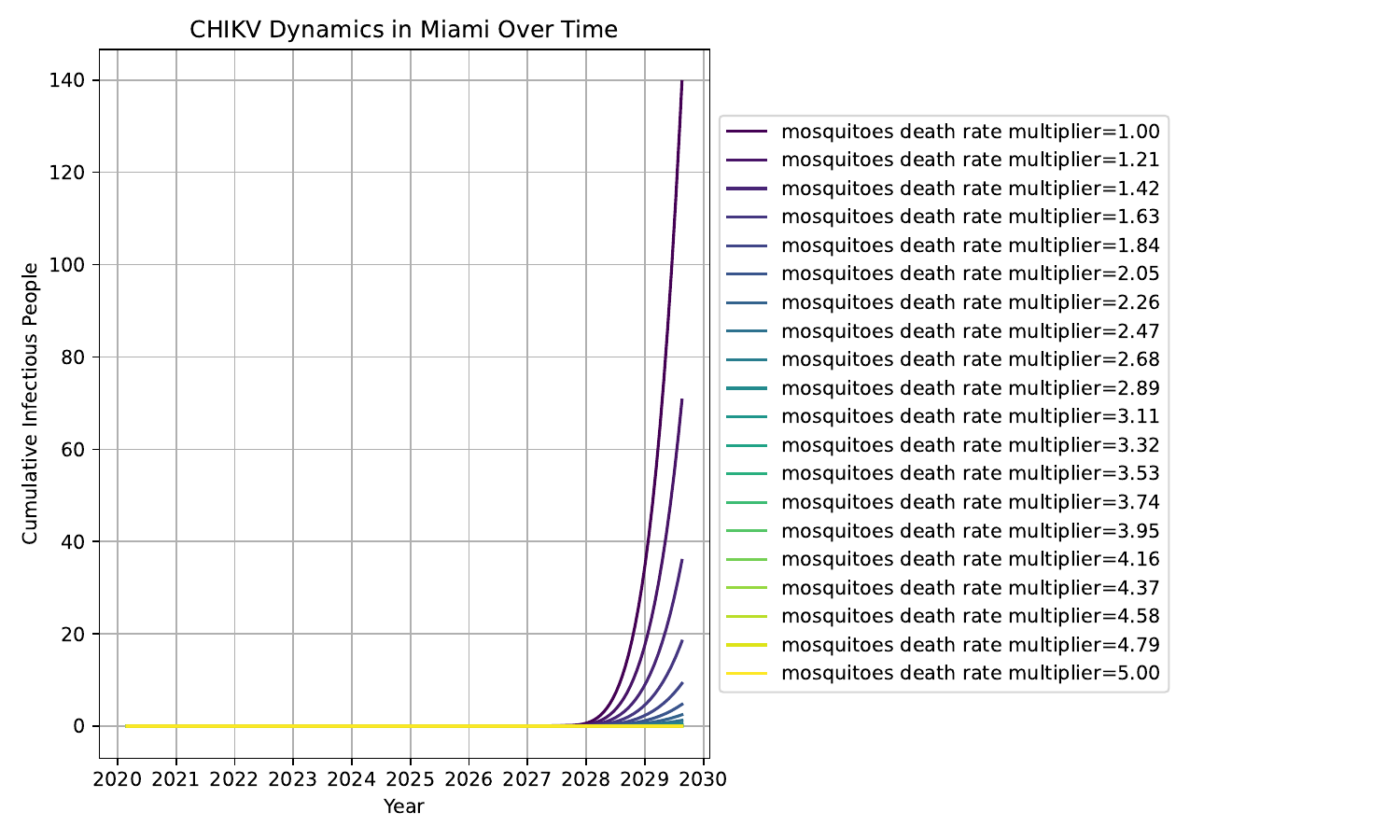}
\end{subfigure}
\caption{Sensitive analysis on (a) larvae and (b) mosquitoes death rate multiplier of Infectious and Symptomatic Human Population in Miami ($I_h^{MIA}$), projected until 2030.}\label{miamiz}
\end{figure}
    
Referencing Figure \ref{eita mah}, we consider the implementation of public health strategies starting in January 2025, specifically focusing on how the disease dynamics evolve in response to reductions in mosquito biting rates. In Figure \ref{proj_Br}, we examine variations in $b_{m_{\text{mult}}}$ from 1 to 0.7, which represents different levels of biting rate control strategies that may be adopted in Brazil throughout 2025, and observe a significant reduction on the number of infectious humans. It is important to note that in all plots of infectious individuals presented in this section, when varying one of the multipliers while performing sensitivity analysis, all other multipliers are held constant at a value of one, ensuring a controlled comparison of each intervention's impact on disease transmission.
\begin{figure}[htp]
\centering
\includegraphics[width=10cm]{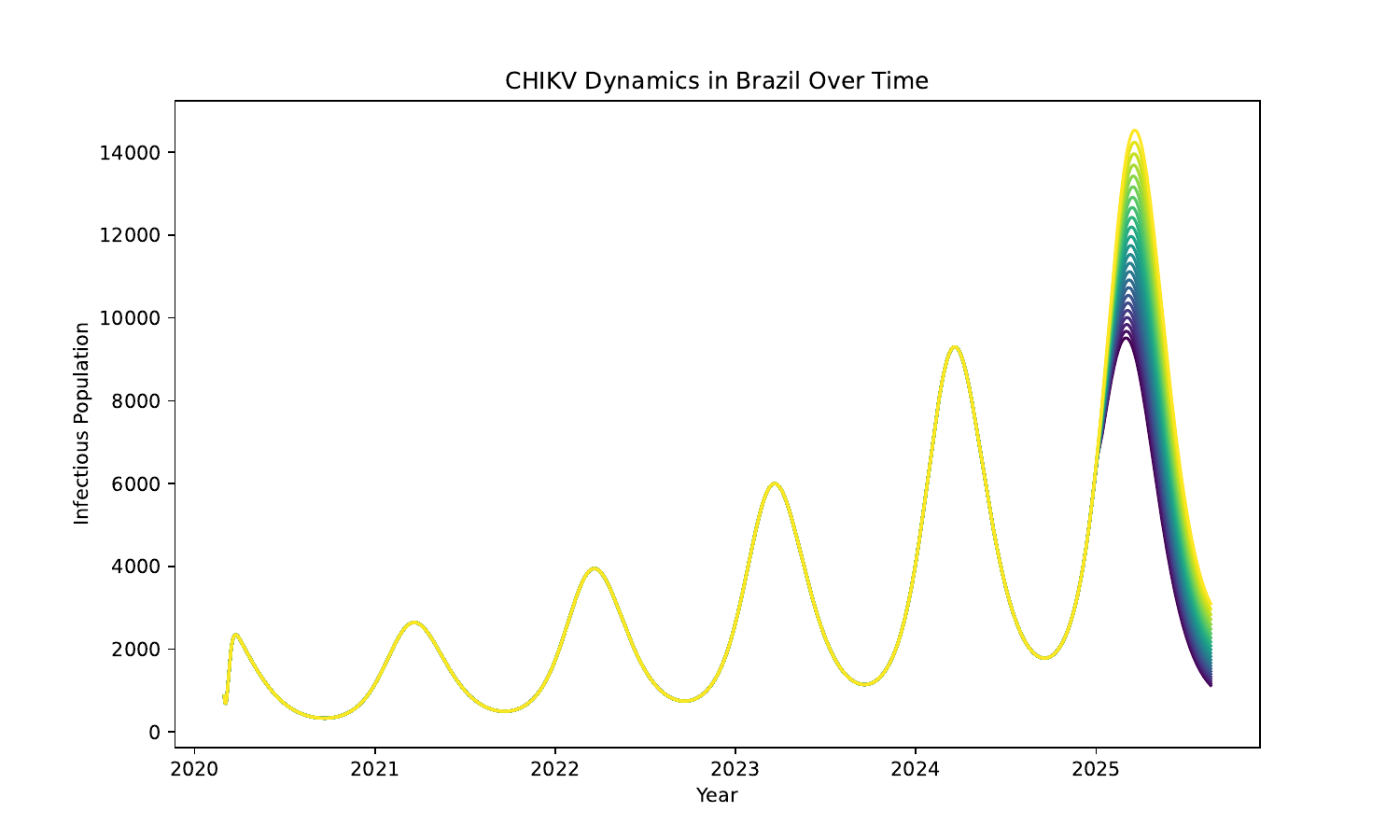}
\caption{Projection of Infectious and Symptomatic Human Population in Brazil ($I_h^{BR}$), forecast until August 2025, with varying biting rate control measures adopted before the peak of infections in 2025.}
\label{proj_Br}
\end{figure}

\subsection{Sensitive Analysis of the Mosquito and Basic Reproduction Numbers} 
%PRECISO TRANSFORMAR EM HYPERLINK E TIRAR DA BIBLIOGRAPHY [16],[17,18,24]

From now on, we average all of the time periodic parameters of our model, to arrive in the autonomous DDE system: 
\begin{equation}\label{plu}
\begin{aligned}[t]
\frac{dS_h^{MIA}}{dt} &= R^{MIA} - b_m^{MIA} \beta_{mh} S_h^{MIA}(t) I_m^{MIA}(t) - \mu_h S_h^{MIA}(t) + m_{MB}S_h^{j}(t) - m_{BM}S_h^{MIA}(t), \\
\frac{dI_h^{MIA}}{dt} &= (1-a) \beta_{mh} b_m^{MIA} S_h^{MIA}(t-\tau^{MIA}) I_m^{MIA}(t-\tau^{MIA}) e^{-\mu_h \tau^{MIA}} - (\eta_h + \mu_h) I_h^{MIA}(t), \\
\frac{dA_h^{MIA}}{dt} &= a \beta_{mh} b_m^{MIA} S_h^{MIA}(t-\tau^{MIA}) I_m^{MIA}(t-\tau^{MIA}) e^{-\mu_h \tau^{MIA}} - (\eta_h + \mu_h) A_h^{MIA}(t) + m_{BM} A_h^{BR}(t) \\
&\quad  - m_{MB} A_h^{MIA}(t), \\
\frac{dL_I^{MIA}}{dt} &= \mu_b^{MIA} \left(1 - \frac{L_s^{MIA}(t) + L_I^{MIA}(t)}{K^{MIA}}\right) p I_m^{MIA}(t) - \mu_l^{MIA} L_I^{MIA}(t), \\
\frac{dL_s^{MIA}}{dt} &= \mu_b^{MIA} \left(1 - \frac{L_s^{MIA}(t) + L_I^{MIA}(t)}{K^{MIA}}\right)\left(S_m^{MIA}(t) + E_m^{MIA}(t) + (1-p) I_m^{MIA}(t)\right) - \mu_l^{MIA} L_s^{MIA}(t), \\
\frac{dS_m^{MIA}}{dt} &= \mu_b^{MIA} \left(1 - \frac{L_s^{MIA}(t - \tau_l^{MIA}) + L_I^{MIA}(t - \tau_l^{MIA})}{K^{MIA}}\right) \left(S_m^{MIA}(t - \tau_l^{MIA}) + E_m^{MIA}(t - \tau_l^{MIA}) \right.\\
&\quad \left.+ (1-p) I_m^{MIA}(t - \tau_l^{MIA})\right) e^{-\mu_l^{MIA} \tau_l^{MIA}} - b_m^{MIA} \beta_{hm} S_m^{MIA}(t) (I_h^{MIA}(t) + A_h^{MIA}(t)) - \mu_m^{MIA} S_m^{MIA}(t), \\
\frac{dE_m^{MIA}}{dt} &= b_m^{MIA} \beta_{hm} S_m^{MIA}(t) (I_h^{MIA}(t) + A_h^{MIA}(t)) - \mu_m^{MIA} E_m^{MIA}(t) \\
&\quad - \beta_{hm} b_m^{MIA} S_m^{MIA}(t - \tau^{MIA}) (I_h^{MIA}(t - \tau^{MIA}) + A_h^{MIA}(t - \tau^{MIA})) e^{-\mu_l^{MIA} \tau^{MIA}}, \\
\frac{dI_m^{MIA}}{dt} &= \mu_b^{MIA} \left(1 - \frac{L_s^{MIA}(t - \tau_l^{MIA}) + L_I^{MIA}(t - \tau_l^{MIA})}{K^{MIA}}\right) p I_m^{MIA}(t - \tau_l^{MIA}) e^{-\mu_l^{MIA} \tau_l^{MIA}} \\
&\quad + \beta_{hm} b_m^{MIA} S_m^{MIA}(t - \tau^{MIA}) (I_h^{MIA}(t - \tau^{MIA}) + A_h^{MIA}(t - \tau^{MIA})) e^{-\mu_l^{MIA} \tau^{MIA}} - \mu_m^{MIA} I_m^{MIA}(t)\\
\frac{dS_h^{BR}}{dt} &= R^{BR} - b_m^{BR} \beta_{mh} S_h^{BR}(t) I_m^{BR}(t) - \mu_h S_h^{BR}(t) + m_{MB}S_h^{j}(t) - m_{BM}S_h^{BR}(t), \\
\frac{dI_h^{BR}}{dt} &= (1-a) \beta_{mh} b_m^{BR} S_h^{BR}(t-\tau^{BR}) I_m^{BR}(t-\tau^{BR}) e^{-\mu_h \tau^{BR}} - (\eta_h + \mu_h) I_h^{BR}(t), \\
\frac{dA_h^{BR}}{dt} &= a \beta_{mh} b_m^{BR} S_h^{BR}(t-\tau^{BR}) I_m^{BR}(t-\tau^{BR}) e^{-\mu_h \tau^{BR}} - (\eta_h + \mu_h) A_h^{BR}(t) + m_{MB} A_h^{MIA}(t) - m_{BM} A_h^{BR}(t), \\
\frac{dL_I^{BR}}{dt} &= \mu_b^{BR} \left(1 - \frac{L_s^{BR}(t) + L_I^{BR}(t)}{K^{BR}}\right) p I_m^{BR}(t) - \mu_l^{BR} L_I^{BR}(t), \\
\frac{dL_s^{BR}}{dt} &= \mu_b^{BR} \left(1 - \frac{L_s^{BR}(t) + L_I^{BR}(t)}{K^{BR}}\right)\left(S_m^{BR}(t) + E_m^{BR}(t) + (1-p) I_m^{BR}(t)\right) - \mu_l^{BR} L_s^{BR}(t) \\
\frac{dS_m^{BR}}{dt} &= \mu_b^{BR} \left(1 - \frac{L_s^{BR}(t - \tau_l^{BR}) + L_I^{BR}(t - \tau_l^{BR})}{K^{BR}}\right) \left(S_m^{BR}(t - \tau_l^{BR}) + E_m^{BR}(t - \tau_l^{BR}) \right. \\
&\quad \left. + (1-p) I_m^{BR}(t - \tau_l^{BR})\right) e^{-\mu_l^{BR} \tau_l^{BR}} - b_m^{BR} \beta_{hm} S_m^{BR}(t) (I_h^{BR}(t) + A_h^{BR}(t)) - \mu_m^{BR} S_m^{BR}(t), \\
\frac{dE_m^{BR}}{dt} &= b_m^{BR} \beta_{hm} S_m^{BR}(t) (I_h^{BR}(t) + A_h^{BR}(t)) - \mu_m^{BR} E_m^{BR}(t) \\
&\quad - \beta_{hm} b_m^{BR} S_m^{BR}(t - \tau^{BR}) (I_h^{BR}(t - \tau^{BR}) + A_h^{BR}(t - \tau^{BR})) e^{-\mu_l^{BR} \tau^{BR}}, \\
\frac{dI_m^{BR}}{dt} &= \mu_b^{BR} \left(1 - \frac{L_s^{BR}(t - \tau_l^{BR}) + L_I^{BR}(t - \tau_l^{BR})}{K^{BR}}\right) p I_m^{BR}(t - \tau_l^{BR}) e^{-\mu_l^{BR} \tau_l^{BR}} \\
&\quad + \beta_{hm} b_m^{BR} S_m^{BR}(t - \tau^{BR}) (I_h^{BR}(t - \tau^{BR}) + A_h^{BR}(t - \tau^{BR})) e^{-\mu_l^{BR} \tau^{BR}} - \mu_m^{BR} I_m^{BR}(t).
\end{aligned}
\end{equation}

Now, according to \cite{zhao2024linear} and following the same notations, we will compute the reproduction numbers for humans and mosquitoes populations, individually. First, define
$$\hat{F}_{v}^{BR}:=\begin{pmatrix}
0 & \mu_{b}^{BR} \\
0 & e^{-\tau_{l}^{BR} \mu_l^{BR}}
\end{pmatrix},
\quad \hat{V}_{v}^{BR}:=\begin{pmatrix}
\mu_{l}^{BR} & 0 \\
0 & \mu_{m}^{BR}
\end{pmatrix},$$
$$\hat{F}_{v}^{MIA}:=\begin{pmatrix}
0 & \mu_{b}^{MIA} \\
0 & e^{-\tau_{l}^{MIA} \mu_l^{MIA}}
\end{pmatrix},
\quad \hat{V}_{v}^{MIA}:=\begin{pmatrix}
\mu_{l}^{MIA} & 0 \\
0 & \mu_{m}^{MIA}
\end{pmatrix}.$$Thus we are in the position to compute the mosquito reproduction number for both patches as 
$$
\mathcal{R}_{m}^{BR} = \rho(\hat{F}_{v}^{BR} \cdot (\hat{V}_{v}^{BR})^{-1}), \;\; \mathcal{R}_{m}^{MIA} = \rho(\hat{F}_{v}^{MIA} \cdot (\hat{V}_{v}^{MIA})^{-1}).
$$
Now, depending on $\operatorname{sign}(\mathcal{R}_{m}^{MIA}-1)$ and $\operatorname{sign}(\mathcal{R}_{m}^{BR}-1)$, we define the basic reproduction number $\mathcal{R}_{0}$ for the multi-patch system. First let us define the following matrices, for $i,j\in \{1,2\},\ i\neq j$, and considering $m\in \{m_{ij}, m_{ji}\}$:
\[
F_{0_{kl}}^{i}:=\begin{bmatrix}
0 & 0 & 0 & (1-a)\beta_{mh} b_m^i S_{h,\text{DFE}}^i e^{-\mu_h \tau_h} \\
0 & m_{kl} & 0 & a\beta_{mh} b_m^i S_{h,\text{DFE}}^i e^{-\mu_h \tau_h} \\
0 & 0 & 0 & p\ \mu_b^{i} \\
0 & 0 & 0 & \mu_b^i p\ e^{-\tau_l^i \mu_l^i}
\end{bmatrix},
\]
\[
M^{i} = \mu_{b}^{i} \left( 1 - \frac{L_{s,\text{DFE}}^i}{K^{i}} \right) p\  e^{-\tau_l^i \mu_l^i},
\]
\[
F_{1_{kl}}^{i}:=\begin{bmatrix}
0 & 0 & 0 & (1-a)\beta_{mh}b_m^i S_{h,\text{DFE}}^{i} e^{-\mu_h \tau_h} \\
0 & m_{kl} & 0 & a\beta_{mh} b_m^i S_{h,\text{DFE}}^i e^{-\mu_h \tau_h} \\
0 & 0 & 0 & p\ \mu_b^{i} \left( 1 - \frac{L_{s,\text{DFE}}^i}{K^{i}} \right) \\
\beta_{hm} b_m^i S_{m,\text{DFE}}^i e^{-\tau^i \mu_l^i} & \beta_{hm} b_m^i S_{m,\text{DFE}}^i e^{-\tau^i \mu_l^i} & 0 & M^i
\end{bmatrix},
\]
where we compute the DFE for the autonomous system originated from (\ref{pem}) to obtain
\begin{align*}
    S_{h,\text{DFE}}^{BR} &= 
\frac{m_{BM} \cdot \mu_h \cdot N_{h0B} + \mu_h \cdot N_{h0M}}{(\mu_h + m_{MB}) (\mu_h + m_{BM}) - m_{MB} \cdot m_{BM}},\\
    S_{h,\text{DFE}}^{MIA} &=  
\frac{m_{MB} \cdot S_{h1,\text{DFE}} + \mu_h \cdot N_{h0B}}{\mu_h + m_{MB}},\\
    L_{s,\text{DFE}}^i &=   
K^{i} \left( 1 - \mu_m^{i} \frac{e^{-\tau_l^{i} \mu_l^{i}}}{\mu_b^{i}} \right),\\
S_{m,\text{DFE}}^i &=   
 \frac{L_{s,\text{DFE}}^i\ \mu_l^i}{\mu_b^i \left( 1 - \frac{L_{s,\text{DFE}}^i}{K^i} \right)}.
\end{align*}
Also, 
\begin{align*}   
m_{MB}=-\log\left(1-\frac{(700000/365)}{N_{h0M}}\right),\;\;
m_{BM}=-\log\left(1-\frac{(700000/365)}{N_{h0B}}\right).
\end{align*}
Furthermore, let $$V := 
\begin{pmatrix}
\eta_h + \mu_h & 0 & 0 & 0 & 0 & 0 & 0 & 0 \\
0 & \eta_h + \mu_h + m_{MB} & 0 & 0 & 0 & 0 & 0 & 0 \\
0 & 0 & \mu_l^{MIA} & 0 & 0 & 0 & 0 & 0 \\
0 & 0 & 0 & \mu_m^{MIA} & 0 & 0 & 0 & 0 \\
0 & 0 & 0 & 0 & \eta_h + \mu_h & 0 & 0 & 0 \\
0 & 0 & 0 & 0 & 0 & \eta_h + \mu_h + m_{BM} & 0 & 0 \\
0 & 0 & 0 & 0 & 0 & 0 & \mu_l^{BR} & 0 \\
0 & 0 & 0 & 0 & 0 & 0 & 0 & \mu_m^{BR} \\
\end{pmatrix}.
$$
Now, considering that $0\in Mat(4,4),$ for each value of $\operatorname{sgn}(\mathcal{R}_{m}^i - 1)$ we define a matrix $F$ as follows:
\begin{equation}
 F:=\begin{cases}
    \begin{pmatrix}
    F_{0_{BM}}^{MIA} & 0 \\
    0 & F_{0_{MB}}^{BR}
\end{pmatrix}\ \in Mat(8,8) & \text{if } \mathcal{R}_{m}^{BR} - 1,\ \mathcal{R}_{m}^{MIA} - 1<0, \\
    \begin{pmatrix}
    F_{1_{BM}}^{MIA} & 0 \\
    0 & F_{1_{MB}}^{BR}
\end{pmatrix}\ \in Mat(8,8) & \text{if } \mathcal{R}_{m}^{BR} - 1,\ \mathcal{R}_{m}^{MIA} - 1>0, \\
    \begin{pmatrix}
    F_{1_{BM}}^{MIA} & 0 \\
    0 & F_{0_{MB}}^{BR}
\end{pmatrix}\ \in Mat(8,8) & \text{if } \mathcal{R}_{m}^{BR} - 1>0,\ \mathcal{R}_{m}^{MIA} - 1<0, \\
    \begin{pmatrix}
    F_{0_{BM}}^{MIA} & 0 \\
    0 & F_{1_{MB}}^{BR}
\end{pmatrix}\ \in Mat(8,8) & \text{if } \mathcal{R}_{m}^{BR} - 1<0,\ \mathcal{R}_{m}^{MIA} - 1>0.
\end{cases}
\end{equation}
Hence, $\mathcal{R}_{0}=\rho(FV^{-1})$.
For each season, we average the time dependent parameters of system (\ref{plu}), to arrive in Figures \ref{Rm1}, \ref{Rm2}.
\begin{figure}[htp]
\centering
\begin{subfigure}[b]{0.45\textwidth}  % Reduced width
\centering
\includegraphics[width=\textwidth]{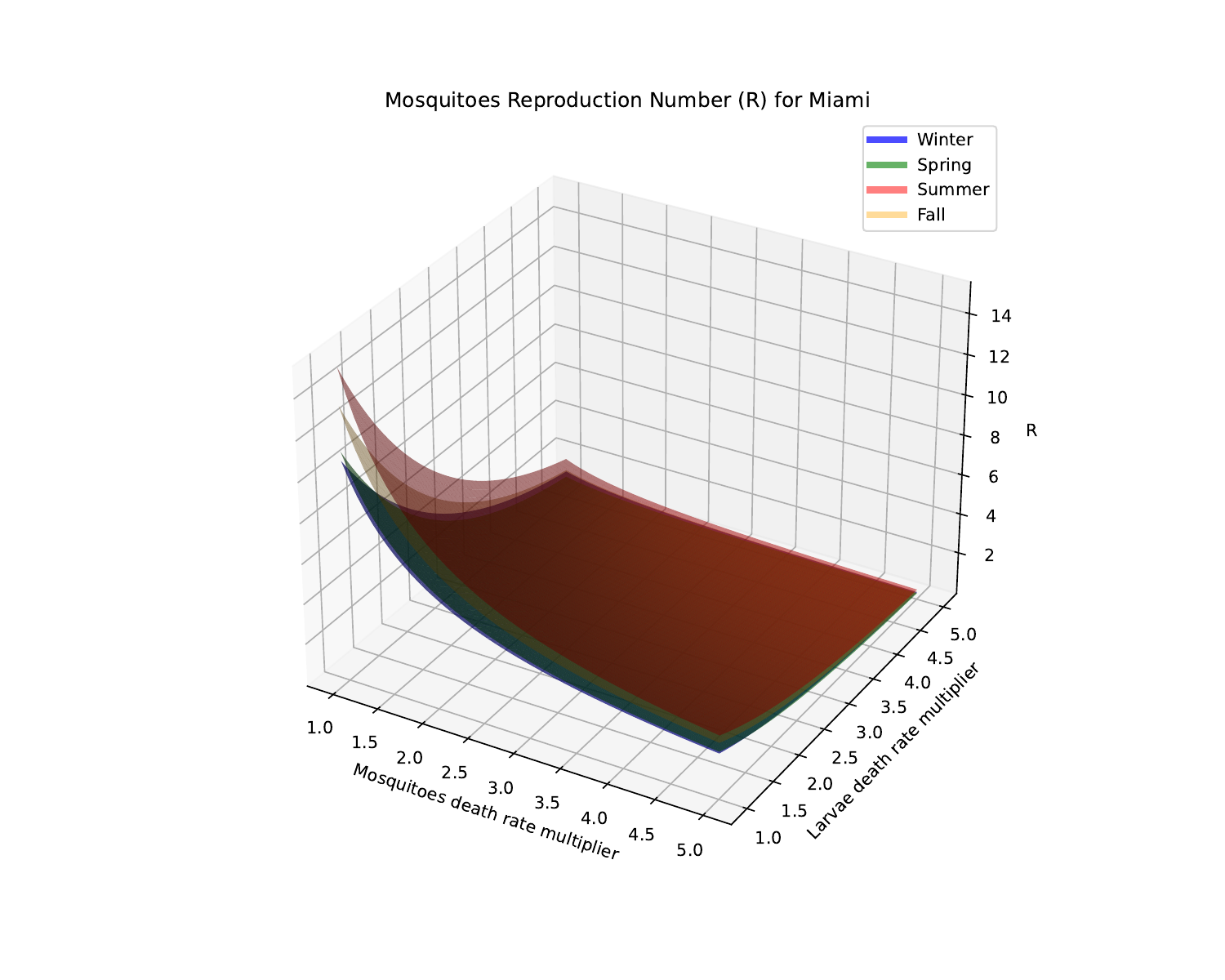}
\caption{$\mathcal{R}_{m}^{MIA}$ vs $(\mu_{m_{\text{mult}}},\ \mu_{l_{\text{mult}}})$}
\label{Rm1}
\end{subfigure}
\hfill  % Creates space between the two subfigures
\begin{subfigure}[b]{0.45\textwidth}  % Reduced width
\centering
\includegraphics[width=\textwidth]{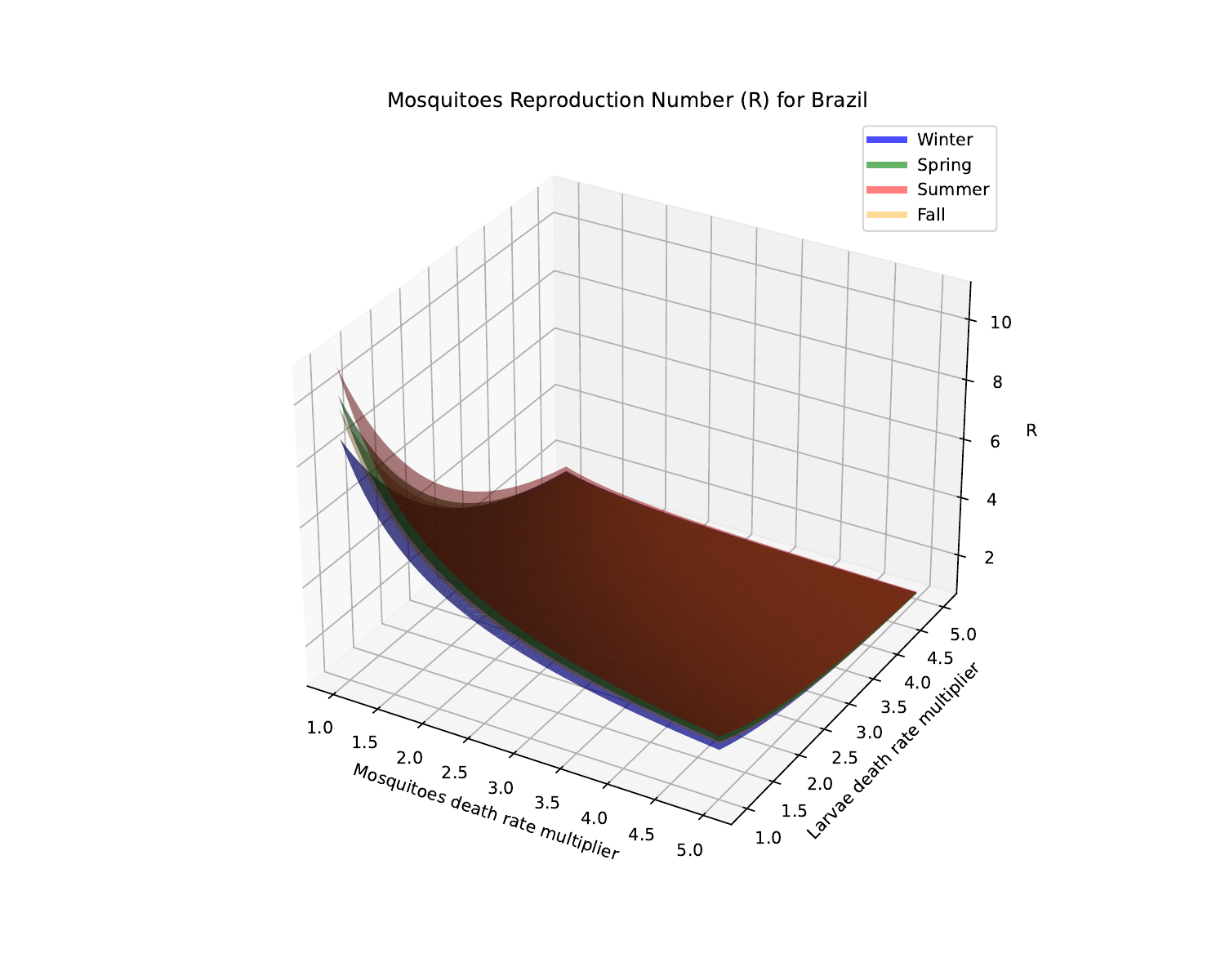}
\caption{$\mathcal{R}_{m}^{BR}$ vs $(\mu_{m_{\text{mult}}},\ \mu_{l_{\text{mult}}})$}
\label{Rm2}
\end{subfigure}
\caption{Sensitive Analysis on the mosquitoes death rate multiplier $(\mu_{m_{\text{mult}}})$ and larvae death rate multiplier $(\mu_{l_{\text{mult}}})$ of $\mathcal{R}_{m}^{MIA}$ and $\mathcal{R}_{m}^{BR}.$ Both Figures (a) and (b) have the plots for each season.}
\label{fig:comparison}
\end{figure}
%\caption{Sensitive Analysis of $R_{m}$ for MIA and BR patches.}
%\label{fig:Rm_comparison}
%\end{figure}

A key observation from Figures \ref{Rm1} and \ref{Rm2} is that the mosquito dynamics in Miami are conducive to a potential outbreak, given the ideal conditions and the number of infected humans. Specifically, we notice that during the summer, the mosquito reproduction number $\mathcal{R}_{m}$ reaches values even higher than those observed in Brazil. Another observation is that in Brazil, as the majority of the endemic regions are located in the Northeast of the country, there is slight variation in $\mathcal{R}_{m}$ across the fall, winter, and spring. Similarly, in Miami the mosquito reproduction number during the spring and fall is quite close, regardless of the values of the control strategy coefficients $(\mu_{m_{\text{mult}}},\ \mu_{l_{\text{mult}}})$.

In Figures \ref{R0zim} and \ref{R0xb_m}, we perform a sensitive analysis of the global basic reproduction number $\mathcal{R}_{0}$ to better understand its behavior under public interventions of vector control strategies and virus spread. It is straightforward to notice that across all seasons, reducing the biting rate in a proportion $\lambda \leq 0.2$ reduces the $\mathcal{R}_{0}$ to values close to one, and eventually less than one. Also, increasing the $\mu_{m_{mult}}, \mu_{l_{mult}}$ in 5 folds individually reduces the $\mathcal{R}_{0}$ to values close to one in all seasons. And increasing both death rate multipliers to values close to 5 eventually leads to $\mathcal{R}_{0} <1.$
\begin{figure}[htp]
\centering
\begin{subfigure}[b]{0.45\textwidth}  % Reduced width
\centering
\includegraphics[width=\textwidth]{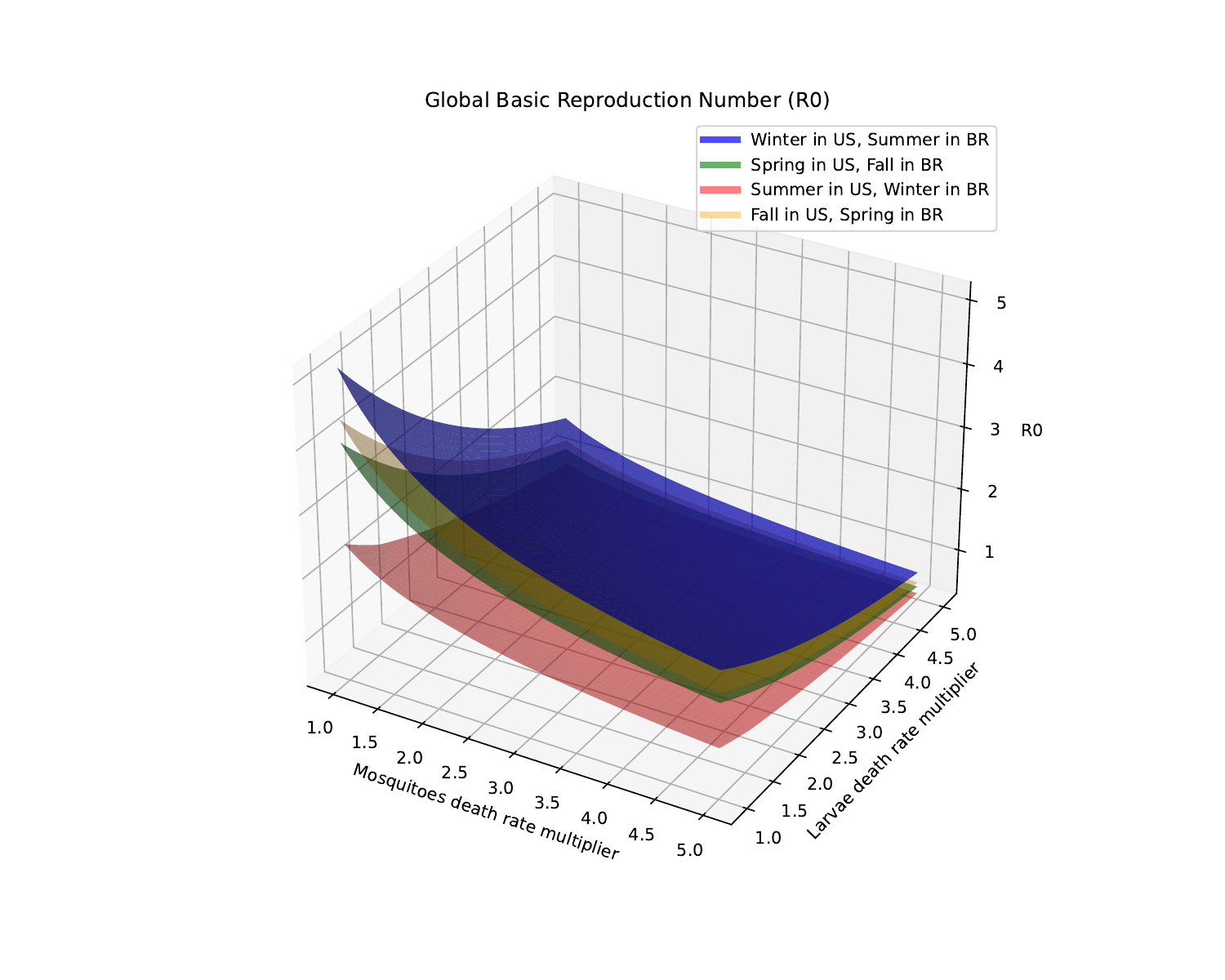}
\caption{$\mathcal{R}_{0}\ vs\ (\mu_{m_{\text{mult}}},\ \mu_{l_{\text{mult}}})$}
\label{R0zim}
\end{subfigure}
%\begin{figure}[htp]
\hfill  % Creates space between the two subfigures
\begin{subfigure}[b]{0.45\textwidth}  % Reduced width
\centering
%\vspace{-0.36cm}
\includegraphics[width=\textwidth]{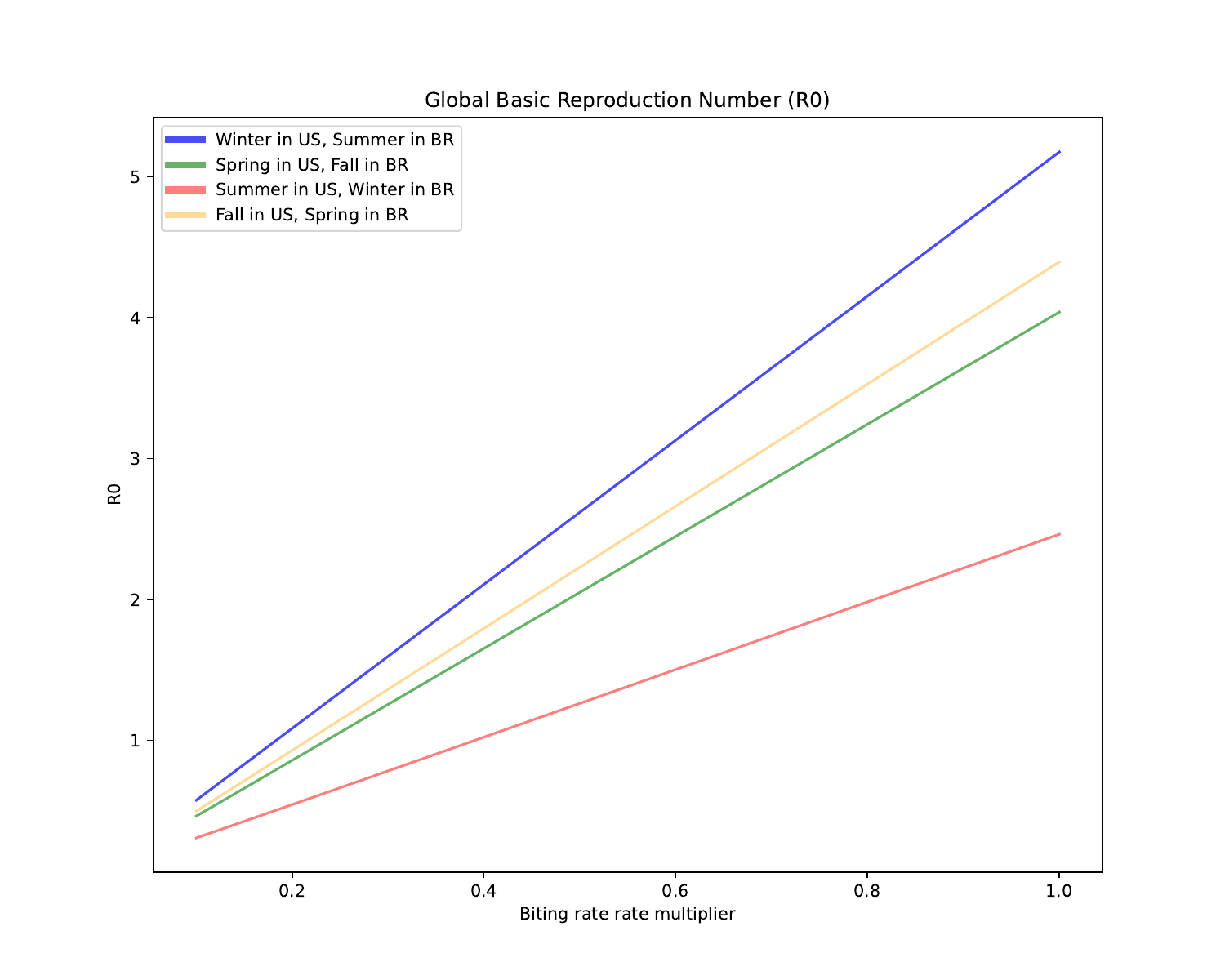}
\caption{$\mathcal{R}_{0}\ vs\ b_{m_{\text{mult}}}$}
\label{R0xb_m}
\end{subfigure}
\caption{In Figure (a) we have $\mathcal{R}_{0}$ in terms of two parameters the mosquitoes death rate multiplier $ (\mu_{m_{mult}})$ and the larvae death rate multiplier $(\mu_{l_{mult}})$. In Figure (b) we have $R_0$ in terms of the biting rate multiplier $b_{m_{mult}}$. Both Figures (a) and (b) have the plots for each season.}
\label{fig:R_0}
\end{figure}

Now, based on Figures \ref{R0zim}, \ref{Rm1} and \ref{Rm2}, consider $\mu_{m_{mult}}=\mu_{l_{mult}}=3$. According to Theorem \ref{principalthm}, the disease will die out, as represented in the Figure \ref{mu=5}. Analogously, in the Figure \ref{bite=.1}, based on \ref{R0xb_m} where $\mathcal{R}_{0}$ decreases linearly as $b_{{m}_{mult}}$ decreases linearly, and considering $b_{{m}_{mult}}=0.1,$ the disease will die out due to Theorem \ref{principalthm} and the fact that the $\mathcal{R}_{0}$ computed in this section is an approximated value of the $\mathcal{R}_{0}$ computed from the original system (\ref{0})--(\ref{0-1}) without the $E_h^i$ and $R_h^i$ equations.
\begin{figure}[htp]
\centering
%\vspace{-0.36cm}
\begin{subfigure}[b]{0.45\textwidth}  % Reduced width
\centering
\includegraphics[width=\textwidth]{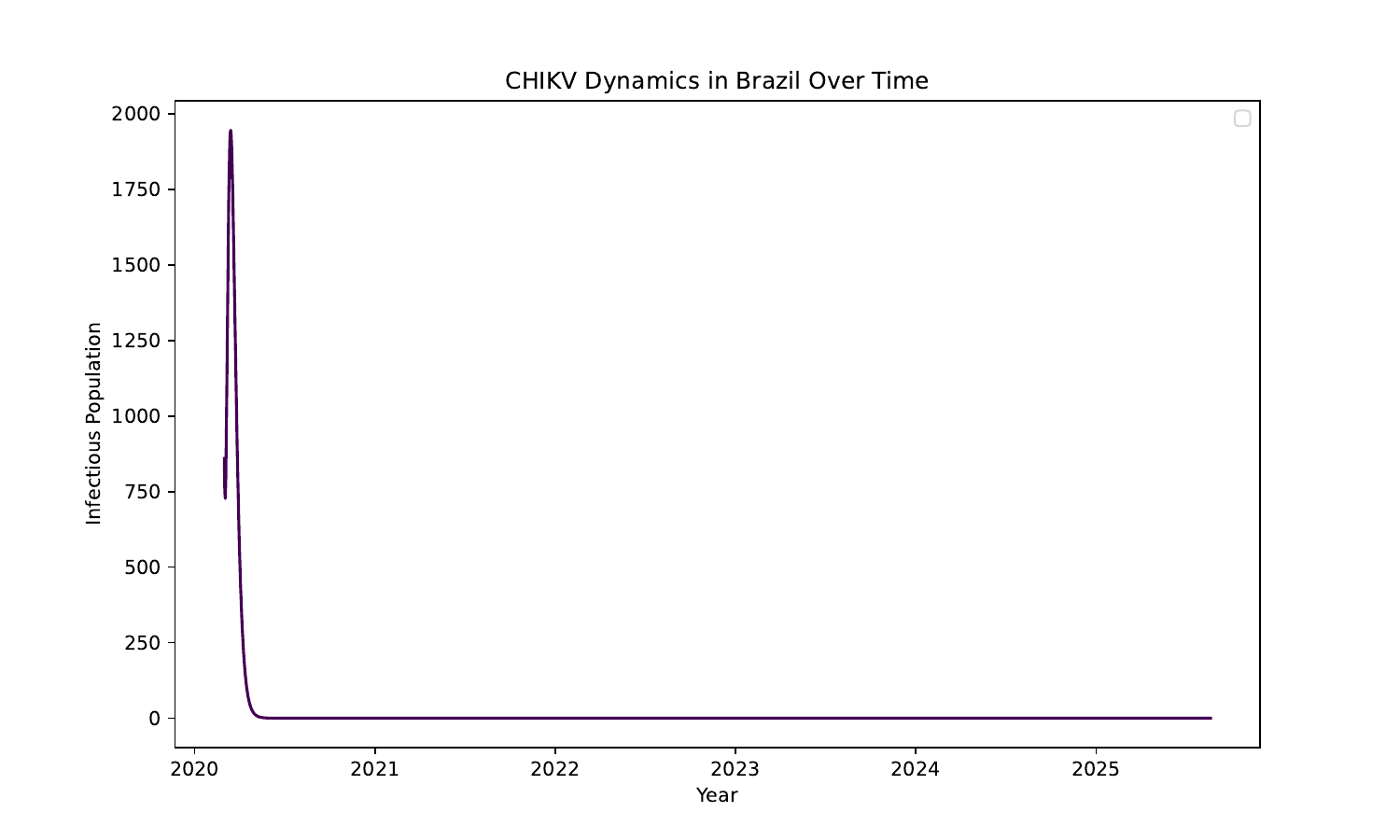}
\caption{Simulation with ${\mu_{l}}_{\text{mult}}={\mu_{m}}_{\text{mult}}=3.$}
\label{mu=5}
\end{subfigure}
\hfill  % Creates space between the two subfigures
\begin{subfigure}[b]{0.45\textwidth}  % Reduced width
\centering
%\begin{figure}[htp]
%\vspace{-0.36cm}
\includegraphics[width=\textwidth]{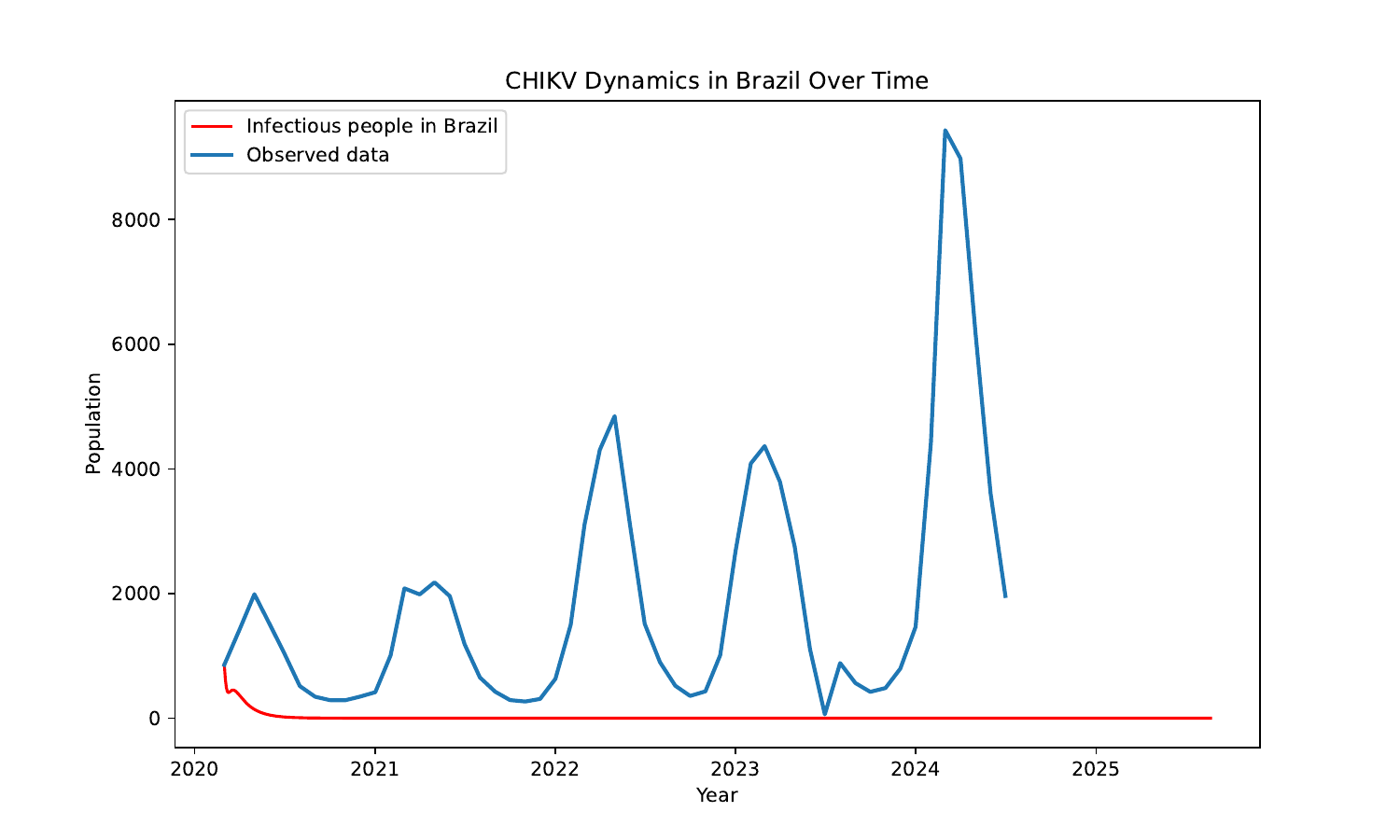}
\caption{Simulation with ${b_{m}}_{\text{mult}}=0.1.$}
\label{bite=.1}
\end{subfigure}
\caption{Simulation of the number of symptomatic infectious human population in Brazil from March 2020 to 2025 in different scenarios. (a) ${\mu_{l}}_{\text{mult}}={\mu_{m}}_{\text{mult}}=3;$ (b) ${b_{m}}_{\text{mult}}=0.1,$ the blue curve represents the data and the red curve is the simulation of the model. } 
\label{fig:simulation}
\end{figure}

%%Moreover, analogously to above, according to Figure \ref{R0zim}, the disease will die out in the following set of death rates multipliers: 
%$$
%\{\{\mu_{{m}_{mult}}\geq 2\}\times  \%{\mu_{{l}_{mult}}=5\}       
%\}\cup \{ \{\mu_{{l}_{mult}}\geq 2.5\}\times  \{\mu_{{m}_{mult}}= 5\} \textcolor{red}{ \} } \cup  \{ \{\mu_{{l}_{mult}}\geq 3\}\times  \{\mu_{{m}_{mult}}\geq 3\}
%\} .
%$$

\section{Discussion}

From a theoretical perspective, this papers analyses the Threshold Dynamics based on the Mosquito and Basic Reproductions Numbers for a periodical time delay multi-patch model with any finite number of patches. Theorem \ref{principalthm} is essentially an Uniform Persistance Theorem \cite{FREEDMAN1995}, which points out to the fact that as time approaches infinity, based on the $R_0$ threshold, the disease either will remain endemic or die out at each patch, and dependent upon the $R_m^i$ threshold, the mosquitoes and larvae population will coexist with human population or simply die out, for each patch, i.e, $i=1,\dots, n$. Surely, if $\mathcal{R}_m^i<1,\ \exists\ 1\leq i \leq n,$ then not only the mosquitoes and larvae population die out, but the disease will disappear at this specific patch as well. On the other hand, if $\mathcal{R}_m^1>1, \ \exists\ 1\leq i \leq n,$ then the mosquitoes, larvae and human population will always coexist, and the disease may die out or not depending solely upon $sign(\mathcal{R}_0-1)$.

From an applied perspective, considering the fact that CHIK has been imported annually to Miami in the last decade from Central and South America including Brazil, where CHIK is endemic, we proposed a multi-patch model to study the geographic spread of CHIKV, incorporating pivotal factors such as human movement, temperature, vertical transmission, and incubation period. We established crucial correlations between the mosquito reproduction number $\mathcal{R}_{m}$ and the basic reproduction number $\mathcal{R}_{0}$, which helps us to better understand the CHIKV transmission dynamics in complex multi-patch environments.%, through Theorem \ref{principalthm}.

Through numerical simulations, validated with real population and temperature data, 
we projected the infectious counts under different scenarios, considering alternative control strategies in Miami-Dade through 2030, and forecast the number of infectious individuals in Brazil through 2025, which allows us to identify optimal control strategies for minimizing the number of infections effectively.

The sensitivity analysis highlights that among vector control interventions, reducing mosquito biting rates is particularly crucial for achieving a disease-free state. The analysis points out to the importance of spatial repellents and targeted sprays, especially when applied before seasonal peaks of infectious cases. Such proactive measures can mitigate disease transmission by curbing contact rates between mosquitoes and humans. Therefore, forecasting when a peak can occur in the future and knowing the right time to apply biting rates control strategies is suggested to be a key point to alleviate the infection dynamics burden.

Additionally, the role of vertical transmission in disease dynamics deserves close attention. If the probability of vertical transmission is time-dependent, it could influence outbreak progression, especially under varying environmental or intervention conditions. Exploring whether this variability is temperature-dependent could provide insights into mosquito population dynamics and their contribution to disease persistence. Furthermore, understanding the difference of the delays in developmental stages between infected and healthy larvae may refine predictive models.

Finally, the movement of humans between patches proves to be a critical factor in disease spread, potentially triggering outbreaks in previously disease-free regions, particularly in Miami-Dade, as the model suggests. This mobility, combined with climatic factors, emphasize the need for geographically coordinated public health strategies, by integrating temperature-adaptive measures that shed light on a forecast of the particular time for public interventions to be adopted, as shown in Figure \ref{proj_Br}, aiming to minimize the number of infections due to mosquito-born diseases, more generally.

%Bibliography
\bibliographystyle{unsrt}  
\bibliography{references}

\end{document}